\documentclass[%
reprint,
amsmath,amssymb,superscriptaddress,
aps,onecolumn,nofootinbib
]{revtex4-2}
\usepackage{natbib}
\usepackage{graphicx}
\usepackage{dcolumn}
\usepackage{bm}
\allowdisplaybreaks
\usepackage[utf8]{inputenc}
\usepackage{hyperref}
\usepackage{amsmath,amsfonts,amsthm,float}
\usepackage{fancybox}
\usepackage[ruled,linesnumbered]{algorithm2e}
\usepackage{bm}
\usepackage{braket} 
\usepackage{graphicx}
\usepackage{tabularx}

\usepackage[normalem]{ulem}
\usepackage{tikz}

\let\isout\sout \renewcommand{\sout}[1]{\ifmmode\text{\isout{\ensuremath{#1}}}\else\isout{#1}\fi}

\definecolor{darkgreen}{HTML}{006400}

\newcommand{\fe}[1]{\textcolor{black}{#1}}

\newcommand{\tsout}[1]{}

\newtheorem{theorem}{Theorem}

\newtheorem{observation}[theorem]{Observation}

\newcommand \QAOA {{\sc QAOA}}
\newcommand \QAOAw {{\sc QAOA-warm}}
\newcommand \mc {{\sc Max-Cut}}

\newcommand{\bx}{\mathbf{x}}

\AtBeginDocument{%
  \providecommand\BibTeX{{%
    \normalfont B\kern-0.5em{\scshape i\kern-0.25em b}\kern-0.8em\TeX}}}

\begin{document}

\title{Bridging Classical and Quantum with SDP initialized warm-starts for QAOA}

\author{Reuben Tate}
\author{Majid Farhadi}
\affiliation{Georgia Institute of Technology, Atlanta, GA 30332, USA}

\author{Creston Herold}
\author{Greg Mohler}
\affiliation{Georgia Tech Research Institute, Atlanta, GA 30332, USA}

\author{Swati Gupta}
\email{Corresponding Author: swatig@gatech.edu}
\affiliation{Georgia Institute of Technology, Atlanta, GA 30332, USA}

\date{\today}

\begin{abstract}
We study the Quantum Approximate Optimization Algorithm (\QAOA) in the context of the \mc\ problem. Noisy quantum devices are only able to accurately execute \QAOA\ at low circuit depths, while classically-challenging problem instances may call for a relatively high circuit-depth. This is due to the need to build correlations between reachable pairs of vertices in potentially large graphs \cite{FGG20}. To enhance the solving power of low-depth QAOA, we introduce a classical pre-processing step that initializes \QAOA\ with a biased superposition of possible cuts in the graph, referred to as a {\it warm-start}. In particular, we initialize \QAOA\ with a solution to a low-rank semidefinite programming relaxation of the \mc\ problem. Our experimental results show that this variant of \QAOA, called \QAOAw, is able to outperform standard \QAOA\ on lower circuit depths in solution quality and training time. \fe{While this improvement is partly due to the classical warm-start, we find strong evidence of further improvement using QAOA circuit at small depth.} We provide experimental evidence of improved performance as well as theoretical properties of the proposed framework.

\end{abstract}

\maketitle

\section{Introduction}
There is growing interest in utilizing near-term quantum technology \cite{P18} to solve challenging problems in combinatorial optimization. Farhi et al.\ \cite{FGG14} recently introduced the Quantum Approximate Optimization Algorithm (\QAOA), designed specifically for combinatorial optimization problems. This is a hybrid quantum-classical algorithm, where the state of a quantum processor is controlled by  variational parameters $\gamma$ and $\beta$, which are optimized using a classical processor.

We consider the \mc\ problem, which is one of the most studied problems in combinatorial optimization. Given a simple weighted graph $G = (V,E)$ with vertex set $V = [n]$, edge set $E \subseteq \binom{V}{2}$, and weights $w: E \rightarrow \mathbb{R}$, the \mc\ problem is to find a partition of $V$ into two sets $S, V \setminus S \subseteq V$ such that the total weight of the edges that are cut by this partitioning, i.e., $\text{cut}(S) := \sum_{e \in E} w_e \cdot \mathbf{1}[e \in S \times (V \setminus S)]$, is maximized. The \tsout{state-of-the-art} \fe{best-possible (under Unique Games Conjecture\footnote{\fe{Under the Unique Games Conjecture (UGC), this 0.878-approximation ratio is the best \fe{factor we can hope to achieve, which simply means that assuming UGC there does not exist a polynomial-time classical algorithm for \mc{} with a $(0.878+\epsilon)$-approximation ratio, for any $\epsilon > 0$,}}  \cite{Khot02,MOO05,KKMO07}.}) polynomial-time}  approximation factor for solving the \mc\ problem is $0.878$ (for graphs with non-negative edge weights), and is given by the seminal Goemans-Williamson (GW) \cite{GW95} algorithm.  This method creates a random partition of the vertex set using the solution to a convex relaxation of \fe{\mc} using a  semi-definite program (SDP). \mc\ is not only NP-hard to compute but also is hard to classically approximate better than a multiplicative factor of $16/17$ for non-negative edge weights \cite{TSSW00,H01}. Graphs with both positive and negative edge weights seem harder to find approximate solutions for \tsout{: if the sum of the edge weights are nonnegative, then one can recast \mc\ as a MaxQP problem for which there exists a polynomial-time algorithm giving a  $O(1/\log n)$ approximation} \fe{(e.g., see \cite{CW04} for a related problem)}. Given a perfect noise-free quantum computer, on the other hand, the \QAOA\ algorithm is able to converge to the optimal solution as the number of QAOA stages $p$ increases. This is due to \QAOA's asymptotic equivalence to the Quantum Adiabatic Algorithm (QAA) as $p\to\infty$ \cite{FGG14}. The caveat is that increasing $p$ not only increases the number of parameters to be optimized, but makes the circuit more susceptible to quantum noise. Sufficiently deep circuits are effectively inaccessible due to the practical limitations of current and near-term quantum hardware. 

In this work, we study the impact of low-rank local optima for \tsout{Max-Cut} \fe{\mc{}} relaxations as initializations for \QAOA. We refer to these non-standard initializations of the \QAOA\ algorithm as {\it warm starts} (following the classical optimization literature \cite{bertsimas2016best, marcucci2020warm, ralphs2006duality}) and explore their impact on the performance of the hybrid variational method. Our warm-states are separable, and are based on a local optima of Burer-Monteiro's low-rank relaxation of \mc\ on a given graph \cite{BM03,BM05}. We use standard \QAOA\ mixers with this initialization and refer to this variant as \QAOAw. A key result of our study is that, our numerical simulations show \QAOAw\ typically outperforms standard \QAOA\ in quality of solution for low $p$-depth. In particular, we perform numerical simulations on \tsout{1148}\fe{1264} graph instances of up to \tsout{11}\fe{12} nodes. We find that \QAOAw\ achieves a higher average approximation ratio than standard \QAOA\ for \fe{96.8\%, 90.0\%, 72.8\%, and 53.6\%} of instances for circuit depths $p=1,2,4,8$ respectively. \fe{While this improvement is partly due to the classical warm-start, we find that the improvement due to the QAOA circuit on the warm-start is significant, e.g., at least 50\% improvement in the approximation ratio at $p=1$ compared to $p=0$ (warm-start initialization) on 74 instances, and at least 80\% improvement on 22 instances at depth $p=1$.}  We also explore the variational parameter space with and without initializations using warm-starts, and show interesting theoretical properties for warm-starts. For \QAOAw{} at $p=1$, our numerical  simulations indicate that the energy landscape frequently has a more ridge-like structure which can potentially be exploited in regards to optimization of the variational parameters. Additionally, with \QAOAw{}, our simulations show an overall improvement in the expected cut values across the landscape. 

For graphs with non-negative edge weights, it is known that standard \QAOA\ starts with a $0.5$-approximation guarantee at $p = 0$ \tsout{(i.e.,only quantum sampling is performed) due to the structure of the initial state $\ket{+}^{\otimes n}$.}\fe{(as measuring the initial state $\ket{+}^{\otimes n}$ yields the same result as if uniform classical sampling all possible cuts).} In contrast, given any graph with non-negative edge weights and a $\kappa$-approximate solution to the Burer-Monteiro relaxation of the problem in $2$ or $3$ dimensions, we prove our pre-processing stage allows \QAOAw\ to guarantee at least a $0.75\kappa$ or $0.66\kappa$ approximation (respectively) at any depth, in particular $p = 0$. This general bound augments the current literature where provable guarantees for standard \QAOA\ at low depth are only known for special cases, e.g., regular graphs. For $n$-node even cycles, depth-$p$ QAOA achieves an approximation ratio of at most $\frac{2p+1}{2p+2}$ whenever $n > 2p$ \cite{MFS19}. It is conjectured that the approximation ratio is exactly $\frac{2p+1}{2p+2}$ \cite{FGG14}. \fe{On the other hand}, for these even cycles, \fe{the warm-starts simply result in an optimal cut} due to the antipodal structure of locally optimal solutions to the utilized Burer-Monteiro relaxation. 

In order to give a fair comparison with \QAOA, we explore the limitations of \QAOAw\ at higher circuit depths. We prove that \QAOAw\ does not guarantee convergence to the optimal solution as $p \to \infty$ for certain 3-dimensional initializations (see Theorem \ref{thm:warmstart}). This is to be expected since standard \QAOA's equivalence to the Quantum Adiabatic Algorithm is dependent on the fact that the initial state $\ket{+}^{\otimes n}$ is a maximum energy eigenstate of the mixing Hamiltonian, but this will not typically be the case for the initial state of \QAOAw. In other words, for any graph $G$, there exists a graph-dependent depth $p$ such that standard \QAOA\ does at least as well as \QAOAw. 

\paragraph{Related Work.} There have been many different approaches to improving QAOA. Zhou et al.~\cite{ZWCPL19} proposed the INTERP and FOURIER heuristics to improve parameter optimization. These approaches bootstrap QAOA parameter initialization to the QAOA solver itself, and do not use any classical-side optimization. Zhu et al.~\cite{ZTBCVMBE20} introduced layer-dependent mixer operators that rely on an ansatz for the QAOA states. Sack and Serbyn~\cite{SS21} meanwhile focused on QAOA parameter optimization by connecting QAOA more closely to its quantum adiabatic origins.  Our approach meanwhile leverages the considerable body of work on classical solvers. B\"artschi et al. \cite{BE20} altered the mixing term to use a Grover-like circuit. However, their approach is not well suited to \mc\ as it does not have a subspace of preferred states. Unreachable states that are independent of initial conditions were explored in ~\cite{APMD19}, and a contrast was drawn between these states and the barren plateau problem, where poor initialization results in inefficient optimization. Our work connects the two cases, finding cases where initial states fail to mix properly and yield low-value approximate solutions. In a recent parallel study by Egger et al. \cite{egger2020warm}, the authors explore two warm-start techniques. In the former, they perform a continuous relaxation of variables for a Quadratic Unconstrained Binary Optimization\footnote{QUBO is in fact equivalent to \mc \cite{dunning2018}, and therefore all our results apply to QUBO as well.}  (QUBO) and modify the mixer in a way that ensures one achieves optimality as the circuit depth tends to infinity. In the latter, they initialize QAOA based on a \emph{single} cut that is classically obtained from a \mc\ instance, and then modify the mixer so that the value of that specific cut can be recovered at depth-1 QAOA. Our approach, on the other hand, is to use low-rank local optima for relaxations to \tsout{Max-Cut} \fe{\mc{}} (with rank greater than 1). Additionally, Egger et al. \cite{egger2020warm} modify the mixer so that the warm start is the lowest energy eigenstate, while we maintain the standard mixer in this work. Overall, since our approach allows more flexibility in the initialization of warm-starts, it ultimately translates to  improvements in performance, especially at low-circuit depths (as discussed in Appendix \ref{sec:appendixEgger}).

%\rte{For example, we find that $p=$ \rt{INSERT} is large enough for standard \QAOA\ to do at least as well \QAOAw\ on \rt{X\%} of the \rt{$n_1$} node graphs; when this same circuit depth is used for \rt{$n_2$} node graphs, standard \QAOA\ outperforms \QAOAw\ only \rt{Y\%} of the time.}

\paragraph{Outline.} We believe our study draws interesting connections between classical and quantum hybrid algorithms while positively impacting the performance of the \QAOA\ algorithm. In the rest of this paper, we review \QAOA\ in Section \ref{sec:QAOA} and the Goemans-Williamson algorithm as well as the low-rank Burer-Monteiro formulation for it in Section \ref{subsec:classicalAlgoPrelims}, we introduce our key ideas as a preprocessing stage in Section \ref{sec:pipeline}, present our computational and theoretical results in Sections \ref{sec:computational} and \ref{sec:theory} respectively, and conclude with a discussion and open questions in Section \ref{sec:discussion}.

\section{Quantum and Classical Optimization Algorithms}
\label{sec:prelims}
\fe{Before delving into the relevant algorithms in quantum and classical settings, we first define the notion of approximation ratio (AR) for \mc\ in general weighted graphs. Given a candidate solution $S' \subseteq V$ for a graph $G$, call $S'$ an $\alpha$-approximate solution if $\frac{\text{cut}(S') - \text{{\sc Min-Cut(G)}}}{\text{\mc}(G) -\text{{\sc Min-Cut(G)}}} \geq \alpha$ where $\text{\mc}(G) = \max_{S \subseteq V} \text{cut}(S)$ and $\text{{\sc Min-Cut}}(G) = \min_{S \subseteq V} \text{cut}(S)$ . Defined this way, the approximation ratio $\alpha$ will always lie in the interval $[0,1]$. An $\alpha$-approximate algorithm is defined to be an algorithm that always returns an $\alpha$-approximate solution (in expectation). Note that when $G$ only contains non-negative edge weights, then $\text{{\sc Min-Cut}}(G) = 0$, in which case, being an $\alpha$-approximate solution is equivalent to being at least an $\alpha$-fraction of the optimal solution.}

\subsection{The Quantum Approximate Optimization Algorithm}\label{sec:QAOA}

In this section, we review the hybrid quantum-classical algorithm of \QAOA\ for the \mc\ problem. \QAOA\ assigns a quantum spin to every binary output variable. In each of the $p$ layers of the algorithm, the problem Hamiltonian $H_C$ and a mixing Hamiltonian $H_B = \sum_{i \in [n]} \sigma_i^x$, where $\sigma_i^k$ is a Pauli matrix for qubit $i$ with $k=x,y,z$, are alternately applied to the initial quantum processor state $\ket{s_0}$, generating a variational wavefunction 
$$
\ket{\psi_p(\gamma, \beta)} \! = \! e^{-i\beta_p H_B} e^{-i\gamma_p H_C} \! \cdots e^{-i\beta_1 H_B} e^{-i\gamma_1 H_C} \! \ket{ s_0 },
$$
where $\ket{ s_0 } = \ket{ + }^{\otimes n}$ is the standard initial state. Sampling from the final variational state will yield a cut with an expected cut value of:
$$
F_p(\gamma,\beta) = \bra{ \psi_p(\gamma, \beta) } H_C \ket{\psi_p(\gamma, \beta)}.
$$

For the maximization problem \mc, the cost Hamiltonian for a graph $G=(V,E)$ (with weights $w: E\to \fe{\mathbb{R}}$) can be written as
$$
H_C = \frac{1}{2} \sum_{(i,j)\in E} w_{ij} (1 - \sigma_i^z \sigma_j^z)\,.
$$
%where $A$ is the  adjacency matrix of the input graph.\footnote{Since we focus on  undirected graphs in this work, we have $A_{i,j} = A_{j,i} = w_{i,j}$ for all $i,j \in V$ where $w_{i,j}$ is the weight of the edge connecting vertex $i$ and vertex $j$.}

The (near) optimal parameters of the algorithm, $\gamma, \beta$, are found by a classical algorithm to maximize the performance of the \QAOA\ algorithm, with $F_p(\gamma,\beta)$ viewed as a multi-dimensional non-convex function. We let $M_p$ denote the expected cut value with optimal choice of $\gamma,\beta$ parameters, i.e.,
$$M_p = \max_{\gamma,\beta} F_p(\gamma,\beta).$$ 
It is not difficult to see that $M_p$ is a non-decreasing function in $p$; moreover, as previously mentioned, $M_p \to \text{\mc}(G)$ as $p \to \infty$ \cite{FGG20}. \fe{For graphs with non-negative edge-weights, the ratio $M_p/\text{\mc}(G) \geq 0.5$ for all $p \geq 0$ \fe{due to the 0.5-approximation ratio achieved at $p=0$ for the standard intialization.}} \tsout{When the circuit has depth $p = 0$, the algorithm is equivalent to a quantum sampling of the initialization state $\ket{s_0}$. For the original \QAOA\ with $\ket{s_0} = \ket{+}^{\otimes n}$, at $p=0$ every edge is cut with probability $0.5$, which gives a $0.5$ approximation (in expectation) for \mc\ on graphs with non-negative edge weights; alternatively, one can interpret the quantum sampling of the $\ket{+}^{\otimes n}$ state as uniformly at random choosing one of the $2^n$ different possible cuts in the graph.}

To find the optimal variational parameters, one can simply perform a dense grid search for $\gamma, \beta \in [-\pi,\pi]^{2p}$, but this would be feasible only for small circuit depths. For scalability, one can instead treat $F_p(\gamma,\beta)$ as a black-box\footnote{For actual quantum devices, the value of $F_p(\gamma,\beta)$ and its gradients can be estimated by taking multiples measurements of $\psi_p(\gamma,\beta)$ in the computational basis.} and utilize a classical optimizer to (iteratively) update and find suitable values of $\gamma$ and $\beta$ in an effort to the maximize the expected cut value. For any classical optimization algorithm $\mathcal{A}$, it will eventually terminate at some $(\gamma,\beta) = (\hat{\gamma},\hat{\beta})$; we will then let $\text{QAOA}(G ; \mathcal{A})$ denote the expected value of the cut at $(\hat{\gamma},\hat{\beta})$, i.e.,
$$\text{QAOA}(G ; \mathcal{A}) = F_p(\hat{\gamma},\hat{\beta}).$$

%For scalability, we use classical optimization methods to find approximately optimal $\gamma,\beta$ values, given any finite circuit depth $p$. We refer to the process of optimizing the variational parameters as the \emph{training loop} and we provide a summary of our implementation of this optimization in Algorithm \ref{alg:trainingLoop}, where we use the {\sc Adam} algorithm that is integrated into the {\sc Tensorflow} library \cite{KB14}.

To optimize the variational parameters, we consider four choices of the optimizer: ADAM ~\cite{adamRef}, COBYLA ~\cite{cobylaRef},  Nelder-Mead ~\cite{neldermeadRef},  and BFGS ~\cite{bfgsRef}. Since $F_p(\gamma,\beta)$ is non-convex, classical optimizers are not guaranteed to stop at a globally optimal choice of $\gamma$ and $\beta$, i.e., the expected result of \QAOA\ will not always be equal to $M_p$ (i.e. the expected result of \QAOA\ had we initialized $\gamma$ and $\beta$ optimally). ADAM and BFGS operate with the first-order information (i.e., using gradient estimates\footnotetext{One can calculate (or approximate) the gradient using a variety of methods. Our implementation approximates the gradient using an analytic forward difference method implemented in {\sc Tensorflow Quantum} (with default parameters \texttt{error\_order=1} and \texttt{ grid\_spacing=0.001}). By \emph{analytic}, we mean that any expectations computed in the calculation are computed \emph{exactly} (instead of using a sampling-based approximation).}), whereas COBYLA and Nelder-Mead operate with the zeroth-order information (i.e., function value estimates). On quantum devices, gradients are estimated using multiple evaluations of the function $F_p(\gamma,\beta)$ at various $(\gamma,\beta)$; these function evaluations are noisy since $F_p(\gamma,\beta)$ itself is estimated by taking an average of multiple quantum measurements. For this reason, gradient-free optimizers are typically more robust against quantum noise and are recommended in practice over gradient-based methods \cite{LTMIJ21}. Application of machine learning techniques for optimizing the variational parameters (a technique known as meta-learning) has also shown promise in the noisy quantum setting \cite{WSR21}. Recent results regarding the concentration of the (standard) QAOA landscape can also be used to speed up optimization of the variational parameters \cite{BBFGN}. %Therefore, the number of calls to Adam and Nelder-Mead are typically smaller than the calls to COBYLA and BFGS.
%However, compared to the improvements in approximation ratio due to warm-starts,
 Even though the runtimes for various optimizers may significantly differ, we find that the choice of the optimizer has much smaller impact on the approximation ratio achieved for \QAOAw. (discussed in Section \ref{sec:computational}). %\sg{this is still not correct. We need to use the refs that Greg provided.}%\rt{Removed part involving number of iterations being smaller for gradient-based methods since that was not something I consistently observed.}
%Throughout this paper, when discussing ``running \QAOA'' or simply ``\QAOA,'' it is to be assumed that we are running Algorithm \ref{alg:trainingLoop} and utilizing the choice of $\gamma_{1, \hdots, p},\beta_{1, \hdots, p}$ found in the last training epoch.\footnote{In the context of Algorithm \ref{alg:trainingLoop}, an epoch is simply one iteration of the \textbf{repeat} loop.}

% \begin{algorithm}[t] 
%         \DontPrintSemicolon
% 		\KwIn{$G = (V,E),w:E \to \mathbb{R}^{\geq 0}, p,\ket{s_0},\gamma^{(0)}_{1,\hdots,p}=\mathbf{0},\beta^{(0)}_{1,\hdots,p}=\mathbf{0}$}
% 		Set $t:=0$.\;
% 		Set hyperparameters of {\sc Adam} algorithm \cite{KB14}: stepsize $\alpha = 0.001$, 1st exponential decay rates $\delta_1 = 0.9$,  2nd exponential decay rate $\delta_2 = 0.999$, numerical stability constant $\epsilon = 10^{-7}$.\;
% 		\Repeat{no improvement in $F_p(\gamma, \beta)$ by more than $10^{-6}\sum_{e \in E} |w_e|$ within 10 epochs}
% 		{
%     		% footnotes cannot directly be placed in the algorithm environment, so we have to place \footnotemark inside the environment and then later place \footnotetext outside of the environment
%     		Calculate\footnotemark   \ gradient: $g^{(t+1)} = -\nabla F_p(\gamma^{(t)},\beta^{(t)})$,\;
%     		Perform one step of the {\sc Adam} algorithm with $g^{(t+1)}$ to obtain $ \gamma^{(t+1)},\beta^{(t+1)}$,\;
%     		$t:=t+1$.\;
% 		}
% 	\caption{\QAOA\ Hybrid Training Loop}
% 	\label{alg:trainingLoop}
% \end{algorithm}

\subsection{Classical Optimization Algorithms}
\label{subsec:classicalAlgoPrelims}
In this section we review \tsout{the state-of-the-art}\fe{two} classical approximation algorithms for \mc. Recall that given a (weighted) graph $G=(V,E)$ with weights $w:E\to \mathbb{R}$, the \mc\ problem is to find a partitioning of the vertices into two subsets, $S$ and $T = V \setminus S$, that maximizes the number of cut edges, i.e., 
$$
\text{\mc}(G) = \frac{1}{2}\max_{S \subseteq V} \sum_{i,j \in V} \fe{w_{i,j}} \mathbf{1}[i \in S] \mathbf{1}[j \notin S].$$

\fe{Note that if  $(i,j) \notin E$, we can just take $w_{ij}$ to be zero. Instead of maximizing over subsets of $V$, one can rewrite the problem as maximizing over $\{-1,1\}^{|V|}$ instead. To do this, we associate} every vertex $i \in V$ \fe{with a decision variable $x_i$, where $x_i = +1$ indicates that vertex} $i \in S$ and $\fe{x_i =}-1$ \fe{indicates that $i \notin S$. Observe that for an edge $(i,j) \in E$, we have that the edge is cut if and only if $x_i \neq x_j$. Using the above fact, one can easily check that:}
 \fe{\begin{equation}
     \frac{1}{4}w_{i,j}(x_i - x_j)^2 = \begin{cases} w_{ij},& (i,j) \text{ is cut},\\
     0,& (i,j) \text{ is not cut}.\end{cases}
 \end{equation}
}

 \fe{By adding up the contribution of each edge and letting $n = |V|$, it becomes clear that one can reformulate the \mc{} problem as the following maximization problem:}
\begin{align}
	\max_{x \in \set{\pm 1}^n} \text{cut}(x) =&
	\max_{x \in \set{\pm 1}^n} \frac{1}{\fe{4}} \sum_{\fe{i<j}} w_{i,j} (x_i - x_j)^2 \label{eq:mc-pm1}\\
	=& \max_{x \in \set{\pm 1}^n} \frac{1}{4} \sum_{i,j}\fe{w_{i,j}} (1 - x_i x_j) \nonumber\\
	=& ~\frac{1}{2} {W} + \max_{x \in \set{\pm 1}^n} \frac{1}{4} \fe{\langle -A, x x^T \rangle}, \label{eq:mc-matrix}
\end{align}
where \fe{$A$ is the adjacency matrix of $G$}, $\langle \cdot, \cdot \rangle$ denotes the \emph{Frobenius} product of two matrices\footnote{Not to be confused with the \emph{bra-ket} notation, the \emph{Frobenius} product of two same-sized matrices $A$ and $B$, denoted by $\langle A,B \rangle$, is equal to $\text{Tr}(A^\dagger B)$ where $\text{Tr}(\cdot)$ denotes the \emph{trace} of a matrix and $(\cdot)^\dagger$ denotes conjugate transposition.}, and $W = \sum_{(i,j) \in E} w_{ij}$.

\paragraph{Goemans-Williamson (GW) Algorithm.} In the seminal work of Goemans and Williamson \cite{GW95} in 1995, the authors pioneered the use of semi-definite programs for solving combinatorial problems. Considering $Y = x x^T \succcurlyeq 0$ from equation \eqref{eq:mc-matrix}, \mc\ becomes equivalent to maximizing $\langle -A, Y \rangle$ by matrix $Y$ from the positive semidefinite cone, subject to having a unit diagonal, in addition to being rank-$1$.\footnote{We use ``$A \succcurlyeq 0$'' to mean that $A$ is a positive semidefinite matrix, i.e., $A$ is a symmetric matrix with real, nonnegative eigenvalues.} Relaxing the last constraint gives us a semidefinite program as follows:
\begin{equation}
\label{eqn:SDPRelaxation}
	\begin{array}{lll@{}ll}
		&\text{maximize}  &  & \langle -A, Y \rangle &\\
		&\text{subject to}&
		& \langle Y, e_i e_i^T \rangle = 1,  & \forall i \in [n],\\
		& &   & Y \in \mathbb{S}^n_+\,, &
	\end{array}
\end{equation}
where $n = |V|$ and $\mathbb{S}^n_+$ is the set of all $n \times n$ positive semidefinite matrices. The value given by the relaxation above was first considered in 1993 by Delorme and Poljak \cite{DP93} in the form of an eigenvalue maximization problem with the equivalence shown shortly after by Poljak and Rendl \cite{PR95} in 1995. The above relaxation is in the form of a semi-definite program and hence since it is a convex program it can be solved in polynomial time up to arbitrary precision, e.g., by using interior point methods \cite{NN94}.

For a Cholesky decomposition of $Y = X^T X$ (with $X \in \mathbb{R}^{n \times n}$), one can think of the solution to the above SDP as an embedding which maps vertex $i$ to $X_{:i}$, i.e., the $i$th column of $X$. This embedding can be viewed as a maximizer of a relaxation of equation (\ref{eq:mc-pm1}) where $x_i$ still has unit distance from the origin, but now in $\mathbb{R}^n$, i.e., $x_i$ lies on the $(n-1)$-sphere.\footnote{The $k$-sphere, denoted $S^k$, is defined as $S^k = \{x \in \mathbb{R}^{k+1} : \Vert x \Vert = 1\}$.} To map this high dimensional solution to a cut in the graph, the GW algorithm considers a random hyperplane through the origin to partition the vertices into two sets according to which side of the hyperplane they lie on; Goemans and Williamson \cite{GW95} showed that this choice of rounding yields a 0.878-approximation to \mc\ in expectation, when the edge weights are non-negative.

%Consider endpoints of an edge $(i,j)$ and the corresponding points $x_i$ and $x_j$ in the SDP solution, and let $\theta$ be the angle between these two points. By definition, $1 - x_i \cdot x_j = 1-\cos \theta$ so the contribution of this edge to the SDP objective is simply $(1-\cos \theta)/2$. On the other hand, the expected contribution of this edge in the rounded solution is equal to probability that these two vertices are cut by the random hyperplane, i.e., $\theta/\pi$. Thus, the expected approximation ratio is at least $2\theta / (\pi - \pi \cos \theta)$. The minimum value for this ratio is thus a lower bound on the approximation factor of the algorithm, i.e., GW is an $\omega$-approximate algorithm where $\omega = \min_{\theta \in [0, \pi]} 2\theta / (\pi - \pi \cos \theta) > 0.878$.

\paragraph{Burer-Monteiro (BM) Method.}
Observe that changing variables as $Y = X^TX$ (with $X \in \mathbb{R}^{n \times n}$), one can eliminate the positive semi-definite constraint in (\ref{eqn:SDPRelaxation}) and obtain the following equivalent reformulation:
\begin{align}
	\text{maximize} & \quad \langle -A, X^TX  \rangle  \nonumber\\
	\text{subject to} &\quad  \Vert x_i \Vert_2 = 1,&& \hspace{-3.5cm} \forall i \in [n], \label{eq:BMMC} \\
	&\quad x_i \in \mathbb{R}^{n},&& \hspace{-3.5cm} \forall i \in [n]\, \label{eq:rank-k}, 
\end{align}
where $x_i$ denotes the $i$th column of $X$. Burer and Monteiro \cite{BM03} proposed relaxing $x_i$ for each vertex to $\mathbb{R}^k$ instead of $\mathbb{R}^n$ in \eqref{eq:rank-k}, i.e., use $x_i \in \mathbb{R}^k$. Unlike the relaxation used in the Goemans-Williamson relaxation, this modification yields a non-convex optimization problem. We refer to this modification as the rank-$k$ Burer-Monteiro \tsout{Max-Cut} \fe{\mc{}} (BM-MC$_k$) relaxation.

Not only is optimizing a non-convex (non-concave) optimization problem difficult, but even finding a local optimum to a non-convex optimization problem can be challenging due to saddle-points. Nevertheless, first and second-order optimization methods have showed promising performance in converging to high quality local optima for low-rank BM formulation of \mc\ (and many other combinatorial optimization problems). Burer and Monteiro invented this heuristic method, motivated by \emph{existence} of a low rank optimal solutions to the original ($n$ dimensional) SDP whenever $\binom{k}{2}$ is no less than the number of constraints of the SDP, known as the Barvinok-Pataki bound \cite{Barvinok95,Pataki98}. Their method has showed promising performance in practice, even in constant dimensions and is an active area of research in non-convex optimization theory \cite{BVB18,BVB19,BMZ01}. Recently, Mei et al.\, \cite{MMMO17} showed that, for \mc\ SDPs \fe{corresponding to graphs with non-negative edge-weights}, any second-order local optimum for the BM formulation is approximately optimal with respect to the original SDP.

\begin{theorem}[Mei et\ al.\ \cite{MMMO17}]
    \label{thm:mei}
	\fe{For graphs with non-negative edge weights,} the objective at a locally optimal solution, for the above non-convex, rank-$k$ SDP formulation, is within a factor $1-\frac{1}{k-1}$ of that of the rank-$n$ SDP.
\end{theorem}

The above theorem highlights the fact that increasing $k$ improves performance of the BM formulation; however, for the purposes of this work (and simple mapping to the Bloch sphere), we restrict our attention to rank-2 and rank-3 solutions. We next discuss our key ideas on bringing in warm starts from classical optimization to quantum algorithms.

\section{Preprocessing Stage for QAOA-Warm}
\label{sec:pipeline}
In this section, we discuss our classical \emph{preprocessing stage} for warm-starts in \QAOA, which are obtained through the Burer-Monteiro \mc\ (BM-MC$_k$) relaxation in rank $k$ (for $k=2,3$). \fe{Given a classical solution $x_i \in S_k$ (for $i\in V$ for graph $G=(V,E)$), our warm-starts comprise a separable product state $\ket{q_1}\otimes \ket{q_2} \otimes \cdots \otimes \ket{q_n}$, wherein the pure state of each qubit $\ket{q_i}$ can be represented on its own Bloch sphere at the location of the corresponding vertex $x_i \in S_k$. These initial qubit positions are obtained using a classical Burer-Monteiro algorithm in rank-2 (BM-MC$_2$) and rank-3 (BM-MC$_3$).} 

\fe{To motivate such an approach for creating warm-starts for QAOA, we highlight two key observations. First, since the objective of BM-MC$_k$ can be written as $\max_{x_i, x_j \in S_{k-1}} \sum_{(i,j)\in E} w_{ij}\Vert x_i - x_j\Vert_2^2$, the classical solutions are incentivized to move the adjacent vertices as far apart as possible, ideally, to opposite ends of the sphere. This helps increase the probability of an edge being in a cut obtained not only by hyperplane rounding but also quantum sampling (as long as the corresponding qubits are aligned with the measurement axis as \fe{much as} possible, i.e. at the north and south poles of the Bloch sphere). In general, if there is a cluster of vertices at both the poles of the sphere, then the probability of capturing the weight of the edges that go across these clusters is increased for both classical and quantum approaches.}

\fe{Next, we find a reduction to the quantum sampling objective from the BM-MC objective for an edge. Consider an edge $e$, such that one of the vertices is located at the top of Bloch sphere. Then the probability of that edge being cut via quantum sampling and the contribution that edge makes to the BM-MC$_3$ objective coincides. Consider an edge $e = (i,j)$ such that $x_i = (0,0,1)^T$, and $x_j =  (\sin \theta \cos \phi,  \sin\theta \sin \phi, \cos \theta)^T$ (where $\theta$ and $\phi$ are the polar and azimuthal angle respectively). The expected contribution of $e$ to the \mc{} from quantum sampling is equal to $w_{i,j}$ multiplied by the probability that the edge $e$ is cut, i.e., $w_{i,j}\sin^2(\theta/2)$. The contribution to the BM-MC$_3$ objective from edge $e$ can be written as $\frac{1}{2}w_{i,j}(1-x_i \cdot x_j)$. By definition, $\cos(\theta) = x_i \cdot x_j$, and thus, the contribution to the BM-MC$_3$ objective is $\frac{1}{2}w_{i,j}(1-\cos(\theta)) = w_{i,j}\sin^2(\theta/2),$
    which is equivalent to the expected contribution of $e$ from quantum sampling\footnote{This may not be true in general for the \mc over the entire graph, due to alignment with the measurement axis. }.}

%Consider rank-3 initializations, where classically the constraint is $x_i \in \mathbb{R}^k, \|x_i\|=1$ to force the relaxation on the 2-dimensional sphere. One can also interpret these as initial states for qubits for each vertex in the graph, thus producing a separable warm-start state. 

%\fe{In general, the probability of an edge being in the sampled cut (whether by hyperplane rounding or quantum measurement) increases as the two vertices approach opposite poles of the sphere (or similarly for Burer-Monteiro's relaxation, if there is a cluster of vertices at both the poles of the sphere).}

%\ch{How did these figures get so far away from when it's discussed in the text?}

\tsout{In general, the probability of an edge being in the sampled cut (whether by hyperplane rounding or quantum measurement) increases as the two vertices approach opposite poles of the sphere.} %The above observation can also be modified for BM-MC$_2$ with the correct embedding of the BM-MC$_2$ solution (see Section \ref{sec:mapToQuantum}).}  
\tsout{One may also be interested in the relationship between the classical hyperplane rounding and the quantum sampling as well.} \tsout{If we  fix one of the vertices to the top of the sphere, one observes that the classical and quantum probabilities are not too far apart as demonstrated in Figure \ref{fig:measurement}. }

%\rt{The above observation suggests that classical solutions obtained via BM-MC$_k$ can be suitable candidates for constructing initial quantum states. }\rt{Of course, for any pair of vertices, it will not always be the case that one of them is at a pole of the Bloch sphere. However, in ideal cases,}\ch{What does ``ideal cases'' mean? Isn't this antipodal clustering only for graphs which BM-MC solves easily?}\rt{the BM-MC$_k$ solutions will roughly resemble two antipodal clusters of vertices corresponding to the \mc, in which case, with a suitable rotation of the solution (i.e. vertex-at-top rotations as described in Section \ref{sec:randomRotations}), classical hyperplane rounding and quantum sampling approximate one another.}

\begin{figure}
	\centering
	    
	\begin{tabular}{ccc}
	& BM-MC$_2$ & BM-MC$_3$\\
    \rotatebox{90}{\hspace{0.5cm}Vertex-At-Top}&\includegraphics[scale=0.4]{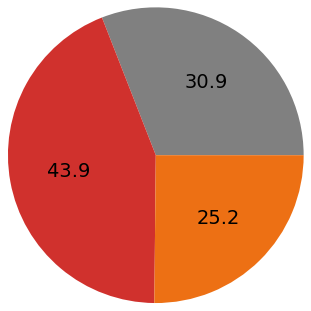} & \includegraphics[scale=0.4]{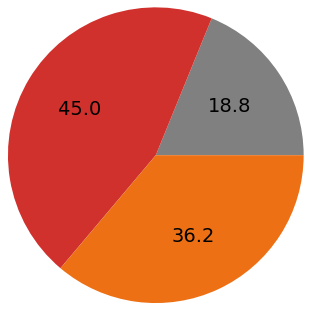}\\
   % \rotatebox{90}{\hspace{0.40cm}Random Uniform}&\includegraphics[scale=0.4]{classicalVsQuantumPie2Dr.png} & \includegraphics[scale=0.4]{classicalVsQuantumPier.png}
    \end{tabular}\raisebox{-1.2cm}{\includegraphics[scale=0.35]{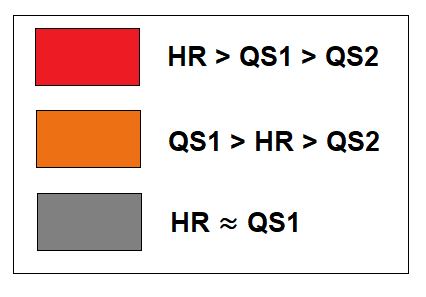}}

	    \caption{\fe{Pie charts representing best expected cut value obtained by using (i) hyperplane rounding of the BM-MC$_k$ solution (HR), (ii) quantum sampling of the BM-MC$_k$ solution (QS1), and quantum sampling of the initial state of standard QAOA (QS2). For every instance, QS2 always yielded the worst result of the three, and for majority of the instances QS1 $\geq$ HR. For HR and QS1, the best of 5 (in terms of SDP objective) locally optimal BM-MC$_k$ solutions are used; for that solution, the best of 5 rotations is used for QS1. The regions marked in gray indicate instances for which QS1 and HR had a tie (difference in approximation ratio of at most 0.001).}}
	    \label{fig:pieChartClassicalVsQuantumSampling}
	\end{figure}
	
\fe{A natural question at this point is if we are worsening the quality of cuts that the QAOA algorithm is initialized with using warm-starts, and if quantum sampling of a classical solution is even competitive compared to a hyperplane rounding of the same. We show in Figure \ref{fig:pieChartClassicalVsQuantumSampling} that quantum sampling of the warm-start (QS1) initialization outperforms the expected cut obtained using the standard initial state for QAOA (QS2). Moreover, with an appropriate initial rotation of the warm-start (Section \ref{sec:randomRotations}), QS1 outperforms hyperplane rounding (HR) for the majority of instances.} 
	
%Our warm-starts comprise a separable product state, wherein the pure state of each qubit can be represented on its own Bloch sphere. 

\fe{We explain next the pipeline of constructing  warm-starts using appropriate initial rotations}.

	\subsection{Solving BM-MC}\label{sec:solvingBMMC} We use the Burer-Monteiro algorithm in $k$ dimensions for $k=2,3$, for finding approximate solutions to \mc. In each dimension, we begin with $n$ points chosen uniformly at random on the unit circle (for $k=2$) or unit sphere (for $k=3$). We represent these points in polar coordinates (for $k=2$) or spherical coordinates (for $k=3$); that is, we keep track of the polar ($\theta$) angles (for $k=2,3$) and azimuthal ($\phi$) angles (for $k=3$) of each point. To find locally optimal solutions, we perform stochastic coordinate ascent\footnote{\emph{Stochastic coordinate ascent} works well in practice in finding a local optimum, see e.g., \cite{MMMO17}. Nevertheless, for guaranteed convergence one can use other methods such as (fast) Riemannian Trust-Region methods.} by making small random perturbations to these angles (thus maintaining feasibility) and update our solution if the objective increases; see Appendix \ref{sec:appendixAlgos} for more detail. We find 5 local optima and take the best solution. Let $\mathbf{x}^*: V \rightarrow S^k$ be a solution to $\text{BM-MC}_k$, i.e.,\, a Burer-Monteiro relaxation of \tsout{Max-Cut} \fe{\mc{}} in the $k$-dimensional space.

	\begin{figure}
    \centering
    \includegraphics[scale=0.4]{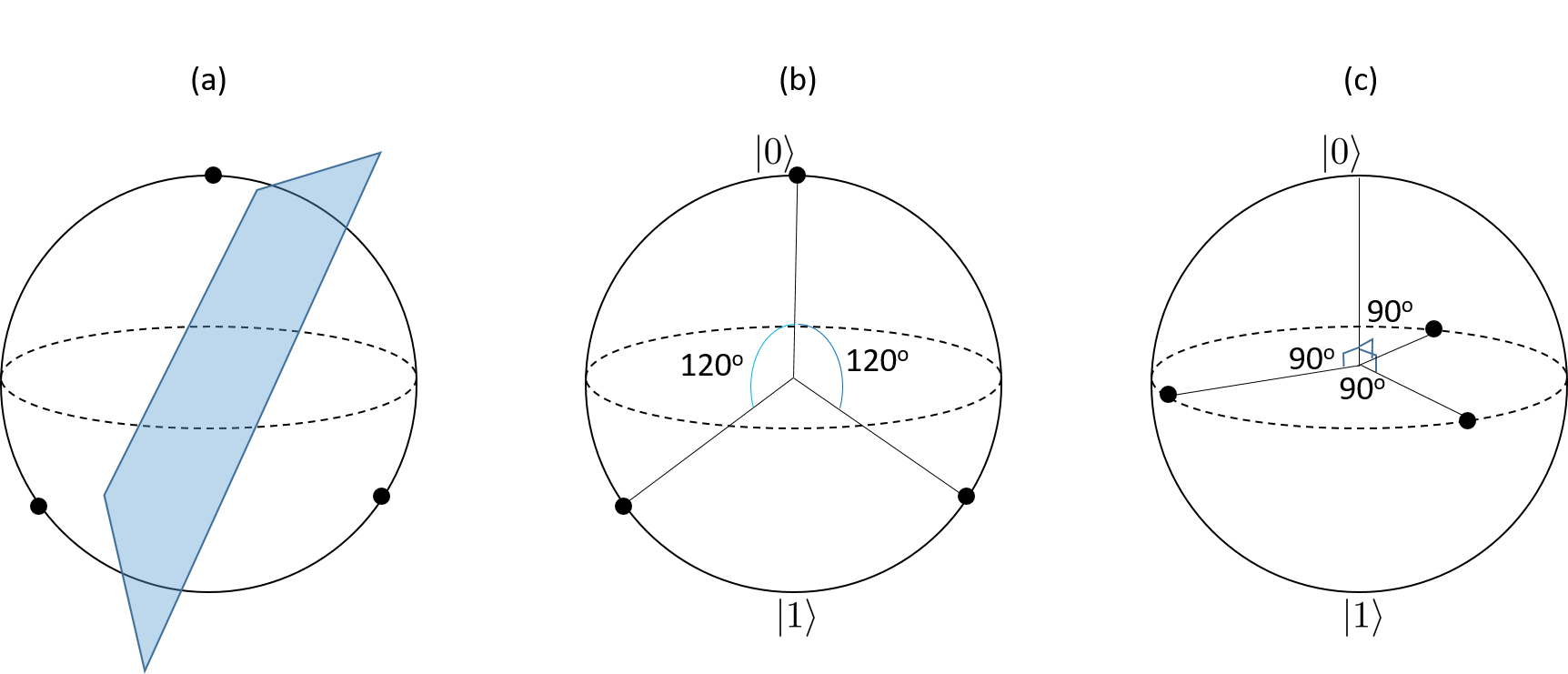}
    \caption{\fe{Comparison of the hyperplane rounding and quantum sampling for a 3-cycle (\mc{}=2): figure (a) shows a local optimal BM-MC$_3$ solution, where any random hyperplane will give a cut of size 2. Both (b) and (c) show two different embeddings of the BM-MC$_3$ solution (from (a)) onto the Bloch sphere. In (b), the qubits lie on $x=0$ plane and quantum sampling results in a expected cut of 1.875. In (c), all qubits lie on the equator of the Bloch sphere (similar to the standard start of QAOA), so each edge has a probability of $1/2$ of being cut, yielding a total expected cut of $1.5$. Both (b) and (c) demonstrate that the orientation of the rotated BM-MC$_3$ solution is important when embedding it into the Bloch sphere and can result in different expected cuts.}}
    %two edges have a probability of $\sin^2(120^\circ/2) = \frac{3}{4}$ of being cut (via quantum sampling) and the remaining edge has a cut probability of $2\cos^2(120^\circ/2)\sin^2(120^\circ/2) = \frac{3}{8}$, 
    \label{fig:diagramClassicalVsQuantum}
    \end{figure}

	\subsection{Random Rotations}\label{sec:randomRotations} Classical hyperplane rounding for BM-MC$_k$ is invariant under a global rotation of the entire solution, however quantum sampling is not. \fe{For example, in Figure \ref{fig:diagramClassicalVsQuantum} we consider 3 rotations of a particular BM-MC initialization of 3 qubits on the Bloch sphere, and though hyperplane rounding is agnostic to a rotation of the Bloch sphere, quantum sampling depends on the choice of the measurement axis. The difference in approximation attained by quantum sampling in two different orientations of the same solution on Bloch sphere demonstrates the importance of choosing a suitable rotation when embedding the BM-MC$_k$ solution to the Bloch sphere.} Thus, before mapping the rank-$k$ approximate solutions from BM-MC$_k$, \fe{a rotation is performed} to mitigate unfavorable orientations due to warm-starts. 
	
	We consider two types of random rotation schemes:
	uniform rotation in $\mathbb{R}^k$ (for all the vertices), and random ``vertex-at-top" rotations where a vertex is sampled uniformly and mapped to the $(0,0,1)^T$ vector for rank-$3$ and $(1,0)^T$ in rank-2 solutions. Uniform random rotations can provably recover a significant fraction of the BM-MC$_k$ objective (see Section \ref{sec:theory}) whereas vertex-at-top rotations serve as a useful heuristic. We describe both of these rotations in Appendix \ref{sec:appendixrotations}. We use the shorthand $R_V(\mathbf{x}^*)$ and $R_U(\mathbf{x}^*)$ to denote the rotations of the approximate solution $\mathbf{x}^*$ by a random vertex-at-top rotation $R_V$ and a random uniform rotation $R_U$ respectively.
	
    \subsection{Mapping to the Bloch Sphere}
    \label{sec:mapToQuantum}
	To map the rotated solutions $R(\mathbf{x}^*) = ((\theta_1,\varphi_1),\dots, (\theta_n, \varphi_n))$ (with $R \in \{R_U,R_V\}$), we can simply map the rank-3 solutions for each vertex to the Bloch sphere (see Figure \ref{fig:transformToInitialState}) using a tensorizable state for each qubit, i.e., the ``quantum mapping" $Q$ is given by:
	$$Q(\mathbf{x}^*) = \fe{Q_3}(\theta_1,\varphi_1) \otimes \dots \otimes{\fe{Q_3}}(\theta_n,\varphi_n),$$
	where
	$${\fe{Q_3}}(\theta,\varphi) = \cos(\theta/2) \ket{0} + e^{i \phi}\sin(\theta/2)\ket{1}.$$
	
	For rank-2 solutions, let $R(\mathbf{x}^*) = (\theta_1,\dots, \theta_n)$ be the rotated approximate solution in polar coordinates where $\theta_i \in [0,2\pi)$ for $i=1,\dots,n$. We embed the solution into the $yz$-plane of the Bloch sphere with the following quantum mapping:
	$$Q(\mathbf{x}^*) = \fe{Q_2}(\theta_1) \otimes \dots \otimes \fe{Q_2}(\theta_n),$$
	where $\fe{Q_2}(\theta)$ is given by:
	$${\fe{Q_3}}(\theta) = \begin{cases} \cos\left(\frac{\theta}{2}\right) \ket{0} + e^{-i \pi/2}\sin\left(\frac{\theta}{2}\right)\ket{1}, &\!\!\!\! \theta \in [0,\pi),\\
	\cos\left(\pi-\frac{\theta}{2}\right) \ket{0} + e^{i\pi/2}\sin\left(\pi-\frac{\theta}{2}\right)\ket{1}, &\!\!\!\! \theta \in [\pi,2\pi).
	\end{cases}$$
	
	The quantum mapping for rank-2 solutions is motivated by the fact that for rank-3 solutions, certain initializations along the $x$-axis cause \QAOAw\ to perform poorly (see Section \ref{subsec:nonoptimalityQAOAw})); mapping to the $yz$-plane of the Bloch sphere allows us to avoid these problematic states. %Precise details of constructing these initial states starting with $\ket{0}^{\otimes n}$ using elementary quantum gates are included in Appendix \ref{sec:mapToQuantum}.

	\begin{figure*}
		\centering
		\includegraphics[scale=0.4]{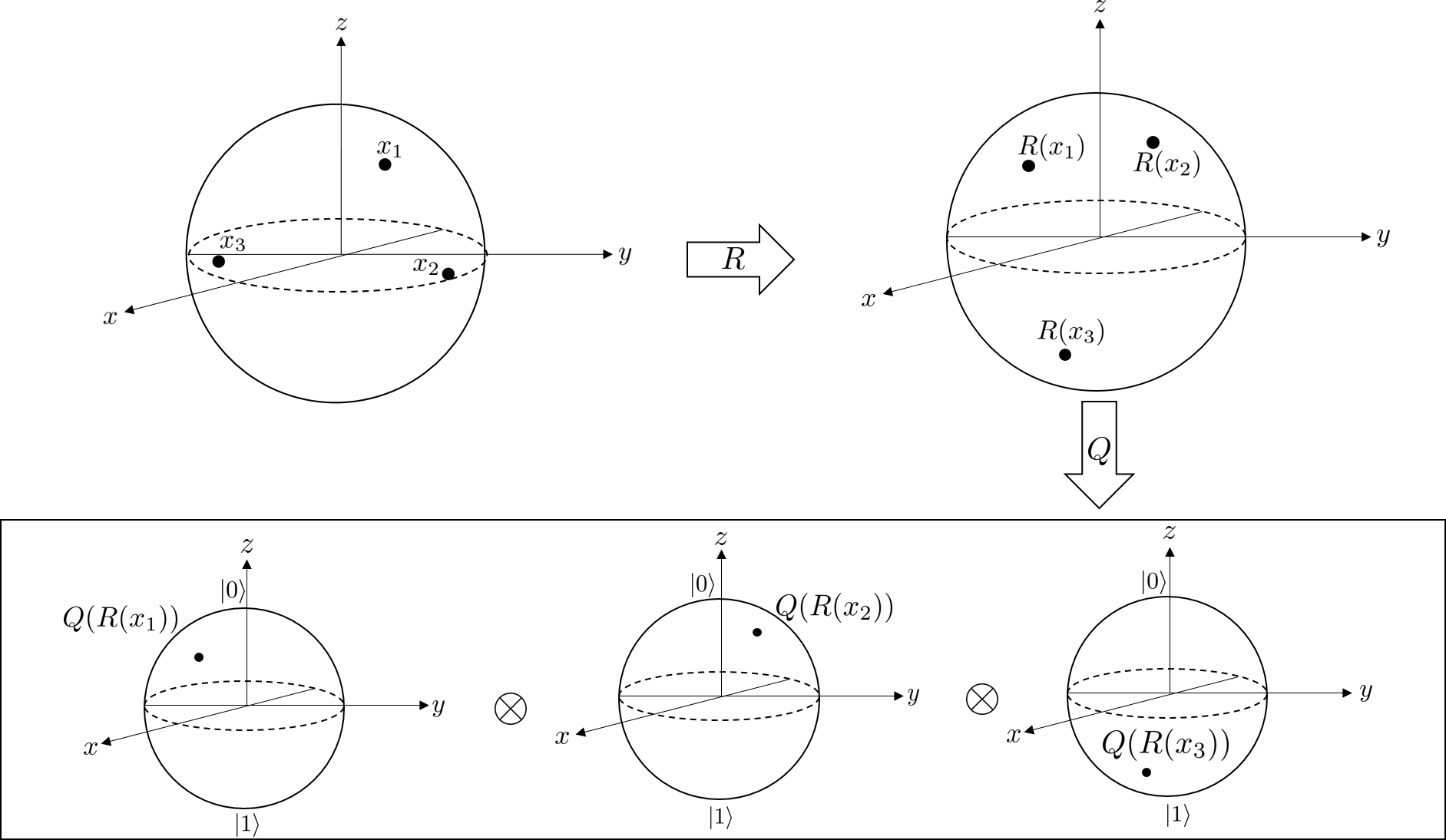}
		\caption{We begin with a locally optimal solution from BM (top-left). We then apply a rotation $R \in \{R_U,R_V\}$; here we show $R_U$ being applied (top-right). Lastly, we use $Q$ to map this rotated solution to a separable quantum state.}
		\label{fig:transformToInitialState}
	\end{figure*}
	\subsection{Performing a biased \QAOA}\label{sec:performingBiasedQAOA} Now we perform \QAOA, as described in Section \ref{sec:QAOA}, while redefining initial state $\ket{s_0}$ as the tensor product of the qubit-states due to the previous step, i.e., 
	$$
	\ket{s_0} = \otimes_{v \in V} Q\fe{_k}(R(\mathbf{x}^*(v)))\,,
	$$
	and run \QAOA\ with the chosen depth and optimize over the $2p$ parameters $\gamma=(\gamma_1,\dots,\gamma_p), \beta=(\beta_1,\dots,\beta_p)$ in order to maximize
	$
	{F_p(\gamma,\beta)} = \bra{\psi_p(\gamma,\beta)} H_C \ket{\psi_p(\gamma,\beta)}$. We initialize $\gamma,\beta$ close to $\mathbf{0}$ which allows us to start with a solution quality close to what would be obtained by just doing quantum sampling.\footnote{Initializing the parameters to \emph{exactly} $\mathbf{0}$ is not advised due issues regarding saddle points. In particular, standard \QAOA\ at $p=1$ has a saddle point at the origin and thus terminates immediately. We instead initialize $\gamma_k,\beta_k$ by sampling uniformly at random from the interval $[-0.0001,0.0001]$ for $k=1,\dots,p$. In the case of standard \QAOA, if we still get stuck after a few epochs (due to the saddle point), we discard the run and retry with new randomized initial parameter values.} Moreover, the ridge-like geometry of the \emph{parameter landscape} (seen in Sections \ref{subsec:landscapes} and \ref{subsec:approxBounds}) also suggests that points near the origin are suitable for initialization of the variational parameters. 

\begin{algorithm}[t]\caption{\QAOAw\ using BM-MC$_k$}
	\label{pseudo-code}
	    \DontPrintSemicolon
		\KwIn{$G = (V, E), w: E \rightarrow \mathbb{R}, p,k, R \in \{R_U,R_V\}$}
		$x \gets \text{BM-MC$_k(G,w)$}$ \tcp*[r]{approximate solution}
		$\ket{s_0} \gets \otimes_{i \in V} Q\fe{_k}(R(x_i))$, for $i \in [n]$, for random $R$\;
		\Return{ \QAOA$(G,w,\ket{s_0},p)$}
\end{algorithm}

We summarize \QAOAw\ in pseudocode in Algorithm \ref{pseudo-code}. In the next section we present experimental results with warm-starts, followed by theoretical development of properties of standard \QAOA\ and \QAOAw.

\section{Results}\label{sec:computational}

In this section, we discuss the results of our numerical simulations of \QAOAw. \fe{We first discuss the details of the preprocessing pipeline and the graph instances used in Section \ref{subsec:experimentalSetup}. In order to compare \QAOAw{} to other \mc{} algorithms, one can use different black-box optimizers, such as ADAM, COBYLA, Nelder-Mead and BFGS. We first run computations to pick a single optimizer, then to pick the rank of the initialization, and the rotation scheme to work with in Sections \ref{sec:optimizerChoice} and \ref{sec:rankRotationChoice}. In Section  \ref{sec:aggregate}, we next provide aggregate results for \QAOAw{} including (i) a comparison against other \mc{} algorithms, (ii) improvement in approximation ratio with increased $p$ depth, and (iii) trends in (median) approximation ratio with varying $p$-depth and graph size. Lastly, to understand the behavior of \QAOAw{}, we discuss the qualitative shape and numerical properties of the parameter landscape of \QAOAw{} (and standard \QAOA) in Section \ref{subsec:landscapes}.}

\subsection{Experimental Setup}
\label{subsec:experimentalSetup}
\paragraph{Graph Instances.} We consider a collection of \tsout{1148}\fe{1264} graphs, $\mathcal{G}$, generated as follows. We first generated a set of unweighted graphs, which includes all non-isomorphic graphs for $n=2$ to $n=6$ vertices (142 instances), and 29 random graphs for each size $n=7$ to $n=\tsout{11}\fe{12}$ sampled from different random graph generators in {\sc Python}'s {\sc NetworkX} \cite{networkx} package. These random graph generators include Erd\"{o}s-Reny\'{i}, Barabasi Albert, Dual of Barabasi-Albert, Watts-Strogatz, and Newman-Watts-Strogatz models (detailed in Appendix \ref{sec:graphInstancesAppendix}). Many experimental studies in the current QAOA literature only consider graphs from a single random graph model (e.g. Erd\"os-Reny\'i); however graphs from such models can have predictable behavior when it comes to \tsout{Max-Cut} \fe{\mc}\footnote{For example, when using the Erd\"os-R\'enyi graph model, if each edge appears independently with probability $q$, and if we take a random cut with $k$ vertices on one side and $n-k$ vertices on the other side, then one would observe $qk(n-k)$ edges across the cut (in expectation).} which could potentially have a large influence on the performance of QAOA. For this reason, we construct an ensemble of graphs $\mathcal{G}$ using a variety of random graph generators.\tsout{demonstrates how varied our ensemble is with respect to two important graph metrics dependent on eigenvalues of the normalized Laplacian \cite{lyons}. }

Next, we create three weighted versions of each of the \tsout{287}\fe{316} unit-weighted instances constructed above, by considering independent edge-weightings drawn from (i) uniform distribution on $\{-10,-9,\cdots,9,10\}\setminus \{0\}$, (ii) uniform distribution on $\{1,2,\dots,10\}$, and (iii) weights of form $\pm2^k$ with $\Pr[w_e = 2^k] = \Pr[w_e = -2^k] = 2^{-k-2}$ for all non-negative integers $k$. The weighted and unweighted instances together give us a total of \tsout{1148}\fe{1264} instances. The last family of weighted instances is constructed due to high variation of performance of classical heuristics on similar instances, observed in a previous study by Dunning et al \cite{dunning2018}. Note that the last two ways of sampling edge-weights results in only positive edge-weight graphs. We will often present results for mixed-weight graphs (positive and negative weights), and positive-only separately.

% We ran our experiments with two graph collections. The first collection, $\mathcal{G}_1$, is the set of all connected unweighted non-isomorphic graphs with five vertices. The collection $\mathcal{G}_1$ was generated in \texttt{SageMath} using a library function that enumerates all non-isomorphic graphs on $n$ vertices (in this case, $n=5$). The second collection, $\mathcal{G}_2$, is a collection of 12-node random graphs given by $$\mathcal{G}_2 = \bigcup_{\xi \in P} \mathcal{G}_{2,\xi}\,,$$ where $P = \{0.1,0.2,0.4,0.6,0.8\}$ and $\mathcal{G}_{2,\xi}$ is a sub-collection of 50 connected random Er\"ods-R\'enyi  graphs where each of the possible $\binom{n}{2}$ edges is included independently with probability $\xi$.\footnote{For each $\xi \in P$, as we generated the random graphs, we would discard any disconnected graphs that were generated. We kept generating graphs in this fashion until we had 50 graphs (for each $\xi \in P$).} In the subsections below, we will use $\mathcal{G} = \mathcal{G}_1 \cup \mathcal{G}_2$ to refer to both sets of graphs.

\paragraph{Running the Preprocessing Stage}
We computed \QAOAw, Goemans-Williamson and standard \QAOA\ solutions for each of the weighted graph instances $G \in \mathcal{G}$. Both standard \QAOA\ and \QAOAw\ were run for circuit depths $p \in \{1,2,4,8\}$, for each optimizer considered, and \QAOAw\ for each considered rank of BM-MC$_k$ ($k=2,3$) and for each rotation type (vertex-at-top and uniform random). We consider the best of 5 warm-starts (in objective value) when selecting BM-MC$_k$ warm-starts, and subsequently the best of 5 random rotations, i.e., the rotation that yields the highest expected cut value at the end of the hybrid-optimization loop. For any two runs of standard \QAOA\ or \QAOAw\ that differ only in choice of optimizer, the initial parameter values used are the same (with $\gamma_i$ and $\beta_i$ sampled uniformly from the interval $[-0.0001,0.0001]$ for all $i=1,\dots,p$). Our implementation of standard \QAOA\ and \QAOAw\ utilizes Google's {\sc Tensorflow Quantum} library and IBM's Qiskit library. The state $\ket{+}^{\otimes n}$ is initialized by applying a Hadamard gate to each qubit in $\ket{0}^{\otimes n}$. For states initialized based on low-rank approximate solutions, we generate the initial state as discussed in Section \ref{sec:pipeline}, which is easily implemented using standard rotation gates. For each epoch of each run of standard \QAOA\ or \QAOAw, our implementation records the values of the variational parameters, the expected cut value at those parameters, and the probability distribution of all $2^n$ cuts. Each run of standard \QAOA\ and \QAOAw\ terminated when the difference in successive values of $F_p(\gamma,\beta)$ was less than $\bar{W}*10^{-6}$, where $\bar{W}$ is the sum of the absolute values of the edge weights. We next summarize the results from these numerical simulations.

\subsection{Optimizer Choice}
\label{sec:optimizerChoice}

\begin{figure}
    \centering
    \includegraphics[scale=0.31]{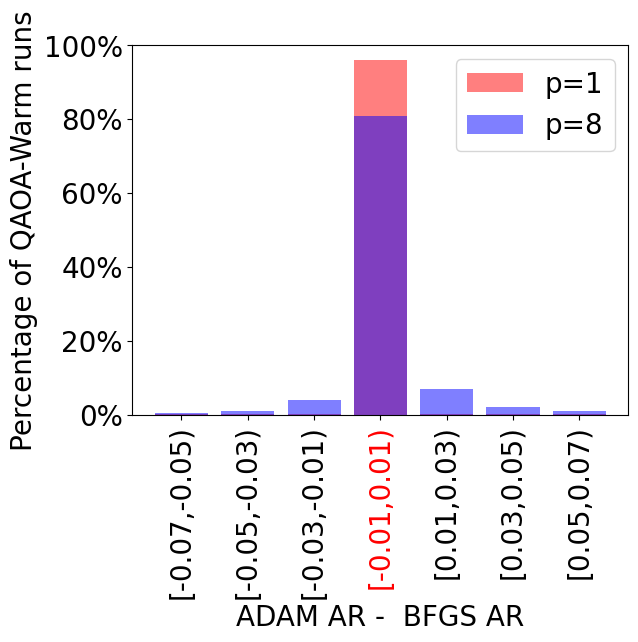}\includegraphics[scale=0.31]{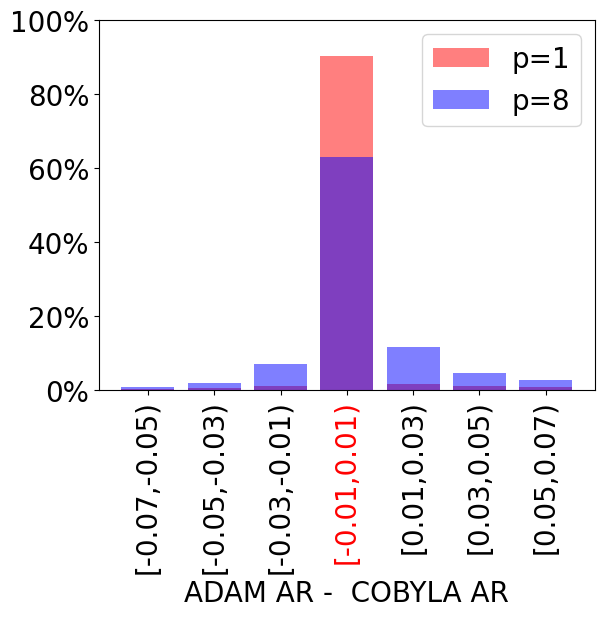}\includegraphics[scale=0.31]{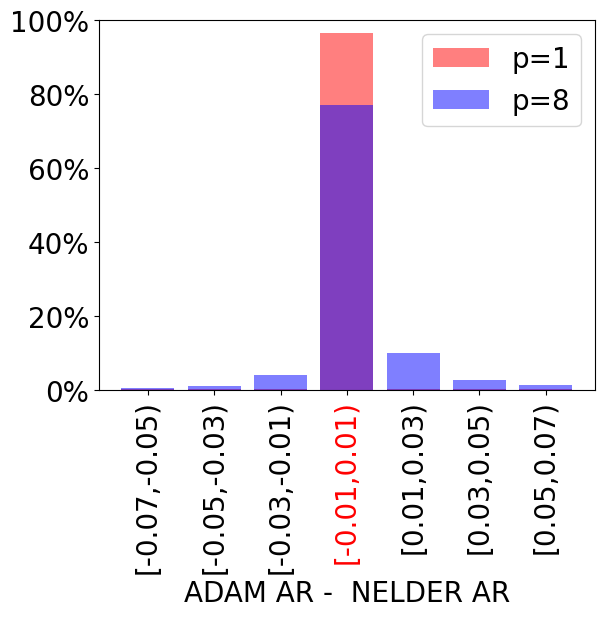}\\
    \includegraphics[scale=0.31]{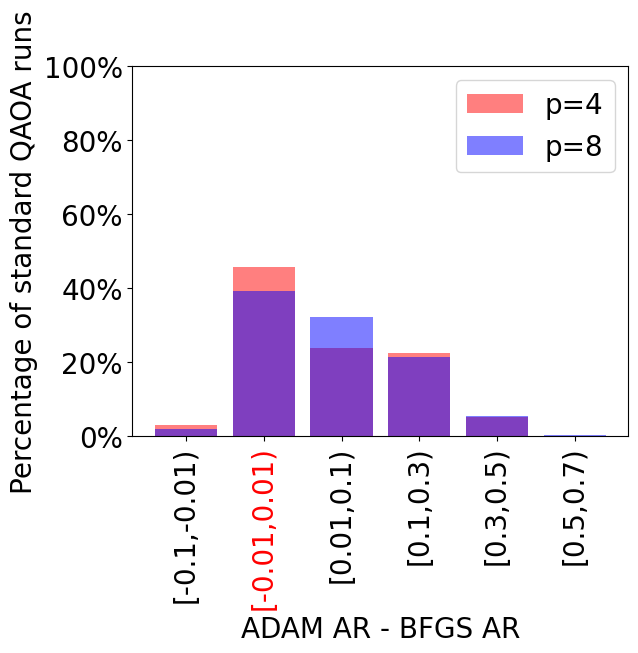}\includegraphics[scale=0.31]{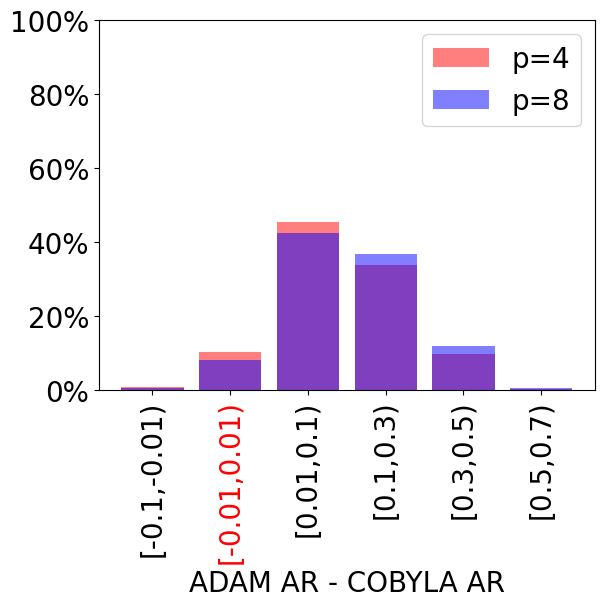}\includegraphics[scale=0.31]{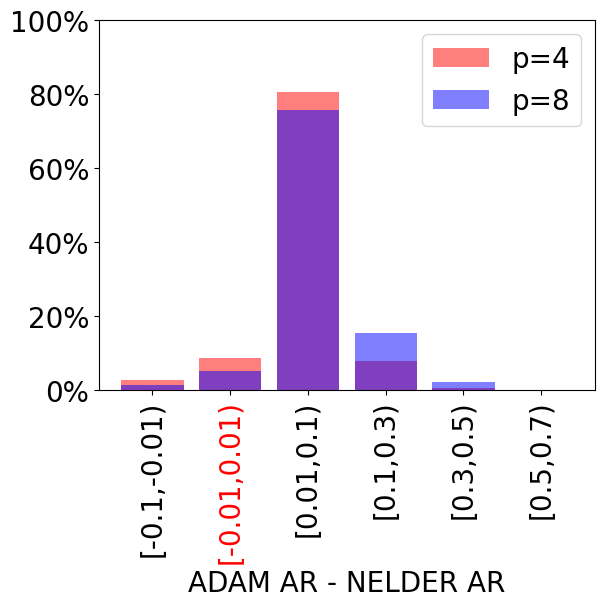}
    \caption{This figure compares runs of \QAOAw\ and standard \QAOA\ that only differ in choice of optimizer. We compare each optimizer to ADAM; more specifically, we plot $r_A-r_O$ where $r_A$ is the approximation achieved by ADAM and $r_O$ is the approximation achieved by the other optimizer considered. The top row shows comparisons of \QAOAw\ runs for circuit depths $p=1$ (red) and $p=8$ (blue) whereas the bottom row shows comparisons for standard \QAOA\ for circuit depths $p=4$ (red) and $p=8$ (blue). Overlapping regions of the histograms are in purple. We give a red label to the bin $[-0.01,0.01]$ to highlight the cases where the optimizers perform similarly (to ADAM).}
    \label{fig:optimizerComparison}
\end{figure}

We consider four different optimizers to optimize the $2p$ variational parameters: ADAM, BFGS, Nelder-Mead, and COBYLA and present comparisons between these set of optimizers. As demonstrated in Figure \ref{fig:optimizerComparison}, when ADAM is compared to the other three optimizers, the expected cut values obtained for \QAOAw\ are similar (i.e. within 0.01 difference in approximation ratio) for at least 90\% of the runs at $p=1$; this percentage decreases at $p=8$ but the approximation ratios are still relatively similar for the majority of the instances. This suggests that the approximation ratios achieved for \QAOAw\ are largely independent of the optimizer used; for this reason, all remaining results involving approximation ratios for \QAOAw\ will be in terms of runs using the ADAM optimizer. It should be noted that all of the optimizers considered vary in regards to runtime, e.g., the cost per iteration and the number of iterations required to train the variational parameters (we discuss this further in Appendix \ref{sec:appendixPreProcess})

Even though the choice of the optimizer had almost no impact on the \QAOAw\ in terms of approximation factors obtained, Figure \ref{fig:optimizerComparison} illustrates a noticeable effect on approximation ratios achieved for standard \QAOA\ however, especially at the higher circuit depths that we tested ($p=4$ and $p=8$).\footnote{We suspect this is an artefact of the parameter landscapes becoming flatter with the warm-starts.} In particular, we find that runs using the ADAM optimizer tend to have better performance for QAOA. For this reason, the remaining results in this paper regarding standard \QAOA\ will only include runs that utilize the ADAM optimizer in order to obtain a more simple, direct, and fair comparison with \QAOAw. 

\begin{table}
    \centering
    \begin{tabular}{ccc}
    & depth $p=1$ & depth $p=8$\\\\
    all graphs & 
    \begin{tabular}{c|cc}
    &  vert.  & uniform\\\hline
    rank-2 & {\bf \fe{0.9581}} & \fe{0.9581} \\
    rank-3 & \fe{0.9576} & \fe{0.9440}
    \end{tabular}
    &
    \begin{tabular}{c|cc}
    &  vert.  & uniform\\\hline
    rank-2 & {\bf \fe{0.9726}} & \fe{0.9718}\\
    rank-3 & \fe{0.9688} & \fe{0.9560}
    \end{tabular}
    \\\\
    \begin{tabular}{c}positive-weight \\ graphs\end{tabular}
    &
    \begin{tabular}{c|cc}
    &  vert.  & uniform\\\hline
    rank-2 & \fe{0.9569} & {\bf \fe{0.9569}} \\
    rank-3 & \fe{0.9556} & \fe{0.9441}
    \end{tabular}
    &
    \begin{tabular}{c|cc}
    &  vert.  & uniform\\\hline
    rank-2 &  {\bf \fe{0.9704}} &  \fe{0.9697} \\
    rank-3 &  \fe{0.9659} &  \fe{0.9548}
    \end{tabular}
    \end{tabular}
    \caption{Multiple tables comparing the average approximation ratio achieved during \QAOAw\ when utilizing different combinations of ranks and rotations during the preprocessing stage. For the top row of tables, these averages were computed using all the graphs in our graph library $\mathcal{G}$ (see Section \ref{subsec:experimentalSetup}) whereas for the bottom row, we restrict our attention to only those graphs in $\mathcal{G}$ with positive edge weights. Each run of standard \QAOA\ and \QAOAw\ terminates when the difference in successive values of $F_p(\gamma,\beta)$ is less than $10^{-6}\bar{W}$ where $\bar{W}$ is the sum of the absolute values of the edge weights.}
    \label{fig:comparingRanksAndRotation}
\end{table}

\subsection{Choice of Rank and Rotations}
\label{sec:rankRotationChoice}
%\sg{These results are for "best of 5" vertex at top rotations. Not for "median vertex at top rotations". Am I correct?} 
\fe{To compare \QAOAw{} against standard-QAOA, GW, and hyperplane rounding of BM-MC$_2$ , we need to narrow in to the choice of the BM-MC$_k$ rank (2 or 3) and the type of rotation (vertex-at-top or uniform random) to use. We explore these two choices in this subsection.}%\tsout{We considered} two possible choices for both the rank of the initialization (rank 2 or rank 3) and the type of rotation (uniform or vertex-at-top) to construct the initial quantum state\rt{we desire to pick the combination that results in the highest approximation ratio}.

Recall that we consider the best-of-5 warm-starts for each type of rotation\footnote{We found that restarting \fe{standard} \QAOA\ multiple times did not impact the results significantly.}. Over the \tsout{1148}\fe{1264} graph instances, for rank-3 initializations, we find that the vertex-at-top rotations \fe{typically} have a \fe{slight} increase in performance over random uniform rotations, \fe{especially when rank-3 solutions are used}  (\fe{e.g., at depth $p=1$, rank-3} vertex-at-top rotations obtain \fe{0.9576}\tsout{0.9228} approximation ratio on average, whereas \fe{rank-3} uniform rotations obtain \fe{0.9440}\tsout{0.8958}). These results seem reasonable since vertex-at-top rotations rarely end up in states that plateau for warm-starts (see Section \ref{subsec:nonoptimalityQAOAw} for an example of such a warm-start). We include a summary of average approximation ratios observed across the four choices of rank and rotations in  Table \ref{fig:comparingRanksAndRotation}. 

On the other hand, when using rank-2 initializations, there is virtually no difference between the two rotation approaches, as rank-2 solutions were specifically designed to avoid bad states for warm-starts. For the ease of presentation, the remainder of the results in this paper will utilize rank-2 initializations with a vertex-at-top rotation scheme as this appears to be one of the most promising combinations for \QAOAw.

\begin{table}
    \centering
\begin{tabular}{|c|cc|cc|cc|cc|}\hline
 & \multicolumn{2}{c|}{p=1} & \multicolumn{2}{c|}{p=2} & \multicolumn{2}{c|}{p=4} & \multicolumn{2}{c|}{p=8} \\ \hline
 & all & positive & all & positive & all & positive & all & positive\\ \hline
 WBGS & \textbf{25.08}\% & \textbf{25.08}\% & \textbf{26.42}\% & \textbf{26.42}\% & \textbf{23.97}\% & \textbf{23.97}\% & \textbf{16.61}\% & \textbf{16.61}\%\\ \hline
WBSG & 0.00\% & 0.00\% & 0.16\% & 0.31\% & 0.32\% & 0.31\% & 0.71\% & 0.79\%\\ \hline
WGBS & \textbf{30.93}\% & \textbf{30.93}\% & \textbf{32.28}\% & \textbf{32.28}\% & \textbf{29.98}\% & \textbf{29.98}\% & \textbf{21.84}\% & \textbf{21.84}\%\\ \hline
WGSB & 0.00\% & 0.00\% & 0.24\% & 0.31\% & 0.32\% & 0.47\% & 2.06\% & 3.30\%\\ \hline
WSBG & 0.00\% & 0.00\% & 0.16\% & 0.16\% & 2.61\% & 2.52\% & 4.27\% & 4.56\%\\ \hline
WSGB & 0.08\% & 0.16\% & 0.00\% & 0.00\% & 2.29\% & 2.04\% & 4.27\% & 3.62\%\\ \hline
BWGS & 0.95\% & 1.89\% & 0.71\% & 1.10\% & 0.16\% & 0.16\% & 0.00\% & 0.00\%\\ \hline
BWSG & 0.08\% & 0.16\% & 0.24\% & 0.47\% & 0.08\% & 0.00\% & 0.00\% & 0.00\%\\ \hline
BGWS & \textbf{22.39}\% & \textbf{22.39}\% & \textbf{17.01}\% & \textbf{17.01}\% & \textbf{6.25}\% & 6.60\% & 1.34\% & 0.47\%\\ \hline
BGSW & 1.50\% & 2.36\% & 4.83\% & 5.50\% & \textbf{10.44}\% & \textbf{10.44}\% & 4.03\% & 6.29\%\\ \hline
BSWG & 0.24\% & 0.47\% & 0.32\% & 0.63\% & 0.71\% & 1.26\% & 0.63\% & 1.26\%\\ \hline
BSGW & 0.32\% & 0.63\% & 0.24\% & 0.47\% & 1.34\% & 1.89\% & 1.11\% & 2.04\%\\ \hline
GWBS & 1.50\% & 1.89\% & 1.50\% & 1.89\% & 1.27\% & 1.57\% & 0.47\% & 0.31\%\\ \hline
GWSB & 0.08\% & 0.16\% & 0.08\% & 0.16\% & 0.08\% & 0.16\% & 0.55\% & 0.94\%\\ \hline
GBWS & \textbf{15.66}\% & \textbf{15.66}\% & \textbf{10.92}\% & \textbf{10.92}\% & 4.91\% & 4.87\% & 1.11\% & 0.63\%\\ \hline
GBSW & 0.71\% & 0.63\% & 3.72\% & 3.62\% & 6.09\% & \textbf{6.09}\% & 2.69\% & 4.72\%\\ \hline
GSWB & 0.00\% & 0.00\% & 0.00\% & 0.00\% & 0.08\% & 0.16\% & 0.40\% & 0.63\%\\ \hline
GSBW & 0.00\% & 0.00\% & 0.00\% & 0.00\% & 0.08\% & 0.16\% & 0.24\% & 0.31\%\\ \hline
SWBG & 0.00\% & 0.00\% & 0.00\% & 0.00\% & 1.27\% & 1.57\% & \textbf{8.86}\% & \textbf{8.86}\%\\ \hline
SWGB & 0.16\% & 0.31\% & 0.32\% & 0.47\% & 1.74\% & 1.73\% & 8.23\% & 7.39\%\\ \hline
SBWG & 0.00\% & 0.00\% & 0.40\% & 0.79\% & 0.87\% & 1.57\% & 1.42\% & 2.20\%\\ \hline
SBGW & 0.16\% & 0.31\% & 0.24\% & 0.31\% & 2.69\% & 2.67\% & \textbf{11.71}\% & \textbf{11.71}\%\\ \hline
SGWB & 0.16\% & 0.31\% & 0.16\% & 0.31\% & 0.08\% & 0.16\% & 0.08\% & 0.16\%\\ \hline
SGBW & 0.00\% & 0.00\% & 0.08\% & 0.00\% & 2.37\% & 1.73\% & 7.36\% & 8.33\%\\ \hline\hline
Total & 100\% & 100\% & 100\% & 100\% & 100\% & 100\% & 100\% & 100\% \\ \hline\end{tabular}
    \caption{\fe{We consider 4 algorithms: Goemans-Williamson (G), rank-2 Burer-Monteiro with hyperplane rounding (B), \QAOAw{} (W), and standard \QAOA{} (S). There is a row for each of the $4! = 24$ ways the algorithms can perform relative to one another with the cell value indicating the percentage of instances for which that ordering occurs. As an example, the top-leftmost value indicates that for 25.08\% of instances, $W \geq B \geq G \geq S$ in terms of expected AR with $W$ and $S$ being depth-1. The four largest entries in each column are bolded for emphasis. To account for numerical error for nearly solved instances, we declare \QAOAw{} (W) as the best as long as it is within 0.001 AR of the best algorithm. We include columns corresponding to the entire graph library $\mathcal{G}$ as well as the subset of $\mathcal{G}$ that have positive-weighted edges.}}
    \label{tab:comparingAlgorithms}
\end{table}

\begin{figure}
    \centering
    \includegraphics[scale=0.4]{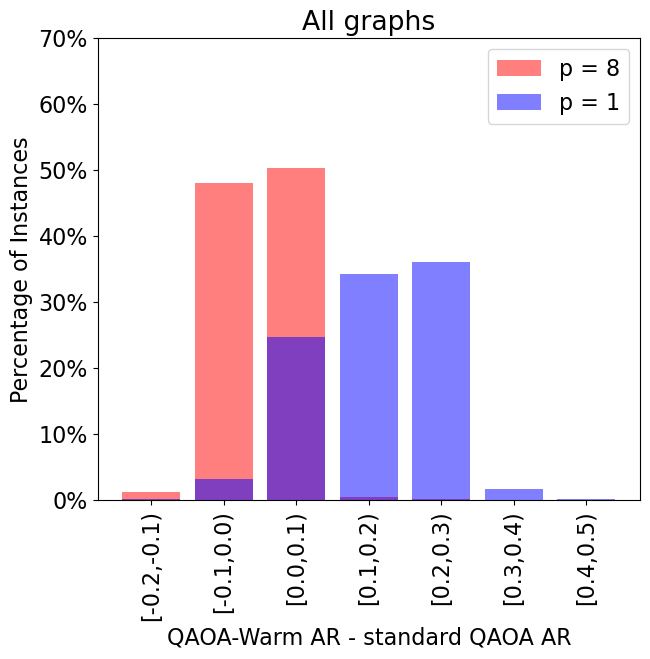}\includegraphics[scale=0.4]{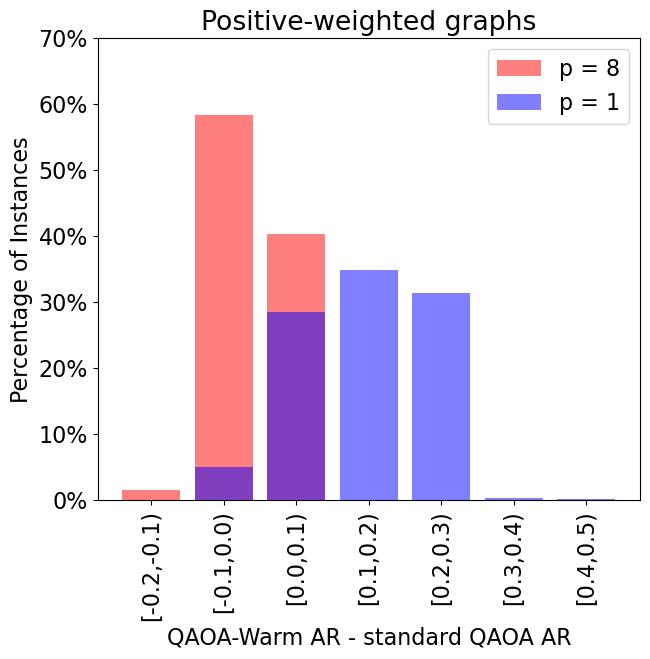}\\
    \caption{Histograms comparing the performance in (depth-$p$) \QAOAw\ (with best of 5 rotations) and (depth-$p$) standard \QAOA\ for both $p=1$ (blue) and $p=8$ (red). Overlapping portions of the histogram are in purple. The \fe{left} plot is generated using the graphs in our graph library $\mathcal{G}$ (see Section \ref{subsec:experimentalSetup}) whereas for the \fe{right} plot, we restrict our attention to only those graphs in $\mathcal{G}$ with positive edge weights. \tsout{Larger bin sizes are used for the far right portions of the histogram; we color the bin label to emphasize this.}} %\ch{I think it's interesting to compare them which would be hard if one is in the appendix.}
    \label{fig:histogramWarmVsStandard}
\end{figure}

% \begin{figure}
%     \centering
%     \includegraphics[scale=0.35]{Alt AR/comparingClassicalAndQuantum2D_v_altAR.png}\\
%     \includegraphics[scale=0.35]{comparingClassicalAndQuantum2D_v_positive.png}\\
%     \caption{Histograms comparing the (expected) approximation ratio achieved via classical hyperplane rounding (of the BM-MC$_2$ solution) with the (expected) approximation ratio achieved by both depth-0 (blue) and depth-8 (red) QAOA-Warm (with best of 5 vertex-at-top rotations). Overlapping portions of the histogram are in purple. The top plot is generated using the graphs in our graph library $\mathcal{G}$ (see Section \ref{subsec:experimentalSetup}) whereas for the bottom plot, we restrict our attention to only those graphs in $\mathcal{G}$ with positive edge weights. \rt{Larger bin sizes are used for the far right portions of the histogram; we color the bin label to emphasize this.}} 
%     \label{fig:histogramCompareWithBM}
% \end{figure}

\subsection{Aggregate Results}\label{sec:aggregate}
\fe{Here we use aggregated results of \QAOAw{} in order to answer three key questions: (Q1) How does \QAOAw{} fare compare to standard \QAOA{} and classical \mc{} algorithms (BM-MC and Goemans-Williamson), (Q2) How much of \QAOAw's approximation ratio can be attributed to the warm-start itself v/s what is done by the quantum circuit, and (Q3) What are the trends in \QAOAw's approximation ratio with varying depth and graph size and how does this compare with standard \QAOA?}

%\sg{Go over again. Check for flow. Add broader picture statements.} 

\tsout{We next consider the aggregate performance in terms of the approximation ratio obtained by \QAOAw\ (using the best of five vertex-at-top rotations). We begin by seeing how the performance of \QAOAw\ compares to both standard \QAOA\ as well as GW. In Figure \ref{fig:pies} we plot a pie chart with the percentage of instances for which different algorithms did best. Starting at circuit depth $p=1$, we see that \QAOAw\ is the best of the three algorithms for 55.2\% of the instances (where initially \QAOA\ was unable to outperform GW). As the circuit depth increases, standard \QAOA\ outperforms the other two approaches at higher $p$ depth (going from less than 2\% to 40\% of the pie as we go from $p=1$ to $p=8$). We observe a similar phenomenon when considering only positive-weighted graphs with standard \QAOA\ taking up even more of the pie (46\%).}

\noindent 
{\bf (Q1).}\fe{To answer the first question, we compare standard \QAOA, \QAOAw, GW, and hyperplane rounding of the BM-MC$_2$ solutions in Table \ref{tab:comparingAlgorithms}. At depth-1, \QAOAw{} is at least as good as the other three algorithms for 56.1\% of the instances meanwhile standard \QAOA{} is the best for less than 1\% of the instances. However, as the circuit depth increases, standard \QAOA{} is the best algorithm for a larger proportion of instances (37.66\% of instances at depth $p=8$); meanwhile, \QAOAw{} is still at least as good as the other algorithms for 49.8\%, nearly half, of the instances.} These results support our claim that warm-starts show improvements in performance of \QAOA\ at low circuit depths. Since standard \QAOA\ achieves the optimal cut in the limit as the circuit depth increases and thus, for any particular graph, there exists some (instance-dependent) circuit depth $p$ for which standard \QAOA\ beats GW \cite{FGG14}. %On the other hand, we show in Section \ref{subsec:nonoptimalityQAOAw} that there exists situations in which \QAOAw\  ``gets stuck" and fails to improve as the circuit depth increases. 
Current and near-term quantum devices are only able to reliably run \QAOA\ for low circuit depth (due to the presense of quantum noise), and therefore we propose that \QAOAw\ can be of significant use in this regime. Although our current implementation of \QAOAw\ does not perform as well at higher circuit depths (compared to standard \QAOA), it may be possible to extend \QAOAw\ in order to see continued improvement with increased circuit depth by changing the mixers; we discuss this more in Section \ref{sec:discussion}.

\fe{We next consider the difference in approximation ratios obtained by \QAOAw\  and standard \QAOA. In Figure \ref{fig:histogramWarmVsStandard}, we provide a detailed comparison between approximation ratios attained by \QAOAw\ and standard \QAOA\ in the form of a histogram.} We see improvements in the approximation ratio ranging from 0.1 to 0.5 when using warm-starts, especially at low circuit depth. These results are consistent with those depicted in \tsout{the pie charts in Figure \ref{fig:pies}} \fe{Table \ref{tab:comparingAlgorithms}}. %moreover, we find that the gap in approximation ratio for \QAOAw\ and standard \QAOA\ is large \rt{over 0.3} for some instances \rt{at depth $p=1$}\sout{(more than 3.0 when considering all graphs in $\mathcal{G}$ and over 0.3 when considering only those graphs in $\mathcal{G}$ with positive edge weights.)}. 
\fe{We note that in this figure, as in the others, we take the best of 5 vertex-at-top rotations for \QAOAw{}; and in Appendix \ref{sec:QAOAWarmWithMedianAndWorstRotations}, we include results in the case where the median and worst (of 5) vertex-at-top rotations are used instead.}

\noindent 
{\bf (Q2.) }\fe{We now address the second key question regarding how much of the performance of \QAOAw{} can be attributed to the warm-start itself. This is an important question to address because if} the improvement generated by \QAOAw\ \tsout{could be} is due only to the initial quantum state at $p=0$ having higher overlap with good solutions, \fe{then there would be no point in running the quantum device}. To test this, we compare, in Figure \ref{fig:histogramP0vsP8}, the \fe{improvement in approximation ratio  from depth-0 \QAOAw\ (i.e. just measuring the initial state obtained from the preprocessing stage) to depth-1 \QAOAw, as well as the improvement when we change the depth from 1 to 8. For 74 instances, we observed that the approximation ratio from \QAOAw{} improved by at least 50\% when going from $p=0$ to $p=1$ and by at least 80\% for 22 instances. This shows the promise of using QAOA on top of the warm-starts. On the other hand, the increase in approximation ratio from depth-1 \QAOAw{} to depth-8 \QAOAw{} is milder, ranging upto 10\% for positive-weighted instances and upto 22.3\% for general graphs.}
%\QAOAw (compared to 10.5\% of instances when going from $p=1$ to $p=8$). Some instances show a large improvement: for 5.9\% of instances, the improvement from depth-0 \QAOAw{} to depth-1 \QAOAw{} is over 50\% (whereas the maximum improvement seen from $p=1$ to $p=8$ is 22.3\%). 
\fe{These results show that running \QAOAw{} does yield an increase in approximation ratio beyond simply sampling the initial warm-start state; however, the returns diminish with higher circuit depths (this is expected because \QAOAw can plateau for some instances, Section \ref{subsec:nonoptimalityQAOAw}).} \tsout{We find that, although the difference is modest (less than 0.1) in the majority of instances, there are many instances where running \QAOA\ greatly improves the solution quality over only sampling the initial warm-start state.} 

\noindent 
{\bf (Q3).} Lastly, \fe{to address the third question,} we consider how the performance of \QAOAw\ varies across $n$ (number of nodes) and $p$ (circuit depth), which we illustrate for our graph library in Figure \ref{fig:performanceAcrossNandP}. While there is a significant improvement in performance for standard \QAOA\ with increasing circuit depth, we find that \QAOAw{} consistently outperforms standard \QAOA\ (on average), except at $p=8$. \fe{We also see that at fixed depth, the performance of both standard \QAOA\ and \QAOAw\ degrades as the number of nodes increases, while the degradation of \QAOAw\ is much flatter compared to standard \QAOA.} \tsout{When considering only positive-weighted graphs weights, these observations persist, although we see better approximation ratios across the board (for both standard \QAOA\ and \QAOAw).} We further discuss pre-processing and parameter search time for \QAOAw\ in Appendix \ref{sec:appendixPreProcess}.

% \sg{For 15\% instances the approximation ratio increases by at least 0.02, and for 10\% instances we see an improvement of around 0.1. OR MORE DRAMATIC: We find that the approximation ratios improve by a factor of 40\% on an average when the QAOA circuit is applied to the warm-started approximation. TRY TO make a cdf-style plot for the improvement (0-1).}

\begin{figure}
    \centering
    \includegraphics[scale=0.35]{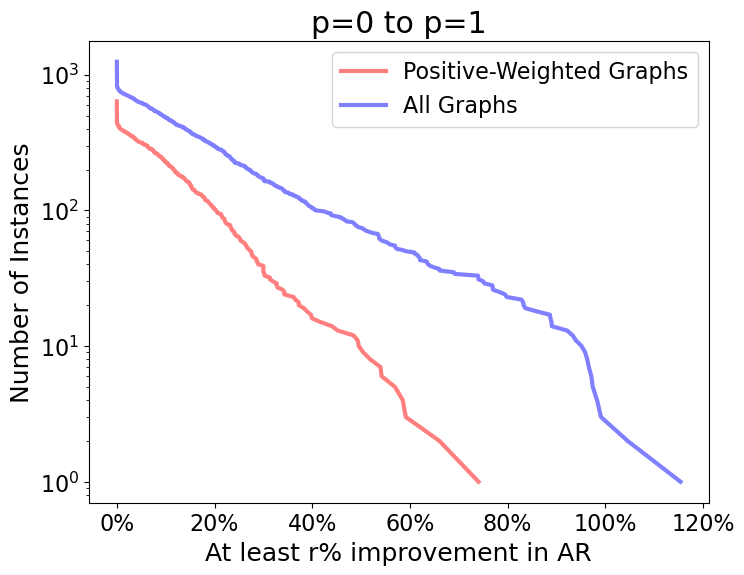}\includegraphics[scale=0.35]{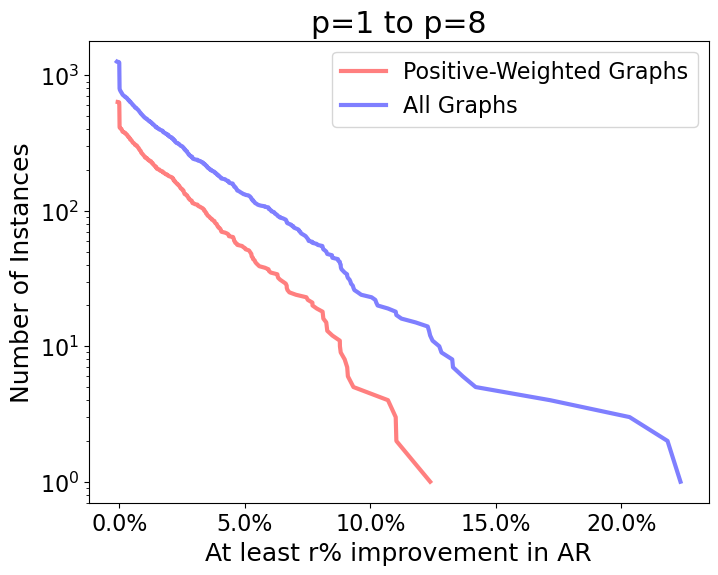}
    \caption{\fe{The number of instances for which \QAOAw{} obtained at least an $r\%$ improvement in expected AR as the circuit depth increases from $p=0$ to $p=1$ (left) and from $p=1$ to $p=8$ (right). For each instance, the best percent improvement (across all five vertex-at-top rotations) is used. Note that \% improvements in approximation ratios go up to 80-120\% from $p=0$ to $p=1$, and up to 12-20\% from as depth increases from $p=1$ to $p=8$.} }
    \label{fig:histogramP0vsP8}
\end{figure}

% \begin{figure}    
% \centering
%  \begin{tikzpicture}
%     \node[inner sep=0] (image) at (0,0) {\includegraphics[scale=0.4]{warmP0VswarmP8.png}};
%     \draw[red,ultra thick] (image.south east) -- (image.north west);
%     \draw[red,ultra thick] (image.north east) -- (image.south west);
%     \draw[red,ultra thick] (image.south west) rectangle (image.north east);
% \end{tikzpicture}
%  \begin{tikzpicture}
%     \node[inner sep=0] (image) at (0,0) {\includegraphics[scale=0.4]{warmP1VswarmP8.png}};
%     \draw[red,ultra thick] (image.south east) -- (image.north west);
%     \draw[red,ultra thick] (image.north east) -- (image.south west);
%     \draw[red,ultra thick] (image.south west) rectangle (image.north east);
% \end{tikzpicture}
%     \caption{\sout{This histogram demonstrates that the quantum aspect of \QAOAw\ is showing improvement beyond the classically-obtained initialization at $p=0$; in particular, we plot the difference in approximation ratio between depth-0 and depth-8 \QAOAw. The plot (in blue) is generated using the graphs in our graph library $\mathcal{G}$ (see Section \ref{subsec:experimentalSetup}); we also evaluate the subset of graphs in $\mathcal{G}$ with positive edge weights (in red). Overlapping regions of the histograms are in purple. Red labels along the horizontal axis denote different-sized bins from the rest of the histogram.}}
%     \label{fig:histogramP0vsP8}
% \end{figure}

% \begin{figure}
%     \centering
%     \includegraphics[scale=0.4]{Alt AR/warmP0VswarmP1_altAR.png}
%     \caption{Caption}
%     \label{fig:my_label}
% \end{figure}

\begin{figure}
    \centering
    \begin{tabular}{ccc}
    All Graphs & Positive-Weighted Graphs\\
    \includegraphics[scale=0.33]{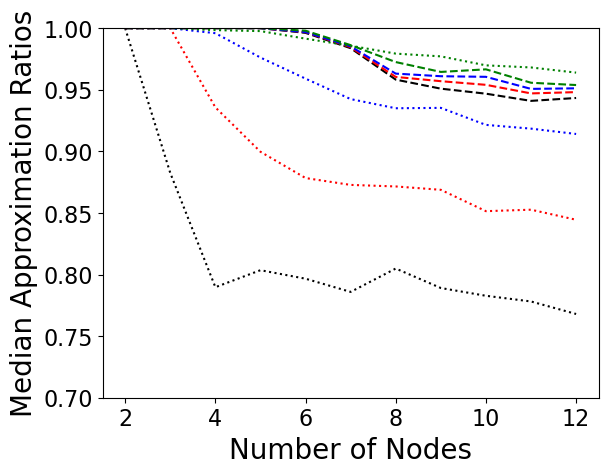}&\includegraphics[scale=0.33]{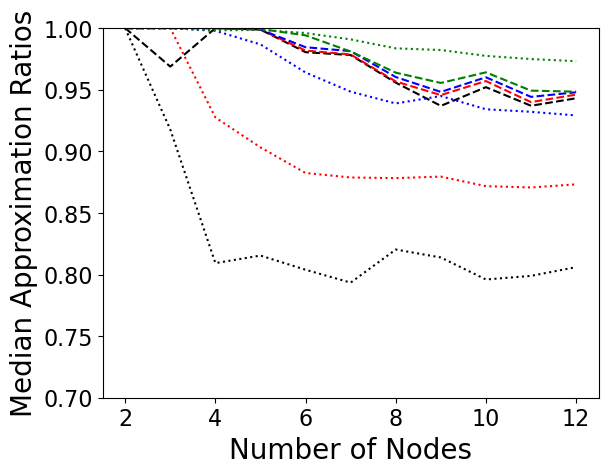}&\raisebox{1cm}{\includegraphics[scale=0.33]{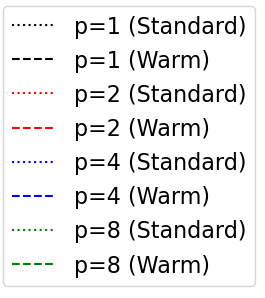}}
    \end{tabular}
    \caption{This figure shows how standard \QAOA\ (dotted) and \QAOAw\ (solid) perform as we alter the circuit depth and the number of nodes. For \QAOAw, we take the best of 5 vertex-at-top rotations. For the left plot, for each $n=2,\dots,\fe{12}$, we find the approximation ratio achieved for both standard \QAOA\ and \QAOAw\ for each $n$-node instance in $\mathcal{G}$ (see Section \ref{subsec:experimentalSetup}), and take the median of those approximation ratios. The right plot is constructed similarly except only instances in $\mathcal{G}$ with positive edge-weights are considered. We plot the results for circuit depths $p=1,2,4,8$.}
    \label{fig:performanceAcrossNandP}
\end{figure}

\subsection{Parameter Landscapes and Trajectories}
\label{subsec:landscapes}
\begin{figure}[t]
	\centering
	\begin{tabular}{cc}
	\begin{tabular}{ccc}
	    &\hspace{0.5cm} {\small No Warm Start} & \hspace{0.5cm}{\small BM-MC$_3$, Vertex 2 to Top}\\
	    
	  	\raisebox{0.5cm}{\includegraphics[scale=0.2]{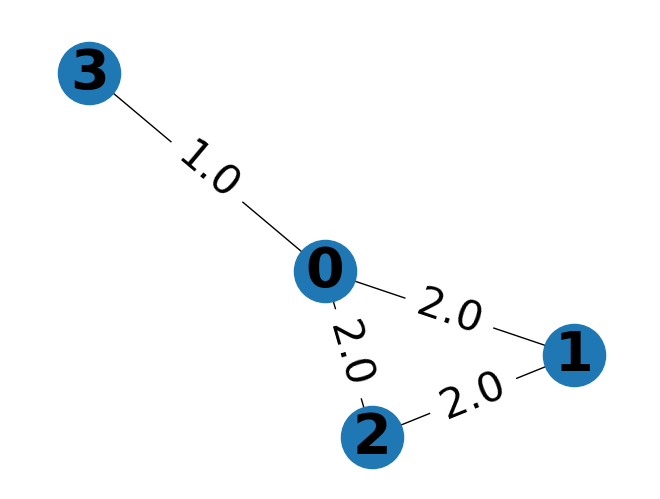}}  & 	\includegraphics[scale=0.33]{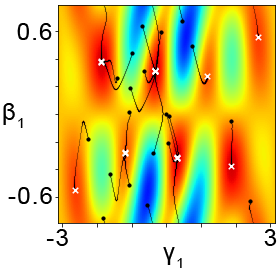} &\includegraphics[scale=0.33]{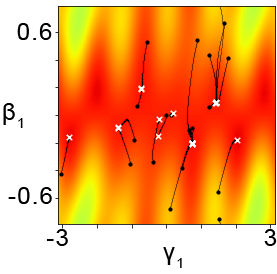}
	    	\\
	       
	       & \hspace{0.5cm}{\small BM-MC$_3$, Vertex 1 to Top}  & \hspace{0.5cm}{\small BM-MC$_3$, Uniform Rotation}\\
	       
	       \raisebox{0.5cm}{\includegraphics[scale=0.5]{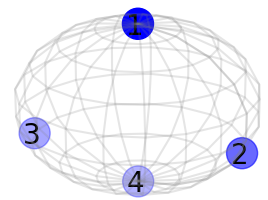}} & \includegraphics[scale=0.33]{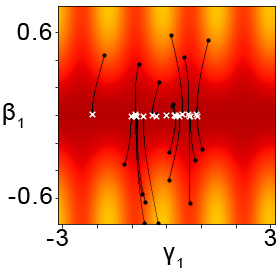} &\includegraphics[scale=0.33]{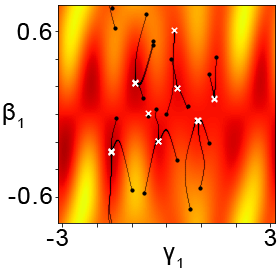} 
	\end{tabular}
	\raisebox{-1.9cm}{\includegraphics[scale=.4]{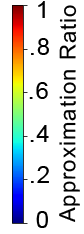}}
	\end{tabular}
	\caption{Parameter landscapes for $\hat{G}$ (top-left) with corresponding SDP solution (bottom-left). For each trajectory of optimization of the variational parameters, we use a black circle to denote the beginning of the trajectory and a white $\times$ to denote the end of the trajectory. When no warm start is used, there are many peaks and valleys (top-center). When vertex 1 rotated to the top; we have a ridge-like landscape with the optimal solutions occurring on the horizontal line $\beta_1=0$ (bottom-center). When rotating vertex 2 at the top instead, the parameter landscape is less ridge-like and the endpoints of the trajectories are more scattered (top-right). When using a uniform rotation we have peaks and valleys similar to when no warm-start was used but with overall better solution qualities (bottom-right).}
	\label{fig:weightedGraphLandscape}
	
\end{figure}

We now consider looking at \emph{all} parameter combinations for $\gamma$ and $\beta$ in order to obtain a better understanding of the landscape that we need to optimize over for standard \QAOA\ and \QAOAw. For any graph $G$, initial state $\ket{s_0}$, and circuit depth $p=1$, we can plot a  \emph{parameter landscape} which allows us to visualize the solution quality as a function of the variational parameters $\gamma_1$ and $\beta_1$. In particular, each point $(\gamma_1,\beta_1)$ in the landscape is assigned a color which corresponds to the approximation ratio (i.e. the quantity \fe{$\frac{F_1(\gamma,\beta) - \text{\sc Min-Cut(G)}}{\text{\mc}(G) - \text{\sc Min-Cut(G)}} $}).

As an example, we plot the parameter landscape for graph $\hat{G}$ in Figure \ref{fig:weightedGraphLandscape} without and with warm-starts (using 2 vertex-at-top rotations and one uniform rotation). For each parameter landscape, we ran the \QAOA\ training loop twenty times with random initializations of $(\gamma_1,\beta_1)$ and overlayed the trajectories of the parameter values throughout the training loop for the variational parameters. When no warm-start is used, the parameter landscape has many peaks and valleys and a wide range of solution qualities; using a warm-start drastically changes the landscape. However, if we rotate one of the approximate solution of BM-MC$_{3}$ for $\hat{G}$ using a vertex-at-top rotation, this yields a ridge-like parameter landscape where the optimal parameter values lie near the line $\beta_1 = 0$. This behavior is no longer there for a different vertex-at-top rotation for the same approximation solution. The endpoints of the optimization trajectories on the resultant are scattered, and the ridge-like shape is not as pronounced.
When performing a uniform rotation, the globally optimal solution qualities are comparable to the solution qualities when rotating vertex 1 to the top; however, the landscape retains some less symmetric peaks and valleys and some of the trajectories end at local optima that are far from optimal.

Overall, we see that the rotation used in the preprocessing stage can have a considerable effect on both the shape of the landscape and the solution qualities. Ideally, with a good choice of rotation, the parameter landscape has a ridge-like shape with high solution qualities near the line $\beta_1 = 0$,  in which case, $\gamma=\beta=\mathbf{0}$ is a natural choice of initialization when running \QAOAw. %We will later see (in Section \ref{subsec:approxBounds}) the ``most ideal" version of this flat, ridge-like shape when discussing initial states with an antipodal structure.

To quantify flatness of the parameter landscapes when using warm-starts, we consider some simple aggregate statistics of the landscapes of all unit-weight graphs\footnote{Due to the symmetries in the \QAOA\ circuit for unit-weight graphs, we know that it suffices to check the values of $F_p(\gamma,\beta)$ for $(\gamma,\beta)$ in $[-\pi,\pi] \times [-\pi/4, \pi/4]$ \cite{ZWCPL19}.} in $\mathcal{G}$. For each graph, we view each point in the parameter landscape as producing a cut with approximation ratio in [0,1]. We compute the minimum, maximum, and average approximation ratios found across each landscape\footnote{The minimum, maximum, and average are computed by considering a discretization of the landscape. In particular, we consider the values of $F_1(\gamma_1,\beta_1)$ for all $(\gamma_1,\beta_1) \in \mathcal{D} = \{(\pi\frac{ i}{50},\frac{\pi}{4}\frac{ j}{  50}) : i=-50,-49,\dots,50 \text{ and } j = -50,-49,\dots,50\}$.}. As shown in Figure \ref{fig:landscapeStatistics}, \QAOAw\ landscapes have lower range of approximation values, e.g., \fe{80.4\%} of the instances have a range of at most 0.4 in the approximation values attained in the landscape. This means that any two choices of $\gamma_1, \beta_1$ parameters will produce solutions with a difference in approximation of at most 0.4. In contrast, only \fe{27.5\%} of our graph instances have such a range of approximation factors for the standard \QAOA. We further see that when we use warm-starts, the overall quality of approximation across the parameter landscape improves. This can be seen by observing a higher minimum, maximum, and average approximation ratios than standard \QAOA.

\begin{figure}
    \centering
    \includegraphics[scale=0.4]{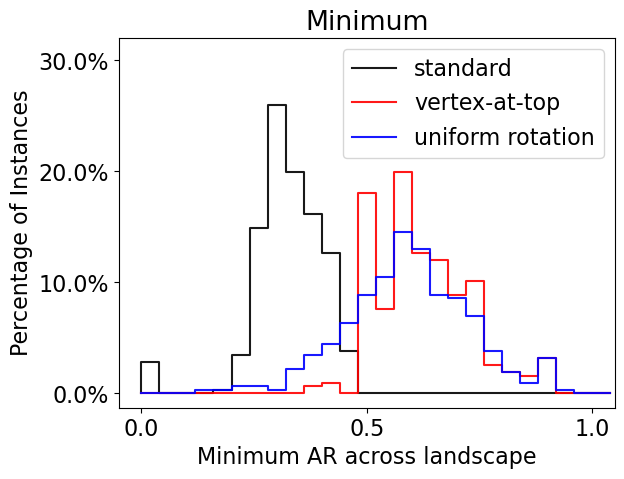}\includegraphics[scale=0.4]{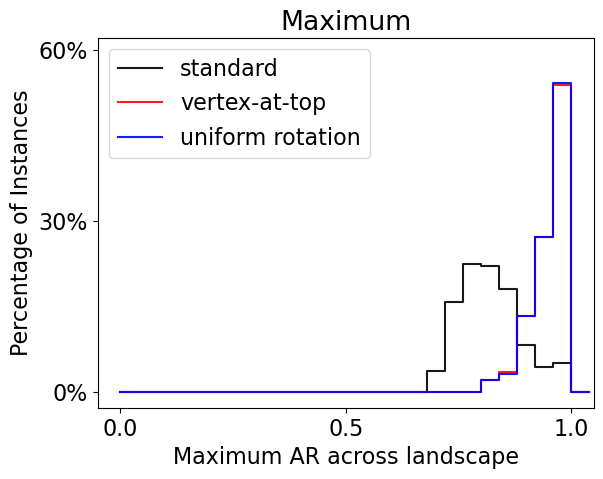}\\
    \includegraphics[scale=0.4]{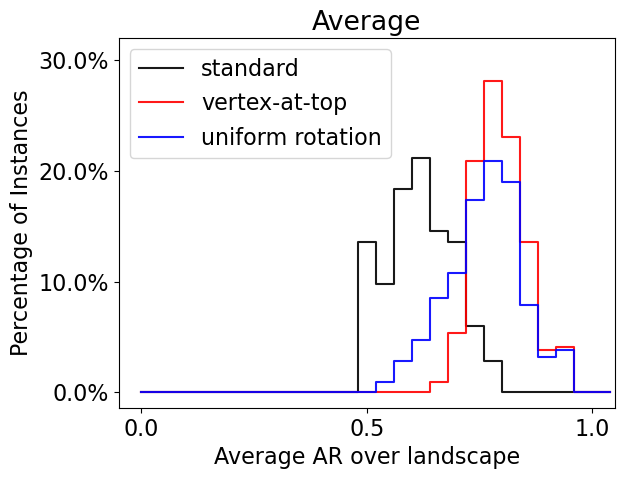}
    \includegraphics[scale=0.4]{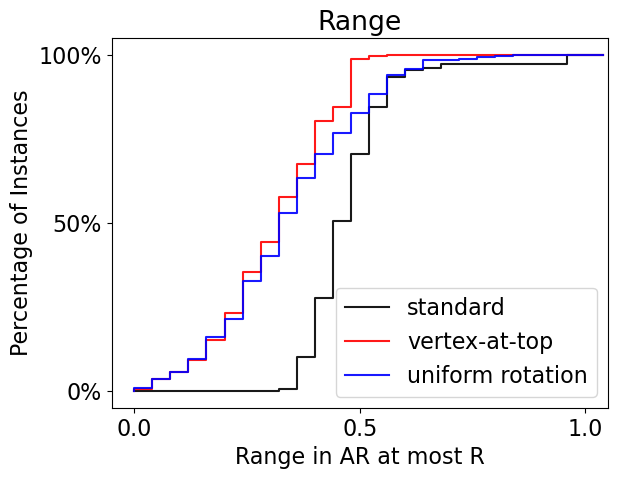}
    \caption{%\ch{fonts too small; don't label every tick?}\rt{(Reuben): Fixed.} 
    This figure shows how various statistics of the parameter landscape change with the variant of \QAOA\ considered (standard \QAOA, \QAOAw\ with vertex-at-top rotations, and \QAOAw\ with random rotations). For each unit weight graphs in our graph library $\mathcal{G}$ (See Section \ref{subsec:experimentalSetup}) and for each \QAOA\ variant, we first generate the parameter landscape; we use a single rank-2 initialization for both rotation schemes considered for \QAOAw. For each landscape, we calculate the minimum, maximum, and average across the landscape in addition to the range (the difference between the highest and lowest approximation ratio achieved in the landscape).}
    \label{fig:landscapeStatistics}
\end{figure}

\section{Theoretical Bounds}\label{sec:theory}

In this section, we theoretically analyze \QAOAw{} and demonstrate its strengths and weaknesses compared to standard \QAOA. The literature on provable approximation guarantees for \QAOA\ is sparse. In 2014, Farhi et al.\ \cite{FGG14} proved a $0.6924$ approximation for $3$-regular graphs at $p = 1$; for triangle-free $d$-regular graphs, Wang et al. \cite{WHJR18} demonstrated that depth-1 QAOA achieves an approximation ratio of at least $\frac{1}{2}\left(1+\frac{1}{\sqrt{e(d+1)}}\right)$.  \fe{Wurst and Love extend Farhi et al.'s result to show that depth-2 \QAOA{} achieves a 0.7559-approximation on 3-regular graphs (and depth-3 \QAOA{} achieves a 0.7924-approximation ratio on 3-regular graphs under some conjectures).} \fe{For higher circuit depths and for more general graphs, not  \tsout{Not} much else} is known about \QAOA\ approximation \fe{bounds}. 

Our results add to this narrative. We first show that, at $p=0$ (i.e. before any gates are applied beyond initialization), \QAOAw\ \fe{on graphs with non-negative edge-weights} achieves at least $0.75\kappa$ and $0.66\kappa$ the \mc{} when using a $\kappa$-approximate BM-MC$_2$ and BM-MC$_3$ solution\footnote{\fe{In this context, we say that a local BM-MC$_k$ solution $x$ is $\kappa$-approximate if the ratio between the BM-MC$_k$ objective at $x$ and the \mc{} is $\kappa$. Note that if $x$ is globally optimal with respect to the BM-MC$_k$ objective, then $\kappa \geq 1$ as the BM-MC$_k$ objective is a relaxation of the \mc{} objective.}} (which may correspond to distributions over cuts) respectively; \fe{for $\kappa > 2/3$ and $\kappa > 3/4$ (for rank-2 and rank-3 respectively), this results in an improvement from the $1/2$-approximation provided by standard \QAOA\ at $p=0$. }\fe{Though the worst-case results on approximation ratios for Burer-Monteiro relaxation give trivial bounds\footnote{\fe{In the worst-case, it is known that a local optimum (up to second order) for a rank $k$ formulation of BM-MC$_k$ is at least a $\lambda = 1-\frac{1}{k-1}$ approximation of the rank-$n$ SDP relaxation for graphs with non-negative edge-weights (Theorem \ref{thm:mei}) \cite{MMMO17}.}} for rank-2 and rank-3 solutions, $\kappa$ could be much higher in practice (e.g., in our simulations, we observed $\kappa \geq 0.999$ for all positive-weighted instances for BM-MC$_3$, the same can be said for BM-MC$_2$ with the exception of 19 instances with the smallest $\kappa$ observed being $\kappa = 0.833$). }

Next, we discuss \QAOAw's performance in the case where the initialization has a particular antipodal structure. We prove that such structures naturally arise when considering locally approximate BM-MC$_3$ solutions for (connected) even cycles. For these cases, vertex-at-top rotations recover the optimal solution. For uniform rotations on these antipodal structures, one can achieve optimality using BM-MC$_2$ solutions. 

\subsection{Approximation Bounds for \QAOAw} 
\label{subsec:approxBounds}
We first show that a biased initialization (using classical algorithms) can improve the theoretically known performance of \QAOA\ in some cases. 

Solving a rank-$n$ SDP relaxation of \mc\ is polynomial time solvable. Obtaining an optimal rank-$1$ solution is the ultimate goal and seems to be the hardest rank-constrained problem. The higher the rank, the more tractable the problem becomes. \tsout{A local optimum (up to second order) for a rank $d$ formulation is within $1-\frac{1}{d}$ of the global optima.~\cite{MMMO17}} We show if one accomplishes the harder objective in the classical phase of \QAOAw, the quantum phase will be better initiated. Our randomized mapping from a low-rank BM solution to the Bloch sphere guarantees to preserve the objective by a factor of $0.75$ and $0.66$ at $p = 0$ for rank-2 and rank-3 initialization respectively.

\begin{theorem}\label{thm:approx}
	Let $G$ be a graph with non-negative edge weights. If $\bx$ is a $\kappa$-approximate solution to {\sc BM-MC}$_3$ (for $G$) in $3$-dimensions, (randomized) initialization of QAOA with $R_U(\mathbf{x})$ has (expected) performance guarantee of $0.66\kappa$ at $p = 0$, i.e., only using quantum sampling with zero circuit depth for QAOA. Similarly, if $\bx$ is a $\kappa$-approximate solution of {\sc BM-MC}$_2$ (for \mc of $G$) in $2$-dimensions, initialization of QAOA with \fe{$R_U(\mathbf{x})$} has an expected performance guarantee of $0.75\kappa$ at $p=0$.
\end{theorem}
\begin{proof}
	We start by proving the $2/3$ performance of a randomized mapping from BM-MC$_3$ to the Bloch sphere.
	Let $F_0^\prime = F_0^\prime(\gamma,\beta)$ be the expected value of {\sc Max-Cut} obtained by quantum sampling (i.e., QAOA for $p=0$) \tsout{and let $x^*$ be a (globally) optimal solution to BM-MC$_3$}. Then,
	\begin{align*}
		&\frac{F_0^\prime}{\text{\mc}(G)}\ge \kappa \cdot \frac{F_0^\prime}{\text{BM-MC}_3(\bx)}
		\tag{$\text{since {\sc BM-MC}}_3(\bx) \tsout{\ge {\kappa \text{BM-MC}_3(x^*)}} \ge \kappa \text{\mc}(G)$} \\
		&\ge \kappa \min_{(i,j) \in E} \frac{\mathbb{E}[\mathbf{1}[i \text{ and } j \text{ have different spins}]]}{\frac{1}{4} \|\bx_i - \bx_j\|^2} \tag{{$ \frac{a+c}{b+d} \geq \min(\frac{a}{b}, \frac{c}{d})$ for $a,b,c,d \geq 0$}; {$w_{ij}$'s cancel}}.
	\end{align*}
	
	%\gm{This is very unclear- too much going on in one sentence:} \sg{check now}\rt{Looks good} 
	It suffices to lower bound the ratio between edge-wise contribution from quantum sampling versus edge-wise contribution to the semi-definite objective (which upper bounds the BM-MC$_k$ denominator). Instead of rotating the sphere, we can choose a random direction $w \in S^2$ to correspond to the positive spin of the Bloch sphere. Consider an edge $e = (i, j) \in E$ whose endpoints are at angles $\alpha$ on $S^2$ with respect to $\mathbf{x}$, i.e., $\bx_i \cdot \mathbf{x}_j = \cos{\alpha}$. Let $\theta_1$ and $\theta_2$ correspond to angles from $x_i$ and $x_j$ to the positive spin $w$ of the (rotated) sphere. We can write
	\begin{align*}
		&\min_{(i,j) \in E} \frac{\mathbb{E}[\mathbf{1}[i \text{ and } j \text{ have different spins}]]}{\frac{1}{4} \|\bx_i - \bx_j\|^2}
		\geq
		\min_{\alpha \in [0,\pi]} \frac{f(\theta_1,\theta_2)}{ \sin^2(\alpha/2)},
	\end{align*}
	where
	$$f(\theta_1,\theta_2) = \cos^2\frac{\theta_1}{2}\sin^2\frac{\theta_2}{2} + \cos^2\frac{\theta_2}{2}\sin^2\frac{\theta_1}{2},$$
	where we replaced $\mathbb{E}[\mathbf{1}[i \text{ and } j \text{ have different spins}]]$ by a sum of probabilities of the two cases corresponding to assignment of different spins to $i$ and $j$, formulated considering the state is a product state and observing that $\|x_i - x_j\|^2 = 2 - 2 \cos(\alpha) = 4 \sin^2(\theta/2)$. We can rewrite the above as
	\begin{align*}
		&\min_{\alpha \in [0,\pi]} \frac{ \mathbb{E}_{\theta_1, \theta_2 | \alpha} [g_1(\theta_1)g_2(\theta_2) + g_2(\theta_1)g_1(\theta_2)]}{ 2 (1-\cos (\alpha))}
		= \min_{\alpha \in [0,\pi]} \frac{ \mathbb{E}_{\theta_1, \theta_2 | \alpha} [1-\cos(\theta_1)\cos(\theta_2)]}{  (1-\cos (\alpha))},
	\end{align*}
	where $g_1(\theta) = 1+\cos(\theta)$
	and
	$g_2(\theta) = 1-\cos(\theta).$

	To further simplify notation of our optimization problem let us assume that instead of rotating $x_i$ and $x_j$ and sampling with respect to a spin direction, we randomly choose the positive spin pivot $w$ such that the $z$-axis is now rotated to be at $w\in S^2$. Without loss of generality, assume $x_i = (1,0,0), x_j = (\cos \alpha , \sin \alpha, 0)$ and $w = (\cos \theta,\sin \theta \cos \phi, \sin \theta \sin \phi) \in S^2$ is uniformly sampled from the sphere. Let
	$$h(\theta,\phi,\alpha) =  \cos \theta (\cos \alpha \cos \theta + \sin \alpha \sin \theta \cos \phi ) .$$
 This give us the following: 
	\begin{align*}
		&\min_{\alpha \in [0,\pi]} \frac{ \mathbb{E}_{\theta_1, \theta_2 | \alpha} [1-\cos(\theta_1)\cos(\theta_2)]}{  (1-\cos (\alpha))}\\
		&= \min_{\alpha \in [0,\pi]} \frac{1 - \frac{1}{4\pi} \int_0^{\pi} \int_0 ^{2\pi}   h(\theta,\phi,\alpha)  \sin \theta d \phi d \theta }{1-\cos \alpha}
		\\
		&= \min_{\alpha \in [0,\pi]} \frac{1 - \frac{1}{2} \cos \alpha \int_0^{\pi} \sin \theta \cos^2 \theta d \theta}{1-\cos \alpha}
		\\
		&= \min_{\alpha \in [0,\pi]} \frac{1 - \frac{\cos \alpha}{2} \left[ \frac{-1}{3} \cos^3 \theta \right]_0^{\pi}}{1-\cos \alpha}
		= \min_{\alpha \in [0,\pi]} \frac{1 - \frac{\cos \alpha}{3} }{1-\cos \alpha}
		= \frac{2}{3}. 
	\end{align*}
	This finishes the proof for BM-MC$_3$.
	
	Recall that for BM-MC$_2$, we perform a uniformly at random rotation along a unit circle on the Bloch sphere passing through $\ket{0}$ and $\ket{1}$.
	The proof is similar to the rank $k=3$ case, and easier. It suffices to lower bound the following ratio: 
	\begin{align*}
		&\min_{(i,j) \in E} \frac{\mathbb{E}[\mathbf{1}[i \text{ and } j \text{ have different spins}]]}{\frac{1}{4} \|\bx_i - \bx_j\|^2}
		\geq
		\min_{\alpha \in [0,\pi]} \frac{f(\alpha)}{ \sin^2(\alpha/2)},
	\end{align*}
	by $0.75$. Here $f(\alpha)$ denotes the probability that two unentangled qubits with (angular) distance $\alpha$ over the sphere/circle, are measured by opposite spins. Similar as in previous proof we can simplify the ratio as
	\begin{align*}
		\min_{\alpha \in [0,\pi]} \frac{ \mathbb{E}_{\theta_1, \theta_2 | \alpha} [1-\cos(\theta_1)\cos(\theta_2)]}{  (1-\cos (\alpha))},
	\end{align*}
	where
	$\theta_1$ and $\theta_2$ are the angles between two vertices and the pivot.

	Again we can think of vertices to be fixed over the sphere and randomly rotate the $\ket{1}$ pivot. Without loss of generality, let $x_i = (1,0)$ and $x_j = (\cos \alpha, \sin \alpha)$. The random pivot can be formulated as $(\cos \theta, \sin \theta)$ where $\theta$ is uniformly distributed over $[0,2\pi)$. We can write
	$\theta_1 = \theta$ and $\theta_2 = \theta-\alpha$. The target ratio can be written as
	\begin{align*}
		&\min_{\alpha \in [0,\pi]} \frac{ 1 - \frac{1}{2\pi}\int_0^{2\pi} \cos(\theta)\cos(\theta-\alpha) d\theta}{  1-\cos (\alpha)}\\
		= & \min_{\alpha \in [0,\pi]} \frac{ 1 - \frac{1}{2\pi}\int_0^{2\pi} \cos^2(\theta) \cos(\alpha) d \theta + 0 }{  1-\cos (\alpha)}\\
		= & \min_{\alpha \in [0,\pi]} \frac{ 1 - \frac{1}{2} \cos(\alpha) }{  1-\cos (\alpha)}
		= \frac{3}{4}.
	\end{align*}
\end{proof}

\fe{In addition to preserving a provable fraction of the classical BM-MC$_k$ objective, \QAOAw{} can provably outperform standard \QAOA{} on certain families of graphs. In particular, for even cycles, we show that any locally optimal BM-MC$_3$ is also a globally optimal rank-1 solution. Hence the preprocessing stage yields a warm-start that is simply a collection of antipodal points corresponding to the \mc{}, which can easily be recovered with a suitable rotation. On the other hand, depth-$p$ standard \QAOA{} is only able to achieve an approximation ratio at most $(2p+1)/(2p+2)$ \cite{FGG14,WHJR18}. More details regarding the behavior of \QAOAw{} on warm-starts with antipodal structures can be found in Appendix \ref{sec:antipodalStructures}. }

\subsection{Limitations of \QAOAw}

\label{subsec:nonoptimalityQAOAw}
We now look at some of the limitations of \QAOAw. Specifically, we observe a decrease in performance due to warm-starts close to the eigenstates of the mixer with zero eigenvalue. %We also briefly discuss \QAOAw{}'s behavior on disconnected graphs.

\subsubsection{\QAOAw\ at High Circuit Depth} 

In the case of the standard initialization for \QAOA, we know that with the optimum choice of parameters $\gamma,\beta$, \tsout{one can find the \mc\ of $G$ in the limit as} \fe{the probability of sampling the \mc\ approaches 1 as} the circuit depth \tsout{$p \to \infty$} \fe{$p$ approaches infinity}. This is not the case for \QAOAw:

\begin{theorem}\label{thm:warmstart}
	There exist initializations of \QAOAw\ \tsout{that prevent the algorithm from obtaining the optimal solution, even as the circuit depth $p \to \infty$}\fe{for which the probability of measuring the \mc\ does not approach 1 as the circuit depth $p$ tends to infinity.}
\end{theorem}

\begin{proof}[Proof idea] Let $G$ be a graph on two vertices connected by an edge of unit weight. Suppose that we run \QAOAw\ starting with the state $\ket{s} := \ket{u} \otimes \ket{v}$ where $\ket{u} := \ket{+} = \frac{1}{\sqrt{2}}(\ket{0}+\ket{1})$ and $\ket{v} := \ket{-} = \frac{1}{\sqrt{2}}(\ket{0}-\ket{1})$. Note that $\ket{s}$ corresponds to an optimal solution to BM-MC$_3$.
	
	For $p=0$, it is easy to show that we obtain, in expectation, 50\% of the \mc\ of $G$. After the cost term of the circuit is applied, the resulting state is unaffected by the mixing term (since the resulting state is an eigenstate of the mixing Hamiltonian with eigenvalue zero) and thus there is no change in measurement. Even higher circuit depths have no effect in driving the system out of the eigenstate and thus $F_p(\gamma,\beta) = 0.5$ for all $p$ and any $\gamma,\beta$, i.e., \QAOAw\  (initialized with $\ket{s}$) obtains only 50\% of the \mc\ of $G$ \fe{(in expectation)} regardless of circuit depth or choice of variational parameters.
\end{proof}
We include a complete proof in Appendix \ref{sec:appendixproofs}. The previous theorem shows that \QAOAw\ may perform poorly on \emph{specific} states, however we next discuss that this behavior is consistent across slight perturbations around this state as well. In Figure \ref{fig:oneEdgePerturbation}, at any point $(\phi, \theta)$ we depict the percentage of \mc\ obtained using the optimal choice of variational parameters if the initial state of the first qubit is given by the polar and azimuthal angles $\theta$ and $\phi$ and the second qubit is diametrically opposed. Note that the optimal \mc\ \tsout{can only be} \fe{is} achieved \fe{with probability 1 only} when both vertices lie in the $yz$-plane. The worst case occurs when the vertices lie on the $x$-axis; this is consistent with Theorem \ref{thm:warmstart}. %and Lemma \ref{lem:halfApproximation}. 
In general, \fe{in expectation,} there is a larger gap to optimality the closer the solution $\mathbf{x}$ is to the $x$-axis; which suggests that it is reasonable to embed the approximate solutions of BM-MC$_2$ in the $yz$-plane of the Bloch sphere (as done in our preprocessing stage). Lastly, we believe that this behavior is consistent at larger circuit depths as well. 

\begin{figure}[t]
	\centering
	\includegraphics[scale=0.5]{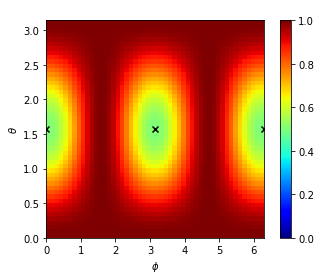}\hspace{0.1cm}\rotatebox{90}{\hspace{1cm} {\small Approximation Ratio}}
	\caption{A plot of the percentage of the \mc\ achieved with \QAOAw\ (when the optimal variational parameters are chosen) with $p=1$ for a one-edge graph $G$ at various starting states $\ket{s_0} = Q(x)$ where one point of $x$ has polar angle $\theta$ and azimuthal angle $\phi$ and the remaining point is diametrically opposed. The starting states that perform the worst, i.e. $\ket{+-}$ and $\ket{-+}$, are marked with a black $\times$. For each point in the figure, the optimal variational parameters were estimated by performing a dense grid-search over the variational parameter space.}
	\label{fig:oneEdgePerturbation}
\end{figure}

%\gm{I suggest getting rid of this subsubsection. QAOA doesn't help on disconnected graphs in general, it will never come up}
%\subsubsection{\QAOAw\ on Disconnected Graphs}The objective of BM-MC$_k$ is invariant under rotation of components; thus, when preparing the initial state, one should choose a suitable rotation for \emph{each} component, otherwise, \QAOAw\ may be unable to align each component \emph{simultaneously} to the measurement axis. However, in practice, noisy intermediate-scale quantum (NISQ) devices are only able to work with a limited number of qubits and thus, it is preferred to run \QAOA\ on each component of a disconnected graph \emph{separately} instead of running \QAOA\ on the entire graph.

\section{Discussion}\label{sec:discussion}
In this work, we proposed using classical approximate solutions to low-rank \mc\ formulations to initialize the \QAOA\ algorithm. There are significant differences in classical approximation algorithms for \mc\ and quantum algorithms. For example, in the classical approach the vertices that share the same 3-dimensional representation on the sphere will always be on the same side of the cut (no matter which hyperplane is selected). In contrast, quantum sampling creates a very different distribution (with a larger support) over cuts, wherein vertices with the same state can be sampled on different sides of the cut. Despite this difference, we observe that as the angle $\theta$ of the vertices to the measurement axis approaches 0, the probability distribution of the classical solution approaches that of the quantum sampling. Intuitively, as vertices start clustering at the antipodes on the 3-dimensional sphere, quantum sampling of the corresponding qubits and hyperplane rounding of the 3-dimensional representation both give similar cuts. Moreover, SDP-based solutions spread adjacent vertices (with positive edge weights) as far as possible on the $k$-dimensional sphere, which can be beneficial for quantum sampling as well. %Therefore, warm-starts allow us to partially solve the whole graph \emph{before} starting the quantum computations, and we suspect that the bounds on the depth of QAOA to ``see" the whole graph for optimal solutions in \cite{FGG20} may not hold for \QAOAw.

Standard \QAOA\ is a local algorithm~\cite{FGG14}. If the circuit depth $p$ is not high enough, then standard \QAOA\ may fail to achieve near-optimal solutions \cite{FGG20,BKKT19}. However, when one considers the preprocessing stage used in \QAOAw, such a locality property no longer exists. % due to the fact that the BM-MC$_k$ formulation does not have such locality properties. 
A clear example of this is BM-MC$_2$ applied to an odd cycle: the optimal solution consists of the vertices evenly spaced apart along the unit circle. However, if a single edge is deleted, the optimal solution collapses to a rank-1 solution. The edge deletion has a global effect on the positions of all the vertices, and consequently, on the probability of each edge being cut. Put another away, although the quantum operations in \QAOAw\ are still local, the warm-start encodes information about the global structure of the graph, in which case, building up correlations between distant qubits (via a high circuit depth) may not be necessary if a high-quality warm-start is used. %the preprocessing stage encodes information about the global structure of the graph which has the potential to allow \QAOAw\ to achieve desirable approximation ratios with lower circuit depth.
%\ch{This is an interesting discussion. Does this just give the quantum portion following the warm-start a leg up? The operations are still local and will only propagate as far as the p-depth, however you won't necessarily need changes to distant qubits if the warm start state was very good.}\rt{I've edited the language to make it clear that the quantum operations for QAOA-Warm are still local so as to not mislead the reader.}

Warm-starts also appear to flatten the energy landscape in terms of $(\beta,\gamma)$. In the most extreme case (for example Figure~\ref{fig:fourCycleLandscape}(b)), the warm start finds the optimal solution, completely decoupling the QAOA optimization loop from $\gamma_1$ and the cost Hamiltonian $H_C$. Even when this does not occur, warm-starting still appears to make \QAOA\ less sensitive to initial $(\beta,\gamma)$ values by starting off in the neighborhood of a possible solution. In particular, the role of $\gamma$ is diminished, as the warm-start has already begun optimizing the cost-energy. This suggests that \QAOAw\ serves as a kind of dimensional reduction, emphasizing the amplitude manipulation of the mixer over the energy weighting of the cost Hamiltonian. This is not a guarantee that the \QAOA\ optimization will find the optimal solution in the reduced space; the reduction may hide the optimal solution for graphs that are especially challenging for SDP solvers. However, this flattening may prove important for physical implementations of \QAOA. The warm-start flattened landscapes may make QAOA more robust to both classical and quantum noise that would otherwise complicate the optimal solution search.

In this work, we restricted our attention to rank-2 and rank-3 initializations, whereas in classical methods, one could also make an attempt at finding rank-$k$ ($k>3$) solutions. These solutions are easier to find, and yield provably better approximations as $k$ increases \cite{MMMO17}. However, increasing the number of dimensions makes the mapping to the quantum states non-trivial. Exploration of higher-rank approximations are left as a future research direction. 

Another direction for future work is to apply \QAOAw\ to other combinatorial problems. One path is reduction of other problems in NP to \mc\ \cite{Karp1972}. Alternately, Quadratic Unconstrained Binary Optimization (QUBO) problems can easily be recast as a \mc\ problem (and vice versa) with the number of variables differing by at most 1 \cite{dunning2018}. Many combinatorial problems are easily expressed in the form of a QUBO \cite{GKD19,lodewijks2020} so \mc\ is also an interesting problem to consider from a practical standpoint. However, it may be of interest to see if our approach can be used directly for other combinatorial problems without resorting to such reductions.

\tsout{Another open question is QAOA performance on mixed-weight graphs (graphs with both positive and negative edge weights). In Figure \ref{fig:histogramWarmVsStandard}, the graphs with the largest gap in approximation ratios between depth-0 and depth-8 \QAOAw\ were all mixed weight graphs. Similar phenomena can be observed in Figures \ref{fig:histogramP0vsP8} and \ref{fig:eggerResults} as well. Note that the 0.878-approximation ratio for GW does not hold for mixed weight graphs. In fact, the only classical approximation factor known for mixed weight graphs is a $O(1/\log(n))$, when the sum of all the edge-weights is positive \cite{CW04}. Moreover, for mixed weight graphs where the sum of the positive edge weights and the sum of (absolute values of) negative edge weights are similar, then standard \QAOA\ begins with an approximation ratio of approximately zero. As a future direction, better understanding of the instances for which \QAOAw\ performs better than standard \QAOA\ for both positive-weighted and mixed weight graphs would be interesting.}

Lastly, as seen in Section \ref{subsec:nonoptimalityQAOAw}, we acknowledge that \QAOAw, in its current form, has limitations; in particular, increased circuit depth does not necessarily yield optimality of \mc\ in the limit. Nonetheless, we believe that \QAOAw\ is a promising approach since, even at low circuit depth, it is able to start with relatively high approximation ratios (compared to standard \QAOA). This performance may be extendable to higher circuit depth via modifications to the mixing Hamiltonian $H_B$; this idea yields positive results in Egger et al.'s work and we believe that mixer modifications may be beneficial to \QAOAw\ as well.

\section{Conclusion}
\label{sec:conclusion}
We explored the idea of \emph{warm-starts} for initializing the quantum state of the \QAOA\ algorithm, and showed promising experimental and theoretical results for low-rank initializations using approximate SDP solution. On average, we find that \QAOAw\ performs better in terms of time and quality of solutions in low depth circuits, compared to standard \QAOA. \fe{Moreover, even though a portion of the approximation ratio of \QAOAw\ can be attributed to the classical warm-start itself, we find that running \QAOAw{} introduces significant improvements in expected cut quality beyond simply (quantum) sampling the initial warm-start state for many instances.} As the circuit depth increases, \QAOAw\ is however unable to converge to the optimal solution (unlike standard \QAOA). We believe that this could be remedied by considering further modifications to \QAOAw\ (e.g. modifying the mixing Hamiltonian), \fe{although standard mixers might provide easier implementation on certain hardware. We further acknowledge that beyond the approximation ratio, there are a variety of methods and metrics in which to measure the performance of \QAOA{} and its possible variants. We leave such an exploration of the cut distributions (and metrics on those distributions) for potential future work; we refer the reader to a paper by Herrman et al. for such results on standard \QAOA\  \cite{HTOLHS21}.}

\fe{Overall, we believe that the use of the standard mixers with warm-starts allows a principled way of bringing in information from classical solvers into quantum algorithms. The concept of warm-starts and plateauing of quality of approximation at higher $p$ depth could be of interest to researchers looking at reachability of solution state space and, at the limitations and strengths of the standard QAOA itself.} 

\section{Acknowledgements}
	This material is based upon work supported by the Defense Advanced Research Projects Agency (DARPA) under Contract No. HR001120C0046.

%\bibliographystyle{ACM-Reference-Format}
%\bibliography{refs}

\def\authornoop#1{}
\begin{thebibliography}{52}%
\makeatletter
\providecommand \@ifxundefined [1]{%
 \@ifx{#1\undefined}
}%
\providecommand \@ifnum [1]{%
 \ifnum #1\expandafter \@firstoftwo
 \else \expandafter \@secondoftwo
 \fi
}%
\providecommand \@ifx [1]{%
 \ifx #1\expandafter \@firstoftwo
 \else \expandafter \@secondoftwo
 \fi
}%
\providecommand \natexlab [1]{#1}%
\providecommand \enquote  [1]{``#1''}%
\providecommand \bibnamefont  [1]{#1}%
\providecommand \bibfnamefont [1]{#1}%
\providecommand \citenamefont [1]{#1}%
\providecommand \href@noop [0]{\@secondoftwo}%
\providecommand \href [0]{\begingroup \@sanitize@url \@href}%
\providecommand \@href[1]{\@@startlink{#1}\@@href}%
\providecommand \@@href[1]{\endgroup#1\@@endlink}%
\providecommand \@sanitize@url [0]{\catcode `\\12\catcode `\$12\catcode
  `\&12\catcode `\#12\catcode `\^12\catcode `\_12\catcode `\%12\relax}%
\providecommand \@@startlink[1]{}%
\providecommand \@@endlink[0]{}%
\providecommand \url  [0]{\begingroup\@sanitize@url \@url }%
\providecommand \@url [1]{\endgroup\@href {#1}{\urlprefix }}%
\providecommand \urlprefix  [0]{URL }%
\providecommand \Eprint [0]{\href }%
\providecommand \doibase [0]{https://doi.org/}%
\providecommand \selectlanguage [0]{\@gobble}%
\providecommand \bibinfo  [0]{\@secondoftwo}%
\providecommand \bibfield  [0]{\@secondoftwo}%
\providecommand \translation [1]{[#1]}%
\providecommand \BibitemOpen [0]{}%
\providecommand \bibitemStop [0]{}%
\providecommand \bibitemNoStop [0]{.\EOS\space}%
\providecommand \EOS [0]{\spacefactor3000\relax}%
\providecommand \BibitemShut  [1]{\csname bibitem#1\endcsname}%
\let\auto@bib@innerbib\@empty
%</preamble>
\bibitem [{\citenamefont {Farhi}\ \emph {et~al.}(2020)\citenamefont {Farhi},
  \citenamefont {Gamarnik},\ and\ \citenamefont {Gutmann}}]{FGG20}%
  \BibitemOpen
  \bibfield  {author} {\bibinfo {author} {\bibfnamefont {E.}~\bibnamefont
  {Farhi}}, \bibinfo {author} {\bibfnamefont {D.}~\bibnamefont {Gamarnik}},\
  and\ \bibinfo {author} {\bibfnamefont {S.}~\bibnamefont {Gutmann}},\
  }\bibfield  {title} {\bibinfo {title} {The quantum approximate optimization
  algorithm needs to see the whole graph: {A} typical case},\ }\href@noop {}
  {\bibfield  {journal} {\bibinfo  {journal} {arXiv preprint arXiv:2004.09002}\
  } (\bibinfo {year} {2020})}\BibitemShut {NoStop}%
\bibitem [{\citenamefont {Preskill}(2018)}]{P18}%
  \BibitemOpen
  \bibfield  {author} {\bibinfo {author} {\bibfnamefont {J.}~\bibnamefont
  {Preskill}},\ }\bibfield  {title} {\bibinfo {title} {Quantum computing in the
  {NISQ} era and beyond},\ }\href@noop {} {\bibfield  {journal} {\bibinfo
  {journal} {Quantum}\ }\textbf {\bibinfo {volume} {2}},\ \bibinfo {pages} {79}
  (\bibinfo {year} {2018})}\BibitemShut {NoStop}%
\bibitem [{\citenamefont {Farhi}\ \emph {et~al.}(2014)\citenamefont {Farhi},
  \citenamefont {Goldstone},\ and\ \citenamefont {Gutmann}}]{FGG14}%
  \BibitemOpen
  \bibfield  {author} {\bibinfo {author} {\bibfnamefont {E.}~\bibnamefont
  {Farhi}}, \bibinfo {author} {\bibfnamefont {J.}~\bibnamefont {Goldstone}},\
  and\ \bibinfo {author} {\bibfnamefont {S.}~\bibnamefont {Gutmann}},\
  }\bibfield  {title} {\bibinfo {title} {A quantum approximate optimization
  algorithm},\ }\href@noop {} {\bibfield  {journal} {\bibinfo  {journal} {arXiv
  preprint arXiv:1411.4028}\ } (\bibinfo {year} {2014})}\BibitemShut {NoStop}%
\bibitem [{\citenamefont {Khot}(2002)}]{Khot02}%
  \BibitemOpen
  \bibfield  {author} {\bibinfo {author} {\bibfnamefont {S.}~\bibnamefont
  {Khot}},\ }\bibfield  {title} {\bibinfo {title} {On the power of unique
  2-prover 1-round games},\ }in\ \href@noop {} {\emph {\bibinfo {booktitle}
  {Proceedings of the thiry-fourth annual ACM symposium on Theory of
  computing}}}\ (\bibinfo {year} {2002})\ pp.\ \bibinfo {pages}
  {767--775}\BibitemShut {NoStop}%
\bibitem [{\citenamefont {Mossel}\ \emph {et~al.}(2005)\citenamefont {Mossel},
  \citenamefont {O'Donnell},\ and\ \citenamefont {Oleszkiewicz}}]{MOO05}%
  \BibitemOpen
  \bibfield  {author} {\bibinfo {author} {\bibfnamefont {E.}~\bibnamefont
  {Mossel}}, \bibinfo {author} {\bibfnamefont {R.}~\bibnamefont {O'Donnell}},\
  and\ \bibinfo {author} {\bibfnamefont {K.}~\bibnamefont {Oleszkiewicz}},\
  }\bibfield  {title} {\bibinfo {title} {Noise stability of functions with low
  influences: invariance and optimality},\ }in\ \href@noop {} {\emph {\bibinfo
  {booktitle} {46th Annual IEEE Symposium on Foundations of Computer Science
  (FOCS'05)}}}\ (\bibinfo {organization} {IEEE},\ \bibinfo {year} {2005})\ pp.\
  \bibinfo {pages} {21--30}\BibitemShut {NoStop}%
\bibitem [{\citenamefont {Khot}\ \emph {et~al.}(2007)\citenamefont {Khot},
  \citenamefont {Kindler}, \citenamefont {Mossel},\ and\ \citenamefont
  {O’Donnell}}]{KKMO07}%
  \BibitemOpen
  \bibfield  {author} {\bibinfo {author} {\bibfnamefont {S.}~\bibnamefont
  {Khot}}, \bibinfo {author} {\bibfnamefont {G.}~\bibnamefont {Kindler}},
  \bibinfo {author} {\bibfnamefont {E.}~\bibnamefont {Mossel}},\ and\ \bibinfo
  {author} {\bibfnamefont {R.}~\bibnamefont {O’Donnell}},\ }\bibfield
  {title} {\bibinfo {title} {Optimal inapproximability results for {MAX-CUT}
  and other 2-variable {CSPs}?},\ }\href@noop {} {\bibfield  {journal}
  {\bibinfo  {journal} {SIAM Journal on Computing}\ }\textbf {\bibinfo {volume}
  {37}},\ \bibinfo {pages} {319} (\bibinfo {year} {2007})}\BibitemShut
  {NoStop}%
\bibitem [{\citenamefont {Goemans}\ and\ \citenamefont
  {Williamson}(1995)}]{GW95}%
  \BibitemOpen
  \bibfield  {author} {\bibinfo {author} {\bibfnamefont {M.~X.}\ \bibnamefont
  {Goemans}}\ and\ \bibinfo {author} {\bibfnamefont {D.~P.}\ \bibnamefont
  {Williamson}},\ }\bibfield  {title} {\bibinfo {title} {Improved approximation
  algorithms for maximum cut and satisfiability problems using semidefinite
  programming},\ }\href@noop {} {\bibfield  {journal} {\bibinfo  {journal}
  {Journal of the ACM (JACM)}\ }\textbf {\bibinfo {volume} {42}},\ \bibinfo
  {pages} {1115} (\bibinfo {year} {1995})}\BibitemShut {NoStop}%
\bibitem [{\citenamefont {Trevisan}\ \emph {et~al.}(2000)\citenamefont
  {Trevisan}, \citenamefont {Sorkin}, \citenamefont {Sudan},\ and\
  \citenamefont {Williamson}}]{TSSW00}%
  \BibitemOpen
  \bibfield  {author} {\bibinfo {author} {\bibfnamefont {L.}~\bibnamefont
  {Trevisan}}, \bibinfo {author} {\bibfnamefont {G.~B.}\ \bibnamefont
  {Sorkin}}, \bibinfo {author} {\bibfnamefont {M.}~\bibnamefont {Sudan}},\ and\
  \bibinfo {author} {\bibfnamefont {D.~P.}\ \bibnamefont {Williamson}},\
  }\bibfield  {title} {\bibinfo {title} {Gadgets, approximation, and linear
  programming},\ }\href@noop {} {\bibfield  {journal} {\bibinfo  {journal}
  {SIAM Journal on Computing}\ }\textbf {\bibinfo {volume} {29}},\ \bibinfo
  {pages} {2074} (\bibinfo {year} {2000})}\BibitemShut {NoStop}%
\bibitem [{\citenamefont {H{\aa}stad}(2001)}]{H01}%
  \BibitemOpen
  \bibfield  {author} {\bibinfo {author} {\bibfnamefont {J.}~\bibnamefont
  {H{\aa}stad}},\ }\bibfield  {title} {\bibinfo {title} {Some optimal
  inapproximability results},\ }\href@noop {} {\bibfield  {journal} {\bibinfo
  {journal} {Journal of the ACM (JACM)}\ }\textbf {\bibinfo {volume} {48}},\
  \bibinfo {pages} {798} (\bibinfo {year} {2001})}\BibitemShut {NoStop}%
\bibitem [{\citenamefont {Charikar}\ and\ \citenamefont {Wirth}(2004)}]{CW04}%
  \BibitemOpen
  \bibfield  {author} {\bibinfo {author} {\bibfnamefont {M.}~\bibnamefont
  {Charikar}}\ and\ \bibinfo {author} {\bibfnamefont {A.}~\bibnamefont
  {Wirth}},\ }\bibfield  {title} {\bibinfo {title} {Maximizing quadratic
  programs: extending grothendieck's inequality},\ }in\ \href
  {https://doi.org/10.1109/FOCS.2004.39} {\emph {\bibinfo {booktitle} {45th
  Annual IEEE Symposium on Foundations of Computer Science}}}\ (\bibinfo {year}
  {2004})\ pp.\ \bibinfo {pages} {54--60}\BibitemShut {NoStop}%
\bibitem [{\citenamefont {Bertsimas}\ \emph {et~al.}(2016)\citenamefont
  {Bertsimas}, \citenamefont {King},\ and\ \citenamefont
  {Mazumder}}]{bertsimas2016best}%
  \BibitemOpen
  \bibfield  {author} {\bibinfo {author} {\bibfnamefont {D.}~\bibnamefont
  {Bertsimas}}, \bibinfo {author} {\bibfnamefont {A.}~\bibnamefont {King}},\
  and\ \bibinfo {author} {\bibfnamefont {R.}~\bibnamefont {Mazumder}},\
  }\bibfield  {title} {\bibinfo {title} {Best subset selection via a modern
  optimization lens},\ }\href@noop {} {\bibfield  {journal} {\bibinfo
  {journal} {The Annals of Statistics}\ }\textbf {\bibinfo {volume} {44}},\
  \bibinfo {pages} {813} (\bibinfo {year} {2016})}\BibitemShut {NoStop}%
\bibitem [{\citenamefont {Marcucci}\ and\ \citenamefont
  {Tedrake}(2020)}]{marcucci2020warm}%
  \BibitemOpen
  \bibfield  {author} {\bibinfo {author} {\bibfnamefont {T.}~\bibnamefont
  {Marcucci}}\ and\ \bibinfo {author} {\bibfnamefont {R.}~\bibnamefont
  {Tedrake}},\ }\bibfield  {title} {\bibinfo {title} {Warm start of
  mixed-integer programs for model predictive control of hybrid systems},\
  }\bibfield  {journal} {\bibinfo  {journal} {IEEE Transactions on Automatic
  Control}\ }\href {https://doi.org/10.1109/TAC.2020.3007688}
  {10.1109/TAC.2020.3007688} (\bibinfo {year} {2020})\BibitemShut {NoStop}%
\bibitem [{\citenamefont {Ralphs}\ and\ \citenamefont
  {G\"uzelsoy}(2006)}]{ralphs2006duality}%
  \BibitemOpen
  \bibfield  {author} {\bibinfo {author} {\bibfnamefont {T.}~\bibnamefont
  {Ralphs}}\ and\ \bibinfo {author} {\bibfnamefont {M.}~\bibnamefont
  {G\"uzelsoy}},\ }\bibfield  {title} {\bibinfo {title} {Duality and warm
  starting in integer programming},\ }in\ \href@noop {} {\emph {\bibinfo
  {booktitle} {Proceedings of 2006 NSF Design, Service, and Manufacturing
  Grantees and Research Conference}}}\ (\bibinfo {year} {2006})\BibitemShut
  {NoStop}%
\bibitem [{\citenamefont {Burer}\ and\ \citenamefont {Monteiro}(2003)}]{BM03}%
  \BibitemOpen
  \bibfield  {author} {\bibinfo {author} {\bibfnamefont {S.}~\bibnamefont
  {Burer}}\ and\ \bibinfo {author} {\bibfnamefont {R.~D.}\ \bibnamefont
  {Monteiro}},\ }\bibfield  {title} {\bibinfo {title} {A nonlinear programming
  algorithm for solving semidefinite programs via low-rank factorization},\
  }\href@noop {} {\bibfield  {journal} {\bibinfo  {journal} {Mathematical
  Programming}\ }\textbf {\bibinfo {volume} {95}},\ \bibinfo {pages} {329}
  (\bibinfo {year} {2003})}\BibitemShut {NoStop}%
\bibitem [{\citenamefont {{\authornoop{BM}}~Burer}\ and\ \citenamefont
  {Monteiro}(2005)}]{BM05}%
  \BibitemOpen
  \bibfield  {author} {\bibinfo {author} {\bibfnamefont {S.}~\bibnamefont
  {{\authornoop{BM}}~Burer}}\ and\ \bibinfo {author} {\bibfnamefont {R.~D.}\
  \bibnamefont {Monteiro}},\ }\bibfield  {title} {\bibinfo {title} {Local
  minima and convergence in low-rank semidefinite programming},\ }\href@noop {}
  {\bibfield  {journal} {\bibinfo  {journal} {Mathematical Programming}\
  }\textbf {\bibinfo {volume} {103}},\ \bibinfo {pages} {427} (\bibinfo {year}
  {2005})}\BibitemShut {NoStop}%
\bibitem [{\citenamefont {Mbeng}\ \emph {et~al.}(2019)\citenamefont {Mbeng},
  \citenamefont {Fazio},\ and\ \citenamefont {Santoro}}]{MFS19}%
  \BibitemOpen
  \bibfield  {author} {\bibinfo {author} {\bibfnamefont {G.~B.}\ \bibnamefont
  {Mbeng}}, \bibinfo {author} {\bibfnamefont {R.}~\bibnamefont {Fazio}},\ and\
  \bibinfo {author} {\bibfnamefont {G.~E.}\ \bibnamefont {Santoro}},\
  }\bibfield  {title} {\bibinfo {title} {Quantum annealing: a journey through
  digitalization, control, and hybrid quantum variational schemes},\
  }\href@noop {} {\bibfield  {journal} {\bibinfo  {journal} {arXiv preprint
  arXiv:1906.08948}\ } (\bibinfo {year} {2019})}\BibitemShut {NoStop}%
\bibitem [{\citenamefont {Zhou}\ \emph {et~al.}(2018)\citenamefont {Zhou},
  \citenamefont {Wang}, \citenamefont {Choi}, \citenamefont {Pichler},\ and\
  \citenamefont {Lukin}}]{ZWCPL19}%
  \BibitemOpen
  \bibfield  {author} {\bibinfo {author} {\bibfnamefont {L.}~\bibnamefont
  {Zhou}}, \bibinfo {author} {\bibfnamefont {S.-T.}\ \bibnamefont {Wang}},
  \bibinfo {author} {\bibfnamefont {S.}~\bibnamefont {Choi}}, \bibinfo {author}
  {\bibfnamefont {H.}~\bibnamefont {Pichler}},\ and\ \bibinfo {author}
  {\bibfnamefont {M.~D.}\ \bibnamefont {Lukin}},\ }\bibfield  {title} {\bibinfo
  {title} {Quantum approximate optimization algorithm: performance, mechanism,
  and implementation on near-term devices},\ }\href@noop {} {\bibfield
  {journal} {\bibinfo  {journal} {arXiv preprint arXiv:1812.01041}\ } (\bibinfo
  {year} {2018})}\BibitemShut {NoStop}%
\bibitem [{\citenamefont {Zhu}\ \emph {et~al.}(2020)\citenamefont {Zhu},
  \citenamefont {Tang}, \citenamefont {Barron}, \citenamefont
  {Calderon-Vargas}, \citenamefont {Mayhall}, \citenamefont {Barnes},\ and\
  \citenamefont {Economou}}]{ZTBCVMBE20}%
  \BibitemOpen
  \bibfield  {author} {\bibinfo {author} {\bibfnamefont {L.}~\bibnamefont
  {Zhu}}, \bibinfo {author} {\bibfnamefont {H.~L.}\ \bibnamefont {Tang}},
  \bibinfo {author} {\bibfnamefont {G.~S.}\ \bibnamefont {Barron}}, \bibinfo
  {author} {\bibfnamefont {F.}~\bibnamefont {Calderon-Vargas}}, \bibinfo
  {author} {\bibfnamefont {N.~J.}\ \bibnamefont {Mayhall}}, \bibinfo {author}
  {\bibfnamefont {E.}~\bibnamefont {Barnes}},\ and\ \bibinfo {author}
  {\bibfnamefont {S.~E.}\ \bibnamefont {Economou}},\ }\bibfield  {title}
  {\bibinfo {title} {An adaptive quantum approximate optimization algorithm for
  solving combinatorial problems on a quantum computer},\ }\href@noop {}
  {\bibfield  {journal} {\bibinfo  {journal} {arXiv preprint arXiv:2005.10258
  [quant-ph]}\ } (\bibinfo {year} {2020})}\BibitemShut {NoStop}%
\bibitem [{\citenamefont {Sack}\ and\ \citenamefont {Serbyn}(2021)}]{SS21}%
  \BibitemOpen
  \bibfield  {author} {\bibinfo {author} {\bibfnamefont {S.~H.}\ \bibnamefont
  {Sack}}\ and\ \bibinfo {author} {\bibfnamefont {M.}~\bibnamefont {Serbyn}},\
  }\bibfield  {title} {\bibinfo {title} {Quantum annealing initialization of
  the quantum approximate optimization algorithm},\ }\href@noop {} {\bibfield
  {journal} {\bibinfo  {journal} {arXiv preprint arXiv:2101.05742 [quant-ph]}\
  } (\bibinfo {year} {2021})}\BibitemShut {NoStop}%
\bibitem [{\citenamefont {B{\"a}rtschi}\ and\ \citenamefont
  {Eidenbenz}(2020)}]{BE20}%
  \BibitemOpen
  \bibfield  {author} {\bibinfo {author} {\bibfnamefont {A.}~\bibnamefont
  {B{\"a}rtschi}}\ and\ \bibinfo {author} {\bibfnamefont {S.}~\bibnamefont
  {Eidenbenz}},\ }\bibfield  {title} {\bibinfo {title} {{Grover} mixers for
  {QAOA}: {Shifting} complexity from mixer design to state preparation},\
  }\href@noop {} {\bibfield  {journal} {\bibinfo  {journal} {arXiv preprint
  arXiv:2006.00354}\ } (\bibinfo {year} {2020})}\BibitemShut {NoStop}%
\bibitem [{\citenamefont {Akshay}\ \emph {et~al.}(2019)\citenamefont {Akshay},
  \citenamefont {Philathong}, \citenamefont {Morales},\ and\ \citenamefont
  {Biamonte}}]{APMD19}%
  \BibitemOpen
  \bibfield  {author} {\bibinfo {author} {\bibfnamefont {V.}~\bibnamefont
  {Akshay}}, \bibinfo {author} {\bibfnamefont {H.}~\bibnamefont {Philathong}},
  \bibinfo {author} {\bibfnamefont {M.~E.}\ \bibnamefont {Morales}},\ and\
  \bibinfo {author} {\bibfnamefont {J.~D.}\ \bibnamefont {Biamonte}},\
  }\bibfield  {title} {\bibinfo {title} {Reachability deficits in quantum
  approximate optimization},\ }\href@noop {} {\bibfield  {journal} {\bibinfo
  {journal} {arXiv preprint arXiv:1906.11259v2}\ } (\bibinfo {year}
  {2019})}\BibitemShut {NoStop}%
\bibitem [{\citenamefont {Egger}\ \emph {et~al.}(2020)\citenamefont {Egger},
  \citenamefont {Marecek},\ and\ \citenamefont {Woerner}}]{egger2020warm}%
  \BibitemOpen
  \bibfield  {author} {\bibinfo {author} {\bibfnamefont {D.~J.}\ \bibnamefont
  {Egger}}, \bibinfo {author} {\bibfnamefont {J.}~\bibnamefont {Marecek}},\
  and\ \bibinfo {author} {\bibfnamefont {S.}~\bibnamefont {Woerner}},\
  }\bibfield  {title} {\bibinfo {title} {Warm-starting quantum optimization},\
  }\href@noop {} {\bibfield  {journal} {\bibinfo  {journal} {arXiv preprint
  arXiv:2009.10095}\ } (\bibinfo {year} {2020})}\BibitemShut {NoStop}%
\bibitem [{\citenamefont {Dunning}\ \emph {et~al.}(2018)\citenamefont
  {Dunning}, \citenamefont {Gupta},\ and\ \citenamefont
  {Silberholz}}]{dunning2018}%
  \BibitemOpen
  \bibfield  {author} {\bibinfo {author} {\bibfnamefont {I.}~\bibnamefont
  {Dunning}}, \bibinfo {author} {\bibfnamefont {S.}~\bibnamefont {Gupta}},\
  and\ \bibinfo {author} {\bibfnamefont {J.}~\bibnamefont {Silberholz}},\
  }\bibfield  {title} {\bibinfo {title} {What works best when? {A} systematic
  evaluation of heuristics for {Max-Cut} and {QUBO}},\ }\href@noop {}
  {\bibfield  {journal} {\bibinfo  {journal} {INFORMS Journal on Computing}\
  }\textbf {\bibinfo {volume} {30}},\ \bibinfo {pages} {608} (\bibinfo {year}
  {2018})}\BibitemShut {NoStop}%
\bibitem [{\citenamefont {Kingma}\ and\ \citenamefont {Ba}(2017)}]{adamRef}%
  \BibitemOpen
  \bibfield  {author} {\bibinfo {author} {\bibfnamefont {D.}~\bibnamefont
  {Kingma}}\ and\ \bibinfo {author} {\bibfnamefont {J.~L.}\ \bibnamefont
  {Ba}},\ }\bibfield  {title} {\bibinfo {title} {Adam: A method for stochastic
  optimization},\ }\href@noop {} {\bibfield  {journal} {\bibinfo  {journal}
  {arXiv preprint arXiv:1412.6980}\ } (\bibinfo {year} {2017})}\BibitemShut
  {NoStop}%
\bibitem [{\citenamefont {Powell}(1994)}]{cobylaRef}%
  \BibitemOpen
  \bibfield  {author} {\bibinfo {author} {\bibfnamefont {M.}~\bibnamefont
  {Powell}},\ }\bibfield  {title} {\bibinfo {title} {A direct search
  optimization method that models the objective and constraint functions by
  linear interpolation},\ }\href@noop {} {\bibfield  {journal} {\bibinfo
  {journal} {Advances in Optimization and Numerical Analysis}\ }\textbf
  {\bibinfo {volume} {275}},\ \bibinfo {pages} {51} (\bibinfo {year}
  {1994})}\BibitemShut {NoStop}%
\bibitem [{\citenamefont {Gao}\ and\ \citenamefont
  {Han}(2012)}]{neldermeadRef}%
  \BibitemOpen
  \bibfield  {author} {\bibinfo {author} {\bibfnamefont {F.}~\bibnamefont
  {Gao}}\ and\ \bibinfo {author} {\bibfnamefont {L.}~\bibnamefont {Han}},\
  }\bibfield  {title} {\bibinfo {title} {Implementing the nelder-mead simplex
  algorithm with adaptive parameters},\ }\href@noop {} {\bibfield  {journal}
  {\bibinfo  {journal} {Computational Optimization and Applications}\ }\textbf
  {\bibinfo {volume} {51}},\ \bibinfo {pages} {259} (\bibinfo {year}
  {2012})}\BibitemShut {NoStop}%
\bibitem [{\citenamefont {Fletcher}(1987)}]{bfgsRef}%
  \BibitemOpen
  \bibfield  {author} {\bibinfo {author} {\bibfnamefont {R.}~\bibnamefont
  {Fletcher}},\ }\href@noop {} {\emph {\bibinfo {title} {Practical Methods of
  Optimization (2nd edition)}}}\ (\bibinfo  {publisher} {John Wiley and Sons},\
  \bibinfo {address} {New York, NY, USA},\ \bibinfo {year} {1987})\BibitemShut
  {NoStop}%
\bibitem [{\citenamefont {Lavrijsen}\ \emph {et~al.}(2021)\citenamefont
  {Lavrijsen}, \citenamefont {Tudor}, \citenamefont {Müller}, \citenamefont
  {Iancu},\ and\ \citenamefont {de~Jong}}]{LTMIJ21}%
  \BibitemOpen
  \bibfield  {author} {\bibinfo {author} {\bibfnamefont {W.}~\bibnamefont
  {Lavrijsen}}, \bibinfo {author} {\bibfnamefont {A.}~\bibnamefont {Tudor}},
  \bibinfo {author} {\bibfnamefont {J.}~\bibnamefont {Müller}}, \bibinfo
  {author} {\bibfnamefont {C.}~\bibnamefont {Iancu}},\ and\ \bibinfo {author}
  {\bibfnamefont {W.}~\bibnamefont {de~Jong}},\ }\bibfield  {title} {\bibinfo
  {title} {Classical optimizers for noisy intermediate-scale quantum devices},\
  }\href@noop {} {\bibfield  {journal} {\bibinfo  {journal} {arXiv preprint
  arXiv:2004.030043}\ } (\bibinfo {year} {2021})}\BibitemShut {NoStop}%
\bibitem [{\citenamefont {Wilson}\ \emph {et~al.}(2021)\citenamefont {Wilson},
  \citenamefont {Stromswold}, \citenamefont {Wudarski}, \citenamefont
  {Hadfield}, \citenamefont {Tubman},\ and\ \citenamefont {Rieffel}}]{WSR21}%
  \BibitemOpen
  \bibfield  {author} {\bibinfo {author} {\bibfnamefont {M.}~\bibnamefont
  {Wilson}}, \bibinfo {author} {\bibfnamefont {R.}~\bibnamefont {Stromswold}},
  \bibinfo {author} {\bibfnamefont {F.}~\bibnamefont {Wudarski}}, \bibinfo
  {author} {\bibfnamefont {S.}~\bibnamefont {Hadfield}}, \bibinfo {author}
  {\bibfnamefont {N.~M.}\ \bibnamefont {Tubman}},\ and\ \bibinfo {author}
  {\bibfnamefont {E.~G.}\ \bibnamefont {Rieffel}},\ }\bibfield  {title}
  {\bibinfo {title} {Optimizing quantum heuristics with meta-learning},\
  }\href@noop {} {\bibfield  {journal} {\bibinfo  {journal} {Quantum Machine
  Intelligence}\ }\textbf {\bibinfo {volume} {3}} (\bibinfo {year}
  {2021})}\BibitemShut {NoStop}%
\bibitem [{\citenamefont {Brandao}\ \emph {et~al.}(2018)\citenamefont
  {Brandao}, \citenamefont {Broughton}, \citenamefont {Farhi}, \citenamefont
  {Gutmann},\ and\ \citenamefont {Neven}}]{BBFGN}%
  \BibitemOpen
  \bibfield  {author} {\bibinfo {author} {\bibfnamefont {F.~G.}\ \bibnamefont
  {Brandao}}, \bibinfo {author} {\bibfnamefont {M.}~\bibnamefont {Broughton}},
  \bibinfo {author} {\bibfnamefont {E.}~\bibnamefont {Farhi}}, \bibinfo
  {author} {\bibfnamefont {S.}~\bibnamefont {Gutmann}},\ and\ \bibinfo {author}
  {\bibfnamefont {H.}~\bibnamefont {Neven}},\ }\bibfield  {title} {\bibinfo
  {title} {For fixed control parameters the quantum approximate optimization
  algorithm's objective function value concentrates for typical instances},\
  }\href@noop {} {\bibfield  {journal} {\bibinfo  {journal} {arXiv preprint
  arXiv:1812.04170}\ } (\bibinfo {year} {2018})}\BibitemShut {NoStop}%
\bibitem [{\citenamefont {Delorme}\ and\ \citenamefont {Poljak}(1993)}]{DP93}%
  \BibitemOpen
  \bibfield  {author} {\bibinfo {author} {\bibfnamefont {C.}~\bibnamefont
  {Delorme}}\ and\ \bibinfo {author} {\bibfnamefont {S.}~\bibnamefont
  {Poljak}},\ }\bibfield  {title} {\bibinfo {title} {Laplacian eigenvalues and
  the maximum cut problem},\ }\href@noop {} {\bibfield  {journal} {\bibinfo
  {journal} {Mathematical Programming}\ }\textbf {\bibinfo {volume} {62}},\
  \bibinfo {pages} {557} (\bibinfo {year} {1993})}\BibitemShut {NoStop}%
\bibitem [{\citenamefont {Poljak}\ and\ \citenamefont {Rendl}(1995)}]{PR95}%
  \BibitemOpen
  \bibfield  {author} {\bibinfo {author} {\bibfnamefont {S.}~\bibnamefont
  {Poljak}}\ and\ \bibinfo {author} {\bibfnamefont {F.}~\bibnamefont {Rendl}},\
  }\bibfield  {title} {\bibinfo {title} {Nonpolyhedral relaxations of
  graph-bisection problems},\ }\href@noop {} {\bibfield  {journal} {\bibinfo
  {journal} {SIAM Journal on Optimization}\ }\textbf {\bibinfo {volume} {5}},\
  \bibinfo {pages} {467–} (\bibinfo {year} {1995})}\BibitemShut {NoStop}%
\bibitem [{\citenamefont {Nesterov}\ and\ \citenamefont
  {Nemirovski}(1994)}]{NN94}%
  \BibitemOpen
  \bibfield  {author} {\bibinfo {author} {\bibfnamefont {Y.}~\bibnamefont
  {Nesterov}}\ and\ \bibinfo {author} {\bibfnamefont {A.}~\bibnamefont
  {Nemirovski}},\ }\href@noop {} {\emph {\bibinfo {title} {Interior-point
  polynomial algorithms in convex programming}}}\ (\bibinfo  {publisher}
  {SIAM},\ \bibinfo {address} {Philadelphia, PA, USA},\ \bibinfo {year}
  {1994})\BibitemShut {NoStop}%
\bibitem [{\citenamefont {Barvinok}(1995)}]{Barvinok95}%
  \BibitemOpen
  \bibfield  {author} {\bibinfo {author} {\bibfnamefont {A.~I.}\ \bibnamefont
  {Barvinok}},\ }\bibfield  {title} {\bibinfo {title} {Problems of distance
  geometry and convex properties of quadratic maps},\ }\href@noop {} {\bibfield
   {journal} {\bibinfo  {journal} {{Discrete \& Computational Geometry}}\
  }\textbf {\bibinfo {volume} {13}},\ \bibinfo {pages} {189} (\bibinfo {year}
  {1995})}\BibitemShut {NoStop}%
\bibitem [{\citenamefont {Pataki}(1998)}]{Pataki98}%
  \BibitemOpen
  \bibfield  {author} {\bibinfo {author} {\bibfnamefont {G.}~\bibnamefont
  {Pataki}},\ }\bibfield  {title} {\bibinfo {title} {On the rank of extreme
  matrices in semidefinite programs and the multiplicity of optimal
  eigenvalues},\ }\href@noop {} {\bibfield  {journal} {\bibinfo  {journal}
  {{Mathematics of Operations Research}}\ }\textbf {\bibinfo {volume} {23}},\
  \bibinfo {pages} {339} (\bibinfo {year} {1998})}\BibitemShut {NoStop}%
\bibitem [{\citenamefont {Boumal}\ \emph {et~al.}(2018)\citenamefont {Boumal},
  \citenamefont {Voroninski},\ and\ \citenamefont {Bandeira}}]{BVB18}%
  \BibitemOpen
  \bibfield  {author} {\bibinfo {author} {\bibfnamefont {N.}~\bibnamefont
  {Boumal}}, \bibinfo {author} {\bibfnamefont {V.}~\bibnamefont {Voroninski}},\
  and\ \bibinfo {author} {\bibfnamefont {A.~S.}\ \bibnamefont {Bandeira}},\
  }\bibfield  {title} {\bibinfo {title} {The non-convex burer–monteiro
  approach works on smooth semidefinite programs},\ }\href@noop {} {\bibfield
  {journal} {\bibinfo  {journal} {arXiv preprint arXiv:1606.04970}\ } (\bibinfo
  {year} {2018})}\BibitemShut {NoStop}%
\bibitem [{\citenamefont {Boumal}\ \emph {et~al.}(2020)\citenamefont {Boumal},
  \citenamefont {Voroninski},\ and\ \citenamefont {Bandeira}}]{BVB19}%
  \BibitemOpen
  \bibfield  {author} {\bibinfo {author} {\bibfnamefont {N.}~\bibnamefont
  {Boumal}}, \bibinfo {author} {\bibfnamefont {V.}~\bibnamefont {Voroninski}},\
  and\ \bibinfo {author} {\bibfnamefont {A.~S.}\ \bibnamefont {Bandeira}},\
  }\bibfield  {title} {\bibinfo {title} {Deterministic guarantees for
  {Burer-Monteiro} factorizations of smooth semidefinite programs},\
  }\href@noop {} {\bibfield  {journal} {\bibinfo  {journal} {Communications on
  Pure and Applied Mathematics}\ }\textbf {\bibinfo {volume} {73}},\ \bibinfo
  {pages} {581} (\bibinfo {year} {2020})}\BibitemShut {NoStop}%
\bibitem [{\citenamefont {Burer}\ \emph {et~al.}(2001)\citenamefont {Burer},
  \citenamefont {Monteiro},\ and\ \citenamefont {Zhang}}]{BMZ01}%
  \BibitemOpen
  \bibfield  {author} {\bibinfo {author} {\bibfnamefont {S.}~\bibnamefont
  {Burer}}, \bibinfo {author} {\bibfnamefont {R.~D.~C.}\ \bibnamefont
  {Monteiro}},\ and\ \bibinfo {author} {\bibfnamefont {Y.}~\bibnamefont
  {Zhang}},\ }\bibfield  {title} {\bibinfo {title} {Rank-2 relaxation
  heuristics for {Max-Cut} and other binary quadratic programs},\ }\href@noop
  {} {\bibfield  {journal} {\bibinfo  {journal} {SIAM Journal on Optimization}\
  }\textbf {\bibinfo {volume} {12}},\ \bibinfo {pages} {503} (\bibinfo {year}
  {2001})}\BibitemShut {NoStop}%
\bibitem [{\citenamefont {Mei}\ \emph {et~al.}(2017)\citenamefont {Mei},
  \citenamefont {Misiakiewicz}, \citenamefont {Montanari},\ and\ \citenamefont
  {Oliveira}}]{MMMO17}%
  \BibitemOpen
  \bibfield  {author} {\bibinfo {author} {\bibfnamefont {S.}~\bibnamefont
  {Mei}}, \bibinfo {author} {\bibfnamefont {T.}~\bibnamefont {Misiakiewicz}},
  \bibinfo {author} {\bibfnamefont {A.}~\bibnamefont {Montanari}},\ and\
  \bibinfo {author} {\bibfnamefont {R.~I.}\ \bibnamefont {Oliveira}},\
  }\bibfield  {title} {\bibinfo {title} {Solving {SDPs} for synchronization and
  {MaxCut} problems via the {Grothendieck} inequality},\ }\href@noop {}
  {\bibfield  {journal} {\bibinfo  {journal} {arXiv preprint arXiv:1703.08729}\
  } (\bibinfo {year} {2017})}\BibitemShut {NoStop}%
\bibitem [{\citenamefont {Hagberg}\ \emph {et~al.}(2008)\citenamefont
  {Hagberg}, \citenamefont {Swart},\ and\ \citenamefont {S~Chult}}]{networkx}%
  \BibitemOpen
  \bibfield  {author} {\bibinfo {author} {\bibfnamefont {A.}~\bibnamefont
  {Hagberg}}, \bibinfo {author} {\bibfnamefont {P.}~\bibnamefont {Swart}},\
  and\ \bibinfo {author} {\bibfnamefont {D.}~\bibnamefont {S~Chult}},\
  }\href@noop {} {\emph {\bibinfo {title} {Exploring network structure,
  dynamics, and function using NetworkX}}},\ \bibinfo {type} {Tech. Rep.}\
  (\bibinfo  {institution} {Los Alamos National Lab.(LANL), Los Alamos, NM
  (United States)},\ \bibinfo {year} {2008})\BibitemShut {NoStop}%
\bibitem [{\citenamefont {Wang}\ \emph {et~al.}(2018)\citenamefont {Wang},
  \citenamefont {Hadfield}, \citenamefont {Jiang},\ and\ \citenamefont
  {Rieffel}}]{WHJR18}%
  \BibitemOpen
  \bibfield  {author} {\bibinfo {author} {\bibfnamefont {Z.}~\bibnamefont
  {Wang}}, \bibinfo {author} {\bibfnamefont {S.}~\bibnamefont {Hadfield}},
  \bibinfo {author} {\bibfnamefont {Z.}~\bibnamefont {Jiang}},\ and\ \bibinfo
  {author} {\bibfnamefont {E.~G.}\ \bibnamefont {Rieffel}},\ }\bibfield
  {title} {\bibinfo {title} {Quantum approximate optimization algorithm for
  {MaxCut}: {A} fermionic view},\ }\href@noop {} {\bibfield  {journal}
  {\bibinfo  {journal} {Physical Review A}\ }\textbf {\bibinfo {volume} {97}},\
  \bibinfo {pages} {022304} (\bibinfo {year} {2018})}\BibitemShut {NoStop}%
\bibitem [{\citenamefont {Bravyi}\ \emph {et~al.}(2019)\citenamefont {Bravyi},
  \citenamefont {Kliesch}, \citenamefont {Koenig},\ and\ \citenamefont
  {Tang}}]{BKKT19}%
  \BibitemOpen
  \bibfield  {author} {\bibinfo {author} {\bibfnamefont {S.}~\bibnamefont
  {Bravyi}}, \bibinfo {author} {\bibfnamefont {A.}~\bibnamefont {Kliesch}},
  \bibinfo {author} {\bibfnamefont {R.}~\bibnamefont {Koenig}},\ and\ \bibinfo
  {author} {\bibfnamefont {E.}~\bibnamefont {Tang}},\ }\bibfield  {title}
  {\bibinfo {title} {Obstacles to state preparation and variational
  optimization from symmetry protection},\ }\href@noop {} {\bibfield  {journal}
  {\bibinfo  {journal} {arXiv preprint arXiv:1910.08980}\ } (\bibinfo {year}
  {2019})}\BibitemShut {NoStop}%
\bibitem [{\citenamefont {Karp}(1972)}]{Karp1972}%
  \BibitemOpen
  \bibfield  {author} {\bibinfo {author} {\bibfnamefont {R.~M.}\ \bibnamefont
  {Karp}},\ }\bibinfo {title} {Reducibility among combinatorial problems},\ in\
  \href@noop {} {\emph {\bibinfo {booktitle} {Complexity of Computer
  Computations: Proceedings of a symposium on the Complexity of Computer
  Computations}}}\ (\bibinfo  {publisher} {Springer US},\ \bibinfo {address}
  {Boston, MA},\ \bibinfo {year} {1972})\ pp.\ \bibinfo {pages}
  {85--103}\BibitemShut {NoStop}%
\bibitem [{\citenamefont {Glover}\ \emph {et~al.}(2019)\citenamefont {Glover},
  \citenamefont {Kochenberger},\ and\ \citenamefont {Du}}]{GKD19}%
  \BibitemOpen
  \bibfield  {author} {\bibinfo {author} {\bibfnamefont {F.}~\bibnamefont
  {Glover}}, \bibinfo {author} {\bibfnamefont {G.}~\bibnamefont
  {Kochenberger}},\ and\ \bibinfo {author} {\bibfnamefont {Y.}~\bibnamefont
  {Du}},\ }\bibfield  {title} {\bibinfo {title} {Quantum bridge analytics {I}:
  A tutorial on formulating and using {QUBO} models},\ }\href@noop {}
  {\bibfield  {journal} {\bibinfo  {journal} {arXiv preprint arxiv:1811.11538}\
  } (\bibinfo {year} {2019})}\BibitemShut {NoStop}%
\bibitem [{\citenamefont {Lodewijks}(2020)}]{lodewijks2020}%
  \BibitemOpen
  \bibfield  {author} {\bibinfo {author} {\bibfnamefont {B.}~\bibnamefont
  {Lodewijks}},\ }\bibfield  {title} {\bibinfo {title} {Mapping {NP}-hard and
  {NP}-complete optimization problems to quadratic unconstrained binary
  optimization problems},\ }\href@noop {} {\bibfield  {journal} {\bibinfo
  {journal} {arXiv preprint arxiv:1911.08043}\ } (\bibinfo {year}
  {2020})}\BibitemShut {NoStop}%
\bibitem [{\citenamefont {Herrman}\ \emph {et~al.}(2021)\citenamefont
  {Herrman}, \citenamefont {Treffert}, \citenamefont {Ostrowski}, \citenamefont
  {Lotshaw}, \citenamefont {Humble},\ and\ \citenamefont {Siopsis}}]{HTOLHS21}%
  \BibitemOpen
  \bibfield  {author} {\bibinfo {author} {\bibfnamefont {R.}~\bibnamefont
  {Herrman}}, \bibinfo {author} {\bibfnamefont {L.}~\bibnamefont {Treffert}},
  \bibinfo {author} {\bibfnamefont {J.}~\bibnamefont {Ostrowski}}, \bibinfo
  {author} {\bibfnamefont {P.~C.}\ \bibnamefont {Lotshaw}}, \bibinfo {author}
  {\bibfnamefont {T.~S.}\ \bibnamefont {Humble}},\ and\ \bibinfo {author}
  {\bibfnamefont {G.}~\bibnamefont {Siopsis}},\ }\bibfield  {title} {\bibinfo
  {title} {Impact of graph structures for qaoa on maxcut},\ }\href@noop {}
  {\bibfield  {journal} {\bibinfo  {journal} {Quantum Information Processing}\
  }\textbf {\bibinfo {volume} {20}},\ \bibinfo {pages} {1} (\bibinfo {year}
  {2021})}\BibitemShut {NoStop}%
\bibitem [{\citenamefont {Feller}(1971)}]{feller_1971}%
  \BibitemOpen
  \bibfield  {author} {\bibinfo {author} {\bibfnamefont {W.}~\bibnamefont
  {Feller}},\ }\href@noop {} {\emph {\bibinfo {title} {An Introduction to
  Probability Theory and Its Applications}}},\ \bibinfo {edition} {3rd}\ ed.,\
  Vol.~\bibinfo {volume} {2}\ (\bibinfo  {publisher} {Wiley},\ \bibinfo
  {address} {Hoboken, New Jersey},\ \bibinfo {year} {1971})\BibitemShut
  {NoStop}%
\bibitem [{\citenamefont {{The Sage Developers}}(2020)}]{sagemath}%
  \BibitemOpen
  \bibfield  {author} {\bibinfo {author} {\bibnamefont {{The Sage
  Developers}}},\ }\href@noop {} {\emph {\bibinfo {title} {{S}ageMath, the
  {S}age {M}athematics {S}oftware {S}ystem}}} (\bibinfo {year} {2020}),\
  \bibinfo {note} {{\tt https://www.sagemath.org}}\BibitemShut {NoStop}%
\bibitem [{\citenamefont {Lyons}\ and\ \citenamefont {Peres}(2017)}]{lyons}%
  \BibitemOpen
  \bibfield  {author} {\bibinfo {author} {\bibfnamefont {R.}~\bibnamefont
  {Lyons}}\ and\ \bibinfo {author} {\bibfnamefont {Y.}~\bibnamefont {Peres}},\
  }\href@noop {} {\emph {\bibinfo {title} {Probability on trees and
  networks}}},\ Vol.~\bibinfo {volume} {42}\ (\bibinfo  {publisher} {Cambridge
  University Press},\ \bibinfo {year} {2017})\BibitemShut {NoStop}%
\bibitem [{\citenamefont {Karush}(1939)}]{Kar39}%
  \BibitemOpen
  \bibfield  {author} {\bibinfo {author} {\bibfnamefont {W.}~\bibnamefont
  {Karush}},\ }\emph {\bibinfo {title} {Minima of functions of several
  variables with inequalities as side conditions}},\ \href@noop {} {Master's
  thesis},\ \bibinfo  {school} {University of Chicago} (\bibinfo {year}
  {1939})\BibitemShut {NoStop}%
\bibitem [{\citenamefont {Kuhn}\ and\ \citenamefont {Tucker}(1951)}]{KT51}%
  \BibitemOpen
  \bibfield  {author} {\bibinfo {author} {\bibfnamefont {H.~W.}\ \bibnamefont
  {Kuhn}}\ and\ \bibinfo {author} {\bibfnamefont {A.~W.}\ \bibnamefont
  {Tucker}},\ }\bibfield  {title} {\bibinfo {title} {Nonlinear programming},\
  }in\ \href@noop {} {\emph {\bibinfo {booktitle} {Proceedings of the Second
  Berkeley Symposium on Mathematical Statistics and Probability}}}\ (\bibinfo
  {publisher} {University of California Press},\ \bibinfo {address} {Berkeley,
  Calif.},\ \bibinfo {year} {1951})\ pp.\ \bibinfo {pages}
  {481--492}\BibitemShut {NoStop}%
\bibitem [{\citenamefont {Lee}\ and\ \citenamefont {Padmanabhan}(2019)}]{LP19}%
  \BibitemOpen
  \bibfield  {author} {\bibinfo {author} {\bibfnamefont {Y.~T.}\ \bibnamefont
  {Lee}}\ and\ \bibinfo {author} {\bibfnamefont {S.}~\bibnamefont
  {Padmanabhan}},\ }\bibfield  {title} {\bibinfo {title} {An
  $\tilde{O}(m/\epsilon^3.5 )$-cost algorithm for semidefinite programs with
  diagonal constraints},\ }\href@noop {} {\bibfield  {journal} {\bibinfo
  {journal} {arXiv preprint arXiv:1903.01859}\ } (\bibinfo {year}
  {2019})}\BibitemShut {NoStop}%
\end{thebibliography}%
\def\authornoop#1{}

\appendix

\section{Computational Details}
\subsection{Details of Rotations}\label{sec:appendixrotations}

\subsubsection{Random vertex-at-top}
		We first describe the rotation in 3-dimensions for vertex $$v_i = \allowbreak (\sin \theta_i \cos \phi_i,  \sin\theta_i \sin \phi_i, \cos \theta_i)^T,$$ which is sampled uniformly at random (for $i\in[n]$).
		The rotation that maps $v_i \in \mathbb{R}^3$ to $(0,0,1)^T$ is obtained by first rotating clockwise along the $z$-axis by $\phi_i$, followed by a clockwise rotation along the $y$-axis by $\theta_i$, followed by a uniform at random rotation $\mu$ in $[0,2\pi]$ around the $z$-axis.
		
		Indeed, one can check that $R_V(\mathbf{x}^*(v_i)) = (0,0,1)^T$ (which will correspond to the quantum state $\ket{0}$ on the Bloch Sphere). For rank-2 solutions, with a uniform at random vertex $v_i = (\cos(\theta_i), \sin(\theta_i))^T$ sampled from $i\in[n]$, we can simply work with polar coordinates and shift all polar angles by $\theta_i$ to obtain the random vertex-at-top rotation. To be precise, we set $\theta_j = \theta_j - \theta_i$ where $\theta_j$ denotes the angle of the point corresponding to the $j$th vertex in the rank-2 solution.

\subsubsection{Uniform at Random Rotation}
		In this case, we  uniformly pick a rotation of the $k$-dimensional sphere and apply it.
		For rank-3 solutions, one way to accomplish this is by picking a point $\hat{x}$ uniformly at random from the surface of the sphere, rotating $\hat{x}$ to the top of the sphere (in a way similar to the vertex-at-top rotations), and then performing a uniform random rotation in $[0,2\pi]$ around the $z$-axis.
		Such an $\hat{x}$ can be generated by  picking $\alpha,\beta$ uniformly at random from the interval $[0,1]$ and then setting $\phi = 2\pi \alpha$ and $\theta = \arccos(2\beta-1)$. The pair $(\theta,\phi)$ will then correspond to the polar and azimuthal angles of the randomly chosen point $\hat{x}$ on the surface of the sphere \cite{feller_1971}.
		
		For rank-2 solutions, we can simply shift all the angles by some random angle. More precisely, set $\theta_j = \theta_j + \hat{\theta}$ where $\theta_j$ denotes the angle of the point corresponding to the $j$th vertex in the rank-2 solution and $\hat{\theta}$ is chosen uniformly at random in $[0,2\pi]$.

\label{sec:appendixAlgos}
\subsection{Finding approximate BM-MC$_k$ solutions.}
In Algorithm \ref{alg:getSDPRank2}, we describe our implementation for obtaining the semidefinite programming (SDP) solution for BM-MC$_k$ for $k=2,3$ using coordinate ascent. In the algorithms below, we write $U(a,b)$ to denote the uniform distribution on the interval $[a,b]$ (where $a,b \in \mathbb{R}$ with $a<b$). 
We set $\eta = 1/20$ for experiments in this work. We normalize the angles output by BM-MC$_3$ to enforce the standard range of angles for spherical coordinates without changing the objective value. 

\begin{algorithm}[t]
        \DontPrintSemicolon
		\caption{Obtain Solution for BM-MC$_k$}
	    \KwIn{Weighted graph $G = (V,E), w: E \to \mathbb{R}, k \in \{2,3\}$}
		\label{alg:getSDPRank2}
		If $k=2$, let $\theta_1,\dots,\theta_n \in \mathbb{R}$ be the angles of $n$ points chosen uniformly at random on the 2-dimensional unit circle. If $k=3$, let $(\theta_1,\phi_1),\dots,(\theta_n,\phi_n)$ be the spherical coordinates of $n$ points chosen uniformly at random on the 3-dimensional sphere.\;
		\Repeat{no improvement in objective by $\geq$ $10^{-5}\sum_{e \in E} |w_e|$ within 100 evaluations.}
		{
    		\For(\tcp*[f]{coordinate ascent}){$i=1$ through $n$}
    		{
    		    Sample the perturbation value(s) $\Delta \theta$ (and $\Delta \phi$ if $k=3$) from $U(-\eta,\eta)$ for small $\eta > 0$.\;
        		Update $\theta_i = \theta_i + \Delta \theta$ (and $\phi_i = \phi_i + \Delta \phi$ if $k=3$) and compute the BM objective.\;
        		If the objective improves, keep the perturbation.\;
    		}
		}
\end{algorithm}

\subsection{Graph Instances.}
\label{sec:graphInstancesAppendix}

As mentioned in Section \ref{sec:computational}, for our simulations, we first generate a collection of unit-weight graphs. From each one, we create multiple weighted graphs by applying different edge-weight distributions to the unit-weight graph. Below, we describe the collection of unit-weight graphs that were generated for this process.

The collection of non-isomorphic graphs up to 6 nodes were generated using \texttt{SageMath} \cite{sagemath}. The remaining instances were generated using various random graph generators found in the \texttt{NetworkX} package \cite{networkx}; the parameter names used below are the same as those used in the corresponding \texttt{NetworkX} functions .

\begin{itemize}
        \item All non-isomorphic connected graphs up to 6 nodes (142 instances)
        \item Erdos-Renyi (\fe{42} instances): for each $n$ from $n=7$ to $n=\fe{12}$, create 7 instances with $p$ sampled from $[0,1]$ uniformly.
        \item Random Regular (\fe{42} instances): for each $n$ from $n=7$ to $n=\fe{12}$, create 7 instances with $d$ sampled uniformly from valid degrees.
        \item Barabasi Albert (\fe{18} instances): for each $n$ from $n=7$ to $n=\fe{12}$ and for all $m$ in $\{1,2,3\}$, create 1 instance (with initial graph being star graph on $m+1$ nodes)
        \item Dual Barabasi Albert (\fe{36} instances): for each $n$ from $n=7$ to $n=\fe{12}$ and for all $\{(m_1,m_2) : m_1,m_2 \in \{1,2,3\}, m_1 \neq m_2\}$ with $p=0.25$, create 1 instance with initial graph on star with  $\max(m_1,m_2)+1$ nodes
        \item Watts Strogatz Graphs (\fe{18} instances): for each $n$ from $n=7$ to $n=\fe{12}$, for all $k$ in  $\{2,4,6\}$, create 1 instance with $p$ sampled uniformly from $[0,1]$.
        \item Newman Watts Strogatz Graphs (\fe{18} instances): for each $n$ from $n=7$ to $n=\fe{12}$, for all k in  $\{2,4,6\}$, create 1 instance with $p$ sampled uniformly from $[0,1]$.
    \end{itemize}
    
    \fe{Figure \ref{fig:spaceOfInstances} demonstrates how varied our ensemble is with respect to two important graph metrics dependent on eigenvalues of the normalized Laplacian \cite{lyons}. }
    
    \begin{figure}
    \centering
    \includegraphics[scale=0.5]{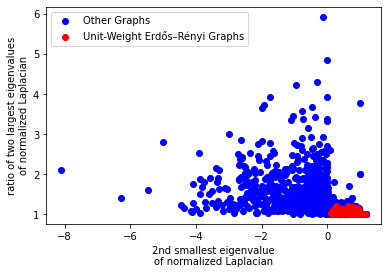}
    \caption{Illustration depicting the range of graph metrics for our instance library $\mathcal{G}$. When comparing unit-weight Erd\H{o}s-R\'enyi graphs (red) with the remaining graphs in $\mathcal{G}$ (blue), there is an increase in the range of values obtained for both graph metrics.}
    \label{fig:spaceOfInstances}
\end{figure}
    
\subsection{\tsout{Median Performance of} \QAOA-Warm \fe{with Median and Worst Vertex-At-Top Rotations}}
\label{sec:QAOAWarmWithMedianAndWorstRotations}
For our numerical simulations in Section \ref{sec:computational}, we use the best of either 5 vertex-at-top rotations or best of 5 uniform rotations for \QAOAw. Performing multiple runs of \QAOAw\ with different rotations and taking the best allows one to mitigate the possibility of using a warm-start with a poor rotation. We present the results with respect to the median rotation here. We plot the results below in Figure \ref{fig:medianRotations}; we see that the results do not differ much from what was seen in Figure \ref{fig:histogramWarmVsStandard}. In Figure \ref{fig:worstRotations}, we also plot the results when the worst of 5 vertex-at-top rotations are used to give an idea of the worst-case performance for \QAOAw.
\begin{figure}[H]
    \centering
    \includegraphics[scale=0.4]{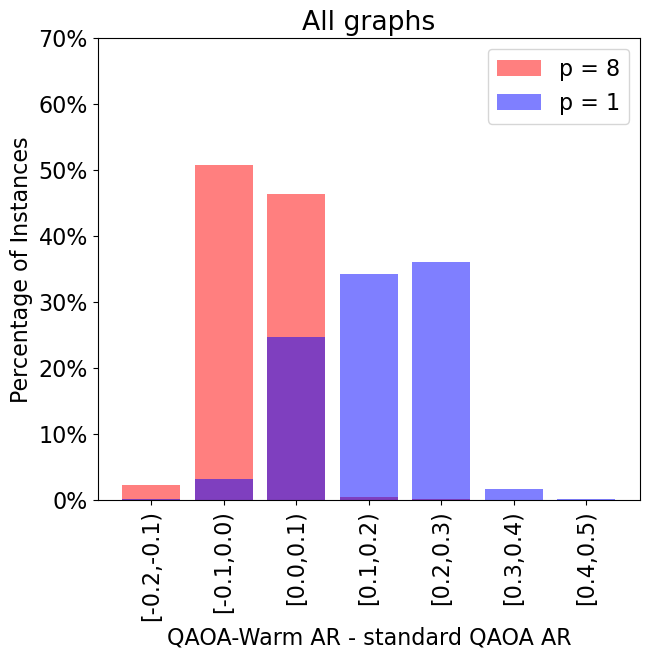}
    \includegraphics[scale=0.4]{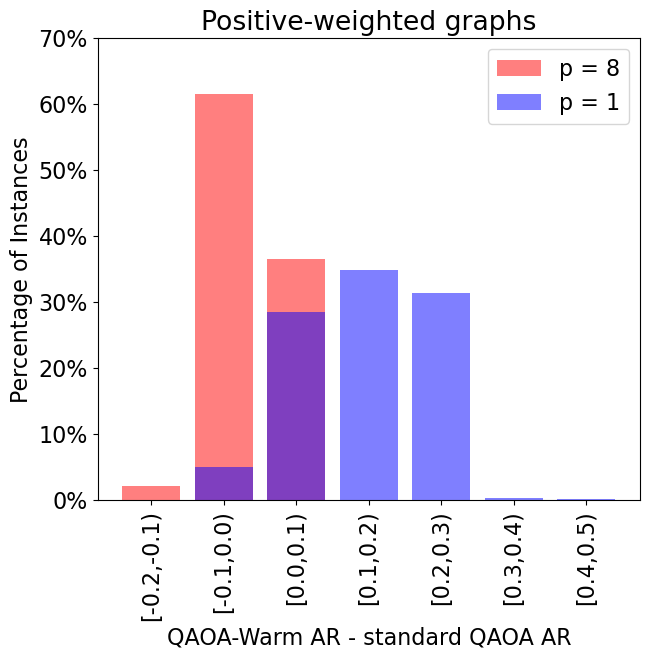}
    \caption{Histograms comparing the performance in (depth-$p$) \QAOAw\ and (depth-$p$) standard \QAOA\ for both $p=1$ (blue) and $p=8$ (red) where the median vertex-at-top rotations are used. Overlapping portions of the histogram are in purple. The \tsout{top} \fe{left} plot is generated using the graphs in our graph library $\mathcal{G}$ (see Section \ref{subsec:experimentalSetup}) whereas for the \tsout{bottom} \fe{right} plot, we restrict our attention to only those graphs in $\mathcal{G}$ with positive edge weights. \tsout{Larger bin sizes are used for the far right portions of the histogram; we color the bin label to emphasize this.}}
    \label{fig:medianRotations}
\end{figure}

\begin{figure}
    \centering \includegraphics[scale=0.4]{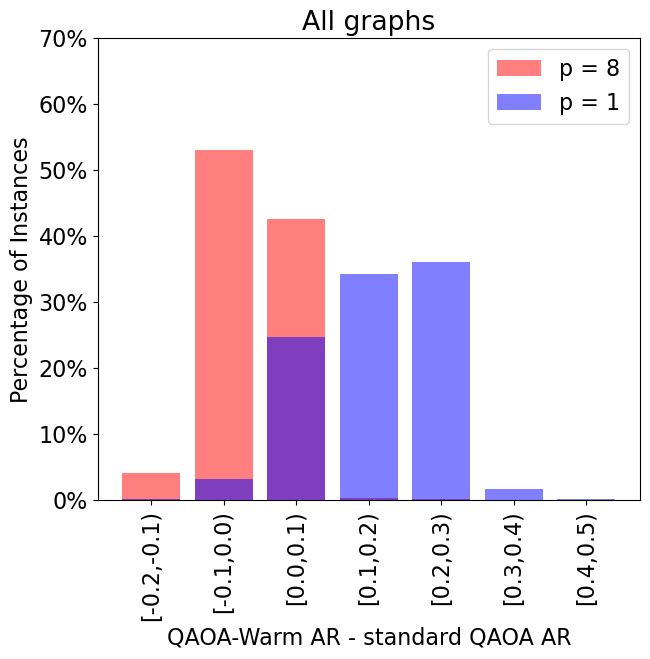}
    \includegraphics[scale=0.4]{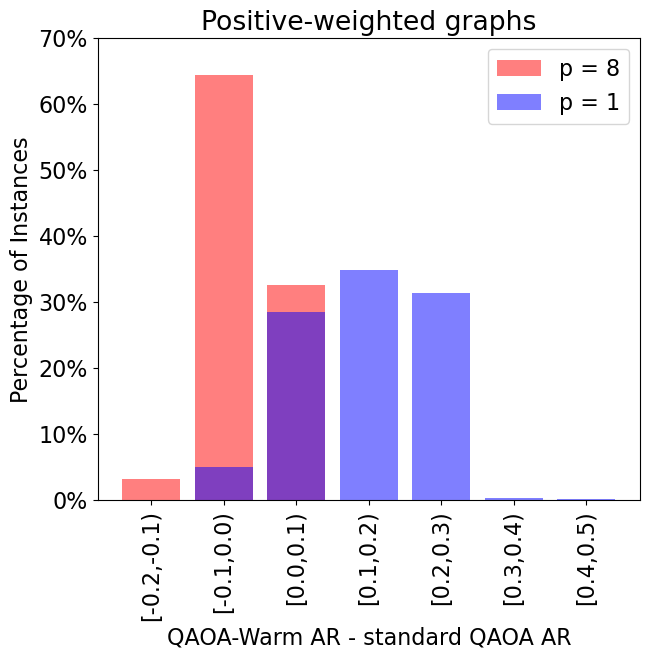}
    \caption{Histograms comparing the performance in (depth-$p$) \QAOAw\ and (depth-$p$) standard \QAOA\ for both $p=1$ (blue) and $p=8$ (red) where the worst vertex-at-top rotations are used. Overlapping portions of the histogram are in purple. The \tsout{top} \fe{left} plot is generated using the graphs in our graph library $\mathcal{G}$ (see Section \ref{subsec:experimentalSetup}) whereas for the \tsout{bottom} \fe{right} plot, we restrict our attention to only those graphs in $\mathcal{G}$ with positive edge weights. \tsout{Larger bin sizes are used for the far right portions of the histogram; we color the bin label to emphasize this.}}
    \label{fig:worstRotations}
\end{figure}

\section{Proofs}\label{sec:appendixproofs}

\paragraph{Proof for Theorem \ref{thm:warmstart}.}
\begin{proof}
	
	We consider a graph $G=(V,E)$ on two vertices connected by an edge of unit weight initialized with $\ket{s} := \ket{u} \otimes \ket{v}$ where $\ket{u} := \ket{+} = \frac{1}{\sqrt{2}}(\ket{0}+\ket{1})$ and $\ket{v} := \ket{-} = \frac{1}{\sqrt{2}}(\ket{0}-\ket{1})$. (Note that $\ket{s} = Q(x^*)$ where $x^* = ((1,0,0)^T,(-1,0,0)^T)$ is an optimal solution to BM-MC$_3$.)

	We first consider the $p=1$ case. For convenience, let $\gamma := \gamma_1$ and $\beta := \beta_1$ for this case.
	Observe, $$\ket{s} = \ket{u} \otimes \ket{v} = \frac{1}{\sqrt{2}}(\ket{0}+\ket{1})\otimes  \frac{1}{\sqrt{2}}(\ket{0}-\ket{1}) =  \frac{1}{2}(\ket{00}-\ket{01}+\ket{10}-\ket{11}).$$ Note that if we were to do a quantum measurement of this state, we would get each of the four states $\ket{00},\ket{01},\ket{10},\ket{11}$ with equal probability, i.e., the theorem holds in the $p=0$ case as well.
	
	Since $H_C$ is the Hamiltonian of the \mc\ problem, $H_C\ket{01} = 1\cdot\ket{01}$. Thus,  $e^{-i\gamma H_C}\ket{01} = e^{-i\gamma \cdot 1}\ket{01} = e^{-i\gamma}\ket{01}$. Similar calculations show that $e^{-i\gamma H_C}\ket{10} = e^{-i\gamma}\ket{10}, e^{-i\gamma H_C}\ket{00} = \ket{00},$ and $e^{-i\gamma H_C}\ket{11} = \ket{11}$ and thus by linearity,
	$$\ket{s'} := e^{-i\gamma H_C}\ket{s} = e^{-i\gamma H_C}\left(\frac{1}{2}(\ket{00}-\ket{01}+\ket{10}-\ket{11})\right) = \frac{1}{2}\left(\ket{00} - \ket{11} + e^{-i \gamma}\big(\ket{10} - \ket{01}\big)\right).$$

	For a 2-node graph, $H_B = \sigma^x_1+\sigma^x_2$, and thus,
	$$H_B(\ket{00}-\ket{11}) = \sigma^x_1\ket{00} - \sigma^x_1\ket{11}+\sigma^x_2\ket{00}-\sigma^x_2\ket{11}  = 0,$$
	and similarly $H_B(\ket{10}-\ket{01})  = 0.$
	
	By the above observations and linearity, we have that $H_B\ket{s'} = 0$, i.e., $\ket{s'}$ is an eigenvector of $H_B$ with eigenvalue 0 and thus, $$\ket{\psi_1(\gamma,\beta)} = e^{-i\beta H_B}\ket{s'} = e^{-i\beta \cdot 0}\ket{s'} = \ket{s'},$$
	i.e., the mixing Hamiltonian has no effect on the quantum state. Writing out $\ket{s'}$, we have
	$$\ket{\psi_1(\gamma,\beta)} = \frac{1}{2}\left(\ket{00} - \ket{11} + e^{-i \gamma}\big(\ket{10} - \ket{01}\big)\right).$$

	If we repeat all of these calculations in the case that $p > 1$, we get that
	$$\ket{\psi_p(\gamma,\beta)} = \frac{1}{2}\left( \ket{00}-\ket{11} + e^{-i\gamma_p}\cdots e^{-i\gamma_1}\big(\ket{10}-\ket{01}\big)\right) $$$$= \frac{1}{2}\left( \ket{00}-\ket{11} +e^{-i\sum_{i=1}^p \gamma_i}\big(\ket{10} - \ket{01}\big)\right),$$
	in which case all four states $\ket{00},\ket{01},\ket{10},\ket{11}$ are measured with equal probability meaning that the expected cut value for $G$ is 50\% of the maximum cut in $G$.
\end{proof}

% \paragraph{Proof for Lemma \ref{lem:halfApproximation}.}

% \begin{proof}
% 	Let $F_0^\prime(\gamma,\beta)$ be the expected value of {\sc Max-Cut} obtained by quantum sampling (i.e., QAOA for $p=0$). Then,
% 	\begin{align*}
% 		\frac{F_0^\prime(\gamma,\beta)}{\text{\mc}(G)} 
% 		&= \frac{\mathbb{E}\left[ \sum_{(i,j) \in E(G)} \mathbf{1}[\text{$i$ and $j$ have different spins}]\right]}{\text{\mc}(G)}\\
% 		&=\frac{ \sum_{(i,j) \in E(G)} \Pr(\text{$i$ and $j$ have different spins})}{\text{\mc}(G)} \\
% 		&=\frac{ \sum_{(i,j) \in E(G)} 0.5}{\text{\mc}(G)} \tag{as initial state is $\ket{+}^{\otimes n}$}\\
% 		&= \frac{0.5m}{\text{\mc}(G)} \tag{where $m$ is number of edges}\\
% 		&\geq \frac{0.5m}{m} \tag{as $\text{\mc}(G) \leq m$}
% 		= 0.5,
% 	\end{align*}
% 	which yields the desired result. The factor of $0.5$ is due to the independent measurements of the qubits in the state $\ket{+}^{\otimes n}$.
% \end{proof}

\paragraph{Proof for Theorem \ref{thm:cycle}.}

{\begin{proof}
		Without loss of generality, let $G$ be a cycle with $n$ vertices.
		Let $x: V \rightarrow S^2$ be a local optimum for BM-MC$_3$.
		Our proof consists of two steps. First, we show $\text{rank}(\text{span}(\set{x_v : v \in V})) \leq 2$. Next, building upon this characterization for the local optima we show that in fact the above rank is exactly $1$ and all edges are (fully) cut, i.e., global optimum is achieved.
		
		We use first order necessary conditions, known as KKT, to derive the desired characterization. Let us formulate the Lagrangian for our constrained optimization problem,
		$$\mathcal{L}(x,\alpha) = \sum_{(u,v) \in E} \|x_u - x_v\|^2 - \sum_{v \in V} \alpha_v (\|x_v\|^2-1)\,,$$
		where $\alpha_v \in \mathbb{R}$ is a multiplier corresponding to the condition $\|x_v\| = 1, ~ \forall v \in V$.
		It is easy to see the objective for our constrained optimization problem is equal to $
		\max_{x:V \rightarrow S^2} \min_{\alpha: V \rightarrow \mathbb{R}} \mathcal{L}(x,\alpha)\,.
		$ Further using Lagrangian duality theory, we apply 
		KKT optimality conditions that require (at any local optima) stationary condition $\frac{\partial \mathcal{L}}{\partial x_v} = 0$ is satisfied for all $v \in V$ in addition to the following feasibility and complementary slackness conditions (which are trivially satisfied):
		
		\begin{itemize}
			\item Primal feasibility requires $\|x_v\| = 1, \quad \forall v \in V$.
			\item Dual feasibility requires $\alpha_v \in \mathbb{R}, \quad \forall v \in V$.
			\item Complementary slackness requires $\alpha_v = 0$ whenever $\|x_v\| \ne 1$.
		\end{itemize}

        The stationary condition can be reformulated as
		$$\frac{\partial }{\partial x_v} \Big(  \sum_{(u,v)\in E} (x_v-x_u)^T (x_v-x_u) - \alpha_v (x_v^T x_v - 1)\Big) = 0 \quad \forall v \in V\,.$$
		Thus, for all $v \in V$, we have the following stationary condition:
		$\sum_{(u,v)\in E} (x_v - x_u) = \alpha_v x_v .$
		
		Considering our graph being a cycle, where every vertex $v \in V$ has two neighbors, the stationary condition implies a linear dependence of $x_v$ with $x_u$'s corresponding to its neighbors. Hence, $\text{rank}(\text{span}( \set{x_v, x_w, x_u})) \leq 2$ where $w$ and $u$ are two neighbors of $v$. Note that if this rank is $1$, one can easily show neighbors of this vertex (and consequently for every vertex) are at antipodal points $x_w = x_u = -x_v$. Otherwise, $x_w$ and $x_u$ are mirrored with respect to $\text{span}(\set{x_v})$. In this case, these three vectors lie on a unique plane. Inductively, one can show all vertices of the cycle lie on the same plane.
		
		With all the points lying on the same plane, it remains to show that the additional dimension (direction) in $\mathbb{R}^3$ allows one to (locally) increase the objective. Without loss of generality, let $x_v \in \mathbb{R}^3$ be a point on the unit sphere with polar angle $\theta = \pi/2$ and azimuthal angle $\phi_v$. Coloring vertices of the cycle by two colors $c_v \in \set{1, 2}$, it is easy to see for $\tilde{x}_v = (1, \pi/2 + (-1)^{c_v} \varepsilon, \phi_v)$ all edges stretch (unless they are antipodal) so the objective increases. This shows $x$ was not a local optimum, in case of a rank $2$ assignment.
\end{proof}}

\paragraph{Proof for Observation \ref{lem:landscape}}
\begin{proof}
	Since we are working only with circuit depth $p=1$, for simplicity, we use $\gamma$ and $\beta$ to denote $\gamma_1$ and $\beta_1$ respectively.
	
	If $\ket{s_0} = Q(R_V(\mathbf{x}))$ where $R_V(\cdot)$ is a vertex-at-top rotation and $\mathbf{x}$ is as described in the statement of the observation, then it is straightforward to see that there exists a reordering of the vertices such that $\ket{s_0} = \ket{0}^{\otimes r} \otimes \ket{1}^{\otimes (n-r)}$ where $|S| = r$ (where the first $r$ qubits correspond to vertices in $S$).
	
	Let $M = \text{\mc}(G)$. 
	Since $H_C$ is the Hamiltonian of the \mc\ problem and $\ket{s_0}$ corresponds to an optimal cut, then $\ket{s_0}$ is a eigenvector of $H_C$ with eigenvalue $M$. Thus, \begin{equation} e^{-i\gamma H_C}\ket{s_0} = e^{-i\gamma \cdot M}\ket{s_0} = \alpha \ket{s_0}. \label{eqn:costHamNoEffect}\end{equation}
	where $\alpha = e^{-i\gamma M}$.
	Using equation (\ref{eqn:costHamNoEffect}), we have that
	$$	\ket{\psi_1(\gamma,\beta)} = e^{-i\beta H_B}e^{-i\gamma H_C}\ket{s_0} = \alpha e^{-i\beta H_B}\ket{s_0} =  \alpha\left(\ket{s_{\beta}}^{\otimes r} \otimes  \ket{s_{\beta}'}^{\otimes (n-r)}\right),$$
	where
	$s_\beta =  \cos(\beta)\ket{0}-i\sin(\beta)\ket{1}$ 
	and
	$s_\beta' = \cos(\beta)\ket{1}-i\sin(\beta)\ket{0}.$
	
	 The expected energy is the sum of the expected energy for each edge $(u,v) \in E$ and each edge contributes a non-zero amount if and only if both endpoints have a different spin after measurement. However, since $\ket{\psi_1(\gamma,\beta)}$ is an unentangled state, then we can consider measuring each vertex independently.\footnote{The $\alpha$ term is a global phase change that doesn't affect the measurement and can thus be ignored.} Consider an edge $(u,v) \in E$ and suppose that $u \in S$ and $v \in S$. Then,
	\begin{align*}
		&\Pr(\text{$u$ and $v$ measured to be $\ket{0}$ and $\ket{1}$ respectively})\\
		&= \Pr(\text{$u$ measured to be $\ket{0}$})\cdot  \Pr(\text{$v$ measured to be $\ket{1}$}) \tag{$\ket{\psi_1(\gamma,\beta)}$ is unentangled}\\
		&= \Pr(\text{$s_\beta$ measured to be $\ket{0}$}) \cdot \Pr(\text{$s_\beta$ measured to be $\ket{1}$}) \tag{by construction}\\
		&= \cos^2(\beta) \cdot \sin^2(\beta). \tag{def of $s_\beta$}
	\end{align*}
	
	Similar calculations show that if $u \in S$ and $v \in S$, then the probability that $u$ is measured to be $\ket{1}$ and $v$ is measured to be $\ket{0}$ is also $\cos^2(\beta)\sin^2(\beta)$. In the case that $u \in S$ and $v \notin S$, one can similarly show that the probability of measuring both to have differing spins is given by $\cos^4(\beta)+\sin^4(\beta)$.
	
	Combining all the calculations above, the expected energy is given by
	$$F_1(\gamma,\beta) =  2(W-M)\sin^2\beta \cos^2 \beta + M(\sin^4\beta +\cos^4\beta),$$
	where $W$ is the sum of all edge weights (i.e. $W = \sum_{e \in E} w_e$).
	
	Observe that when $\beta = 0$, the above equation reduces to $F_1(\gamma,\beta) = M$ as desired.
	By applying various trigonometric identities and algebraic manipulations, we can rewrite the above function as
	$$F_1(\gamma,\beta) = \frac{1}{4}\Big((2M-W)\cos(4\beta)+2M+W)\Big).$$
	
	The \mc\ of a graph is at least half the sum of the edge weights, i.e.,  $M \geq W/2$ (which implies $2M-W \geq 0$).\footnote{To see this, observe that the expected sum of the weights of the edges crossing a random cut (where one independently places each vertex on one side of the cut or the other with probability $1/2$) will be $W/2$. Since the expectation is $W/2$, then there must exist at least one cut where the sum of the weights crossing the cut is \emph{at least} $W/2$, i.e., $M \geq W/2$. } Since $\cos(4\beta)$ is decreasing in $|\beta|$ for $\beta \in [-\pi/4,\pi/4]$, then it must be that $F_1$ is decreasing (in $|\beta|$ for $\beta \in [-\pi/4,\pi/4]$).
\end{proof}

\section{QAOA-Warm on Antipodal Structures}
\label{sec:antipodalStructures}
\begin{figure}[t]
	\centering
	\hspace{-0.3cm}\begin{tabular}{c@{\hskip 0.1cm}c@{\hskip 0.1cm}c@{\hskip 0.1cm}c@{\hskip 0.1cm}cc}
	\hspace{0.5cm}(a)&(b)&(c)& \hspace{-0.5cm}(d)\\
	\raisebox{-0.03cm}{\includegraphics[scale=0.38]{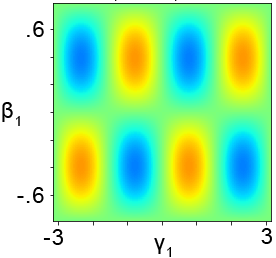}}
	&\includegraphics[scale=0.38]{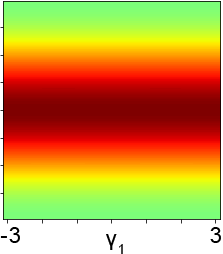}
	&\includegraphics[scale=0.38]{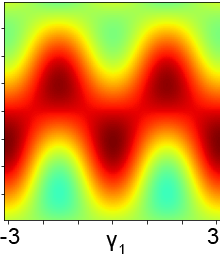}&\includegraphics[scale=0.38]{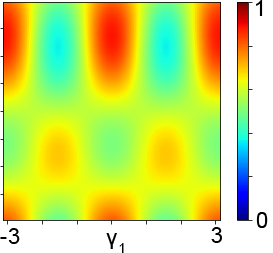}\rotatebox{90}{\hspace{0.35cm} {\small Approximation Ratio}}
	\end{tabular}
	\caption{Parameter landscapes of the four-cycle $C_4$ for $p=1$. When no warm start is used, the landscape has many peaks and valleys in the form of local maxima and minima (a). For both BM-MC$_2$ and BM-MC$_3$, a vertex-at-top rotation yields a convex landscape with a ridge-like shape (b), thereby effectively capturing the optimal solution for the 4-cycle. When a uniform rotation is used, a BM-MC$_2$ formulation (c) is able to achieve optimality for \emph{some} choice of parameters whereas this is not the case for a BM-MC$_3$ formulation (d).}
	\label{fig:fourCycleLandscape}
\end{figure}
We illustrate a set of graph instances where \QAOAw\ has a significant advantage over standard \QAOA\ by considering BM-MC$_k$ solutions that have a special structure. For any positive integer $k$, we say that $\mathbf{x} \in (\mathbb{R}^k)^n$ has an \emph{optimal antipodal structure (in $\mathbb{R}^k$)} for a graph $G=(V,E)$ if there exists $S \subseteq V$ and a unit vector $u \in \mathbb{R}^k$ such that $(S,V\setminus S)$ is a \mc\ of $G$ and $x_i = u$ if $i \in S$ and $x_i = -u$ if $i \notin S$. That is, the points corresponding to each vertex lie at antipodal points on the sphere in a way determined by some \mc\ of $G$. If we consider $\ket{s_0} = Q(R_V(\mathbf{x}))$ where $R_V$ is a random vertex-at-top rotation and $\mathbf{x}$ has optimal antipodal structure, then we basically recover the \tsout{Max-Cut} \fe{\mc}. In this case, \QAOAw\ with initial state $\ket{s_0}$ yields the \mc\ of $G$ for $p=0$.

For any connected bipartite graph and any $k$ (including $k=2,3$), one can show that any globally optimal solution $x$ of $\text{BM-MC}_k$ will have the antipodal structure described above. For the special case of even cycles, we find that the {\sc BM-MC}$_3$ optimization of \mc\ always finds the global optimum. These observations simply imply that random vertex-at-top rotations recover good solutions from the classical regime.

\begin{theorem}\label{thm:cycle}
	For a union of $r$ even-cycles, any local optimum $\mathbf{x}^*$ for {\sc BM-MC}$_3$ is a global optimum. 
\end{theorem}
The characterization of local optima above for even cycles simply implies that initializing \QAOAw\ with the random vertex-at-top rotation $Q(R_V(\mathbf{x}^*))$ will also recover the \tsout{Max-Cut} \fe{\mc}. To prove the structure of local optima in Theorem \ref{thm:cycle}, we exploit the structure of the graph and utilize KKT conditions to first show that any local optimum for BM-MC$_3$ has rank at most 2.\footnote{In non-linear optimization, the Karush–Kuhn–Tucker (KKT) conditions are first-order necessary conditions which characterize the set of optimal solutions. The usage of the KKT conditions generalizes the method of Lagrange multipliers \cite{Kar39}\cite{KT51}.} Further, we show that any rank-2 solution can be improved locally, thus resulting in rank-1 local (and therefore, global) optima. Details of the proof can be found in the Appendix \ref{sec:appendixproofs}.

It is conjectured that the performance of standard \QAOA\ for $n$-node even cycles is $({2p+1})/({2p+2})$ when $n > 2p$ \cite{FGG14,WHJR18}. The above theorem is a concrete example where \QAOAw\ outperforms standard \QAOA, due to a warm-start with a classical optimal solution. We find that warm-starts often result in flatter parameter landscapes for \QAOAw, e.g., see Figure \ref{fig:fourCycleLandscape} depicting the landscapes for various variants of \QAOAw\ on cycle $C_4$ on four vertices (i.e. $C_4 = (V,E)$ with $V=\{1,2,3,4\}$ and $E = \{(1,2),(2,3), (3,4), (1,4)\}$). For the vertex-at-top rotation in particular, notice that the solution quality monotonically decreases in $|\beta_1|$ with the optimal parameters all lying on the line $\beta_1 = 0$. We make this precise in the following observation. 

\begin{observation}
	\label{lem:landscape}
	Let $k\in\{2,3\}$ and $G = (V,E)$ be a graph with weights $w: E \to \mathbb{R}$ and $S \subseteq V$ be such that $(S,V \setminus S)$ is a \mc\ of $G$.  Let $y$ be a unit vector in $\mathbb{R}^k$. Let $\mathbf{x}$ be such that $x_u = y$ if $u \in S$ and $x_u = -y$ if $u \notin S$. If we initialize \QAOAw\ with the initial state $\ket{s_0} = Q(R_V(\mathbf{x}))$ where $R_V$ is a random vertex-at-top rotation, then, we recover the \mc\ since the states are aligned with $\ket{0}$ and $\ket{1}$. The expected cut value at $(\gamma_1,\beta_1)$ is given by
	$$F_1(\gamma_1,\beta_1) = \frac{1}{4}\Big((2M-W)\cos(4\beta)+2M+W)\Big),$$
	where $M = \text{\mc}(G)$ and $W = \sum_{e \in E} w_e$.
	Observe that $F_1(\gamma_1,0) = \text{\mc}(G)$ for all $\gamma_1 \in \mathbb{R}$ and $F_1(\gamma_1,\beta_1)$ decreases as $|\beta_1|$ increases for all $|\beta_1| \in [0,\pi/4]$.
\end{observation}

The form of the expression for $F_1(\gamma_1,\beta_1)$ follows from the fact that the cost term of the quantum circuit has no effect and that $\beta_1$ can be interpreted as a rotation angle (about the $x$-axis) in the Bloch sphere that moves the state away from the measurement axis. The details of this derivation are included in Appendix \ref{sec:appendixproofs}.

In contrast to the vertex-at-top rotations preserving the optimality of antipodal solutions (Observation \ref{lem:landscape}), this is not always the case for uniform rotations. In Figure \ref{fig:fourCycleLandscape} for example, the \emph{uniform} rotation does not yield the optimal cut for $C_4$ for any choice of parameters in rank-3. However, if we instead use the BM-MC$_2$ solution with a uniform rotation to obtain $R(\mathbf{x})$, then there \emph{does} exists a combination of parameter values that yields the optimal cut 
(by choosing $\gamma_1 = 0$ and an appropriate choice of $\beta_1$, application of the quantum circuit can be interpreted as a rotation that aligns $R(\mathbf{x})$ with the measurement axis in this case). This is due to potential proximity of uniformly rotated rank-3 solutions to the eigenstates of the mixer, which we can avoid in rank-2 initializations as discussed in \tsout{the next section} \fe{Section \ref{subsec:nonoptimalityQAOAw}}.

\section{Pre-processing Time vs Parameter Search Time}\label{sec:appendixPreProcess}

\begin{figure}
    \centering
    \includegraphics[scale=0.5]{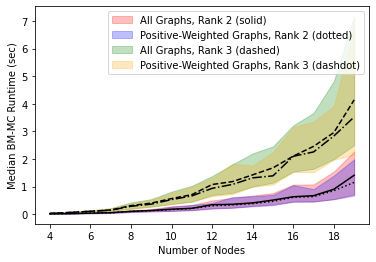}\includegraphics[scale=0.5]{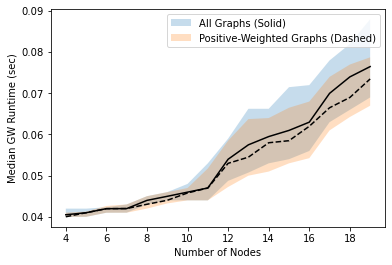}
    \caption{This figure shows how the median runtime changes for both GW and BM-MC$_k$ ($k=2,3$) as the number of nodes increases. The extended graph library $\mathcal{G}'$ (2076 instances) was used to generate the results above; we also run plot the results for just the positive-weighted graphs in $\mathcal{G}'$ as well. The top and bottom of the colored regions corresponding to 75 and 25 percentiles respectively.}
    \label{fig:classicalRuntimes}
\end{figure}

\begin{figure}
    \centering
    \includegraphics[scale=0.5]{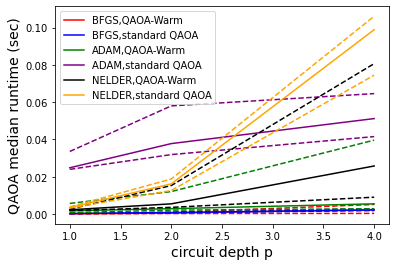}\includegraphics[scale=0.5]{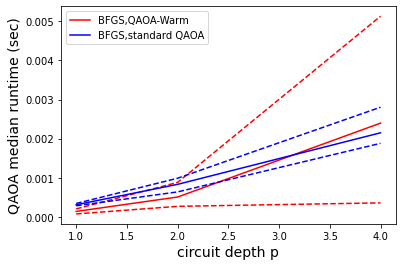}
    \caption{This figure shows how the median runtime changes for the optimization loop of both standard \QAOA\ and \QAOAw\ for various optimizers (ADAM, BFGS, and Nelder-Mead). COBYLA was not included due to technical limitations with our software; in particular, we were unable to gain direct access to the source code needed to in order to exclude the runtime of function or gradient evaluations. as the circuit depth increases. These runtimes do not include the time taken to evaluate/estimate the function values or gradients of the expected cut value $F_p(\gamma,\beta)$ (since in practice, such calls would be made on the quantum device). These plots were generated by randomly selecting 20 8-node graphs from our graph library $\mathcal{G}$ (see Section \ref{subsec:experimentalSetup}), with 10 of the 20 graphs having only positive edge weights. For each solid colored line (corresponding to the median), there are two dashed lines of the same color above and below representing the 75th and 25th percentiles respectively. On the right, we plot the runtimes for BFGS separately in order to more easily see the trend in runtime as $p$ increases.}
    \label{fig:QAOARuntime}
\end{figure}
Here, we compare runtimes for various aspects of \QAOAw\ to those of standard \QAOA\ and GW. For the preprocessing stage, finding an approximate solution for BM-MC$_k$ takes up the bulk of the time (i.e., 1-3 seconds). The rotation applied to the solution and the mapping of the rotated solution to the Bloch sphere is negligible. We plot the runtimes for BM-MC$_k$ for $k=2,3$ in Figure \ref{fig:classicalRuntimes}. To get a better idea of scaling, we consider an expanded graph library $\mathcal{G}'$ consisting of 2076 instances; $\mathcal{G}'$ is generated in the same way as $\mathcal{G}$ (see Section \ref{subsec:experimentalSetup}) but we instead consider graphs of up to 19 nodes. Finding approximate solutions to rank-2 BM-MC$_2$ is considerably faster than rank-3 BM-MC$_3$; furthermore, the runtimes are similar regardless of edge weight values. Plots of GW's runtime for all graphs in $\mathcal{G}'$ are included in Figure \ref{fig:classicalRuntimes}; as before, we see the runtimes are similar even if restrict our attention to only positive-weighted graphs. Note that our code for BM-MC runs is not optimized, and possibly faster implementations for this might be possible.

For both classical algorithms (GW and BM-MC$_k$), we see that the runtime increases superlinearly in the number of nodes $n$. In regards to theoretical results, the runtime of GW is dominated by solving the SDP; Lee and Padmanabhan \cite{LP19} develop an algorithm where one can get within factor $1-\varepsilon$ of the optimal SDP objective in $\widetilde{O}(m/\varepsilon^{3.5})$ time where $m$ is the number of edges in the graph. Similarly, for BM-MC$_k$, Mei et al. \cite{MMMO17} show that one can use a variant of the fast Riemannian trust-region algorithm to find a locally optimal solution in $O(n^2dk^4\log n)$ time for $d$-regular graphs. 

We now consider the runtime of the optimization loop used in both standard \QAOA\ and \QAOAw\ as seen in Figure \ref{fig:QAOARuntime} for various optimizers (ADAM, BFGS, Nelder-Mead). To get an idea of the runtime of the classical portions of the optimization loop, we exclude\footnote{We exclude such portions since including them would not be reflective of the runtime obtained on an actual quantum device; a quantum device can estimate $F_p(\gamma,\beta)$ (the expected cut value) in time polynomial in $n$ whereas a numerical simulation would (typically) take time that is exponential in $n$.} the time taken to estimate the function values or gradients of $F_p(\gamma,\beta)$. During our preliminary experiments, we found that the number of nodes did not have any noticable effect on the runtime of the optimization loop for either standard \QAOA\ or \QAOAw\  for any of the optimizers. However, for all optimizers, Figure \ref{fig:QAOARuntime} shows that more time is needed to optimize $\gamma$ and $\beta$ as the circuit depth increases. With the exception of BFGS, for all optimizers and circuit depths, it appears that \QAOAw\ converges to a set of parameters more quickly compared to standard QAOA.

We now discuss the runtime of the preprocessing stage of \QAOAw\ relative to the runtime of \QAOAw 's optimization loop. A direct comparison is difficult since the former is independent of the circuit depth $p$ and the latter is independent of the number of nodes $n$. However, for the $p$ and $n$ considered in our experiments, it appears (from Figures \ref{fig:classicalRuntimes} and \ref{fig:QAOARuntime}) that the preprocessing stage takes orders of magnitude longer. We remark that our current implementation for finding approximate BM-MC$_k$ solutions (Algorithm \ref{alg:getSDPRank2}) was not designed to find solutions quickly; we suspect other methods can find solutions more quickly. Additionally, we remark that the runtime preprocessing stage appears to scale modestly as the number of nodes increases. The trends in Figure \ref{fig:QAOARuntime} also suggest that as the circuit depth $p$ increases, that the proportion of \QAOAw\ spent in the preprocessing stage diminishes. Moreover, the real runtime of the optimization loop on an actual quantum device would be longer since one needs to consider the time needed to query the quantum device in order to estimate the value or gradient of $F_p(\gamma,\beta)$ at every iteration of the optimization loop. Lastly, there is the additional benefit that if one wants to perform multiple \QAOAw\ runs with different initializations of the variational parameters or different rotation schemes, then one only needs to find a solution to BM-MC$_k$ once.

\section{Comparison With Egger et al.}
    \label{sec:appendixEgger}
%\subsection{Similarities and Differences}
 
 Both our approach and that of Egger et al. \cite{egger2020warm} consider a variant of \QAOA\ initialized with a separable state obtained by some classical method. In this section, we describe the similarities and differences between the two approaches.

In their first approach, which Egger et al. refer to as ``Continuous warm-start QAOA", they consider QAOA applied to the Quadratic Unconstrained Binary Optimization (QUBO) problem which can be formulated as
    \begin{equation}
    \label{eqn:QUBO}
        \min_{x \in \{0,1\}^n} x^T\Sigma x,
    \end{equation}
    where $\Sigma$ is an $n \times n$ symmetric matrix. One can consider the relaxation,
    \begin{equation}
    \label{eqn:QUBO_Relaxed}
        \min_{x \in [0,1]^n} x^T\Sigma x,
    \end{equation}
    where now each $x_i$ lies in the interval $[0,1]$. If $\Sigma$ is positive semidefinite, then Equation \ref{eqn:QUBO_Relaxed} is a convex program which can easily be solved by classical optimizers to obtained an optimal solution $x^*$. Next, Egger et al. produces an unentangled state by mapping each $x_i^* \in [0,1]$ to a portion of a great-circle on the Bloch sphere; more specifically, the initial state $\ket{s_0}$ is given by
    \begin{equation}
    \label{eqn:startingState}
    \ket{s_0} = \bigotimes_{j=1}^n \left[ R_Y\left(2\arcsin\left(\sqrt{x^*_j}\right)\right) \ket{0}\right],\end{equation}
    where $R_Y(\theta)$ is a rotation on the Bloch sphere about the $y$-axis by angle $\theta$. This initialization significantly differs from our initialization scheme in that our relaxation relaxes each variable to a unit vector (instead of a position in an interval) and (for rank-3 initializations) we are not restricted to any particular portion of the Bloch sphere. Additionally, since \mc\ is equivalent to QUBO \cite{dunning2018}, our approach can also be used to solve arbitrary QUBO problems; however, this approach by Egger et al. is not applicable to \mc\ since one can show that the corresponding $\Sigma$ matrix in Equation \ref{eqn:QUBO_Relaxed} would not be positive semidefinite.\footnote{In fact, the matrix would be \emph{negative} definite in the case of \mc. It is straightforward to show that this implies that any locally optimal solution to the proposed relaxation of the QUBO (corresponding to the \mc\ instance) would yield only purely binary solutions in $\{0,1\}^n$.}
    
   Egger et al. also modify the mixing Hamiltonian $H_B$ in their approach; in particular, they choose\footnote{The authors do consider further modifications of the initial state by introducing an $\epsilon>0$ parameter that ensures the qubits are not initialized too close to the poles; this is done to avoid degenerate initializations that would cause Egger et al.'s variant of \QAOA\ to fail to converge as the circuit depth $p \to \infty$. The mixer is also adjusted accordingly.} $H_B = \bigoplus_{j=1}^n H_{B,j}$ where
    \begin{equation}
    \label{eq:eggerMixer}
        H_{B,j} = \begin{bmatrix}
    2x_j^*-1 & -2\sqrt{x_j^*(1-x_j^*)}\\
    -2\sqrt{x_j^*(1-x_j^*)} & 1-2x_j^*
    \end{bmatrix}.
    \end{equation}
    
    One can show that $\ket{s_0}$ from Equation \ref{eqn:startingState} is a ground state of $H_B$ as described above; Egger et al. remark that this allows us to apply the adiabatic theorem and conclude that this variant of QAOA approaches the optimal solution as the circuit depth tends to infinity (assuming an optimal choice of variational parameters). Unfortunately, our \QAOAw\ approach has no such guarantees.
    
    Egger et al. also consider another variant of \QAOA\ called ``Rounded warm-start \QAOA." Unlike their previous approach, this approach is more readily applicable to \mc{}. In this approach, a cut $(S,V \setminus S)$ is first generated via classical means (e.g. the rounding procedure found in the GW algorithm). Then, each qubit corresponding to a vertex in $S$ is initialized to $R_Y(\pi/3)\ket{0}$; similarly, each vertex in $V \setminus S$ is initialized to $R_Y(2\pi/3)\ket{0}$. The mixer used in this approach is the same as Equation \ref{eq:eggerMixer} but with the diagonal elements multiplied by $-1$. It can be shown that this approach allows one to recover the same cut that was initially used to create the initial quantum state. However, unlike their previous approach, the adiabatic theorem is not applicable in this case (since the initial state is no longer a ground state of the mixer), and thus, not much is known about the theoretical convergence of this rounded approach.

\paragraph{Experimental Comparison} We perform a similar set of numerical simulations that Egger et al. used for \mc\ (i.e. their rounded approach) and compare their approach with our own. When using their approach, for each instance considered, we also use a GW-solver to (classically) obtain 10 cuts, keeping only the best 5. Each of these cuts are used to create a different initial state using their Rounded Warm-Start \QAOA\ approach (with a regularization parameter of $\varepsilon = 0.25$ which guarantees that their approach can recover the cut used for initialization at a particular set of variational parameters at $p=1$). Similar to our numerical simulations for \QAOAw, we also test their approach at circuit depths $p=1,2,4,8$.

\begin{figure}
    \centering
    \includegraphics[scale=0.28]{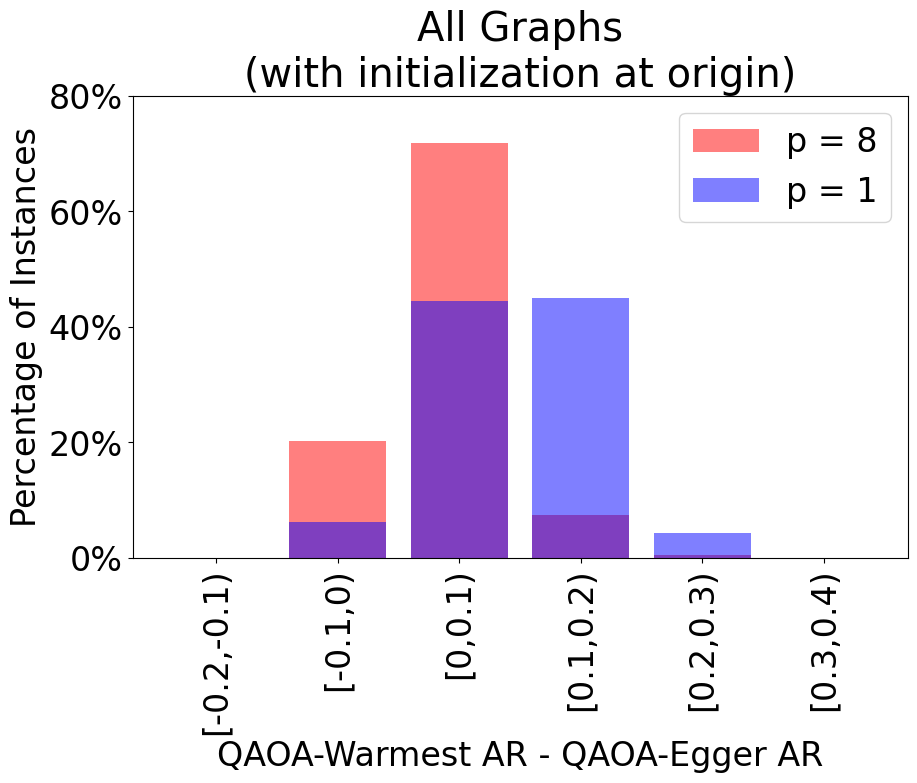}\includegraphics[scale=0.28]{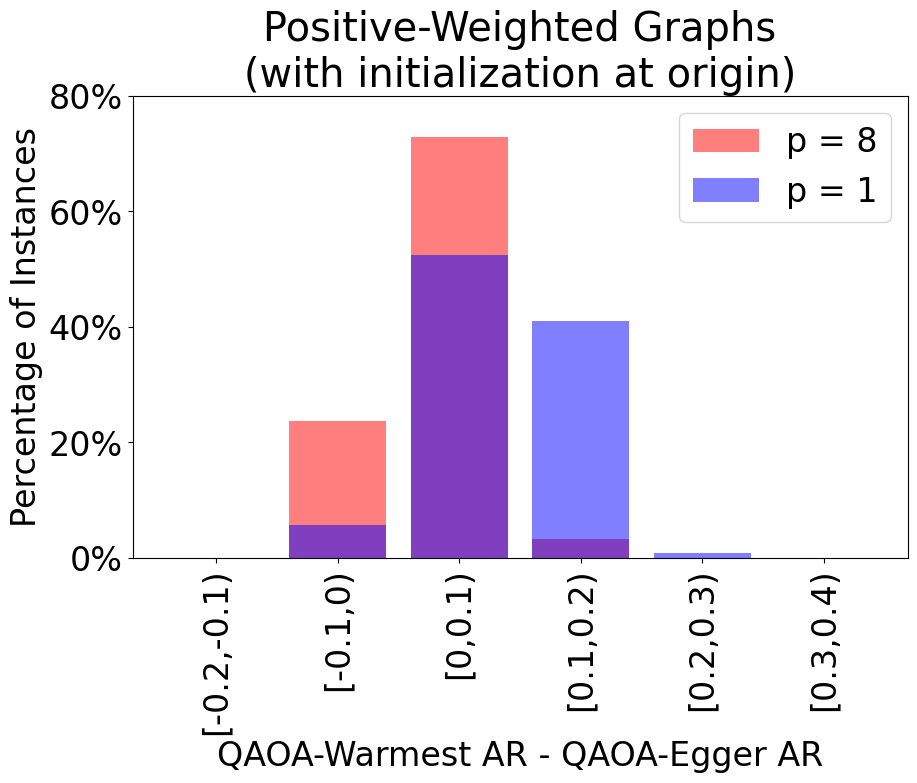}\\
    \includegraphics[scale=0.28]{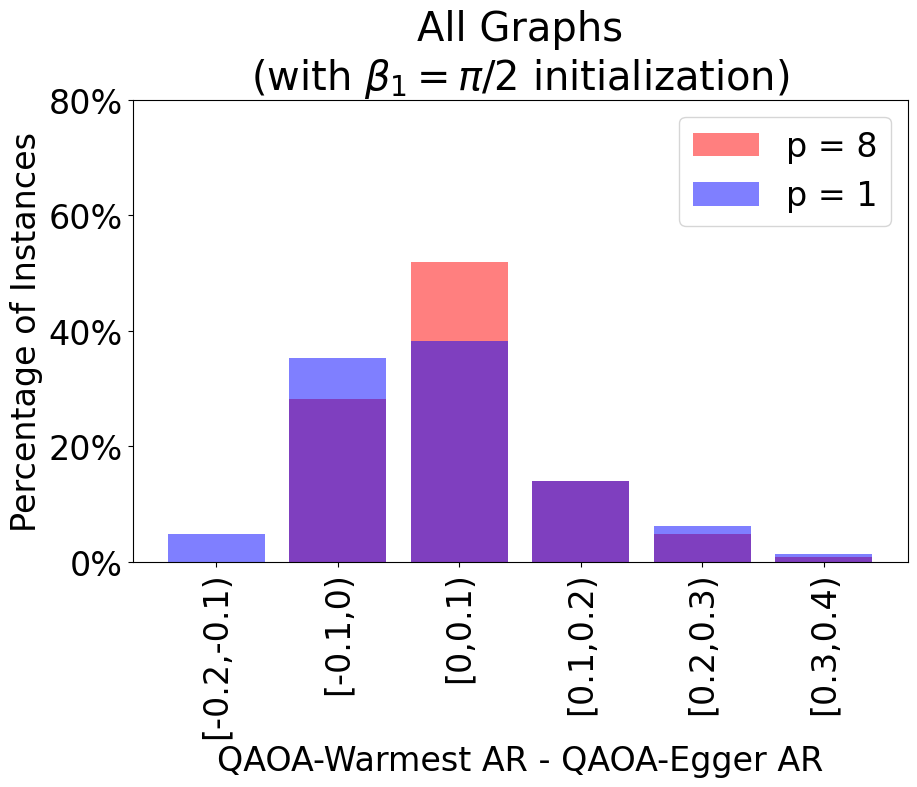}\includegraphics[scale=0.28]{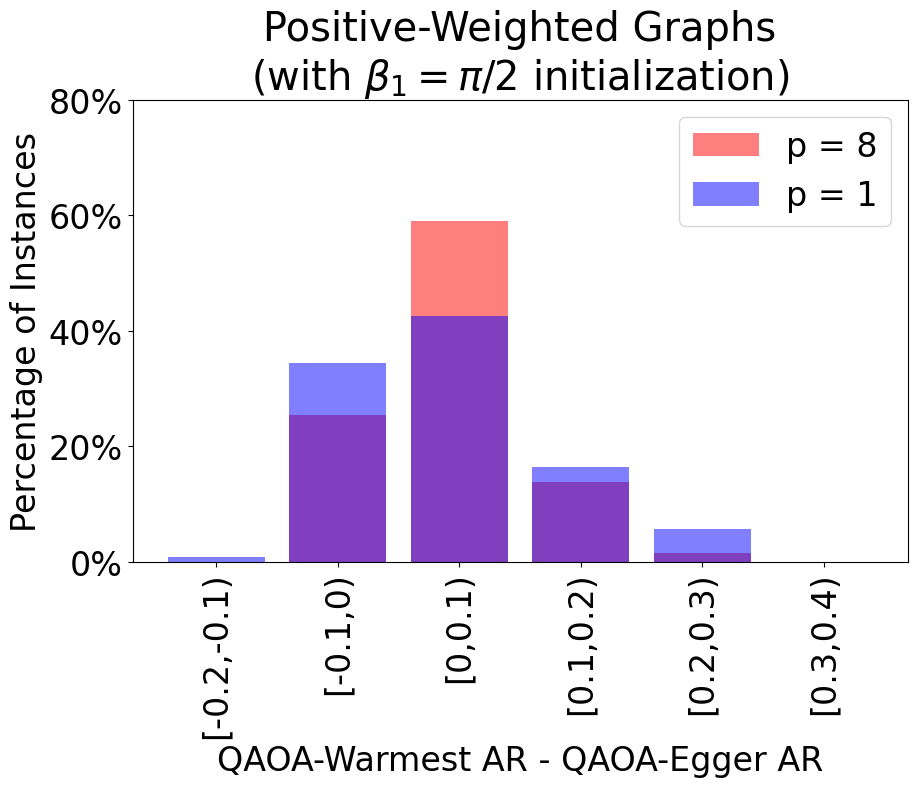}\\
    \caption{Histograms are used to compare the performance of \QAOAw{} and Egger et al.'s \cite{egger2020warm} Rounded Warm-Start QAOA approach in our graph library $\mathcal{G}$ (see Section \ref{subsec:experimentalSetup}) for circuit depths $p=1$ (blue) and $p=8$ (red); regions in purple are the overlaps of the two histograms. The left column uses all the graphs in $\mathcal{G}$ (see Section \ref{subsec:experimentalSetup}) whereas the right column only considers those with positive edge weights. The top row considers initializing the variational parameters to the origin for both approaches whereas the bottom row considers initializing $\beta_1 = \pi/2$ (with the remaining parameters at zero) for Egger et al.'s approach. }
    \label{fig:eggerResults}
\end{figure}

We consider two possible choices regarding the initialization of the variational parameters $\gamma$ and $\beta$: (1) initializing near the origin and (2) initializing at $\beta_1 = \pi/2, \gamma_1=0$ (with the remaining parameters being initialized to 0 for $p > 1$). The former initialization is the same as that used for \QAOAw\ and the latter is guaranteed to produce the cut used to initialize the quantum state.

 For the Egger et al. approach, we found that for most instances, when taking the best of 5 cuts produced by GW, that at least one of them would be the optimal cut. This yields uninteresting results since choosing $\beta_1 = \pi/2$ (with the remaining parameters being 0) would automatically yield the optimal cut when using Egger et al.'s approach. For this reason, we have removed all such instances (\fe{76.2\%} of the instance library) in our comparison.  The removal of such instances also removes all instances for which \QAOAw\ achieves optimality at $p=0$ (\fe{26.3\%} of the instance library).

As seen in Figure \ref{fig:eggerResults}, at low circuit depth ($p=1$), regardless of initializaiton scheme used for Egger et al.'s approach, \QAOAw\ is able outperform it for the majority of the instances. The advantage that \QAOAw\ has over Egger et al.'s approach is subdued at higher $p$ (e.g. we see a leftward shift in the histograms when going from $p=1$ to $p=8$) or when using the secondary initialization scheme for Egger et al.'s approach. \tsout{Lastly, magnitude in the difference in approximation is more spread out when considering all the graphs in our instance library $\mathcal{G}$ as opposed to when we only use the subset of $\mathcal{G}$ containing positive-weighted graphs.}

\end{document}